\documentclass[twoside,11pt]{article}

\usepackage{blindtext}

\usepackage[abbrvbib, preprint]{jmlr2e}

\usepackage{amsfonts, mathtools, nicefrac, amscd, latexsym, mathrsfs, dsfont, enumitem}
\mathtoolsset{showonlyrefs}
\usepackage{MnSymbol}
\usepackage{xcolor}
\usepackage{subcaption}

\usepackage{tabularx, fix-cm}
\usepackage{multirow, multicol}

\usepackage[algo2e, ruled]{algorithm2e}%

\newcounter{hyp}

\newenvironment{hyp}[1][]%
{%
  \refstepcounter{hyp}%
  \begin{itemize}[
    label=\bfseries (H\arabic{hyp})%
      \if\relax\detokenize{#1}\relax
    \else\ (#1)\fi,
    leftmargin=*
  ]%
  \item%
}%
{%
  \end{itemize}%
}

\newcommand{\hypref}[2][]{%
  \textbf{(H\ref{#2})}%
  \if\relax\detokenize{#1}\relax
  \else\textbf{(#1)}\fi
}

\usepackage{xspace}

\newcommand*{\bx}{\mathbf{x}}
\newcommand*{\f}{\mathbf{f}}
\newcommand*{\bK}{\mathbf{K}}
\newcommand*{\dd}{\mathrm{d}}
\newcommand*{\zero}{\mathbf{0}}
\newcommand*{\id}{\mathbf{I}}
\newcommand*{\bSigma}{\boldsymbol{\Sigma}}
\newcommand*{\bmean}{\boldsymbol{m}}

\newcommand*{\Xset}{\mathsf{X}}

\newcommand*{\Kregion}{\mathsf{R}}
\newcommand*{\dom}{\mathsf{C}}
\newcommand*{\Kcube}{\mathsf{K}}
\newcommand*{\dcube}{\mathsf{Q}}
\newcommand*{\OpenSet}{\mathsf{G}}
\newcommand*{\NodeSet}{\mathsf{Y}}
\newcommand*{\BoundarySet}{\mathsf{S}}
\newcommand*{\LocalCube}{\mathsf{U}}
\newcommand*{\BumpCube}{\mathsf{W}}
\newcommand*{\Subcube}{\mathsf{E}}
\newcommand*{\AuxCube}{\mathsf{V}}
\newcommand*{\JacobiInterval}{\mathsf{L}}
\newcommand*{\EuclideanBall}{\mathsf{B}}
\newcommand*{\ConeSet}{\mathsf{C}_{\mathrm{cone}}}
\newcommand*{\rad}{r}
\newcommand*{\radext}{\widetilde{r}}

\newcommand*{\RootFamily}{\mathcal{K}}
\newcommand*{\DyadicFamily}{\mathcal{D}}
\newcommand*{\StopPartition}{\mathcal{Q}}
\newcommand*{\IntegrandClass}{\mathcal{F}}

\newcommand*{\QuadClass}{\mathcal{A}}
\newcommand*{\SmoothClass}{\mathcal{C}}
\newcommand*{\QuadRule}{A}
\newcommand*{\Jcrit}{J}

\newcommand*{\Lp}[1][p]{L^{#1}(\Xset, \Pi)}
\newcommand*{\Ha}[1][\Xset]{\mathcal{H}^\alpha(#1)}
\newcommand*{\Hk}[1][\Xset]{\mathcal{H}_k(#1)}
\newcommand*{\iid}{\textit{i.i.d}\xspace}
\newcommand*{\ie}{\textit{i.e.}}
\newcommand*{\eg}{\textit{e.g.}}
\newcommand*{\var}{\mathbb{V}\mathrm{ar}}

\usepackage{lastpage}
\jmlrheading{}{2026}{1-\pageref{LastPage}}{}{}{}{Edoardo Bandoni, Christian P. Robert and Julien Stoehr}

\ShortHeadings{Rate-Optimal Randomised Kernel Quadrature}{Bandoni, Robert and Stoehr}
\firstpageno{1}

\begin{document}

\title{Rate-Optimal Randomised Kernel Quadrature}

\author{\name Edoardo Bandoni \email edoardo.bandoni@dauphine.psl.eu \\
       \addr CEREMADE\\
       Universit\'e Paris-Dauphine, Universit\'e PSL, CNRS\\
       75016 Paris, France
       \AND
       \name Christian P. Robert \email xian@ceremade.dauphine.fr \\
       \addr CEREMADE\\
       Universit\'e Paris-Dauphine, Universit\'e PSL, CNRS\\
       75016 Paris, France\\
       \addr Department of Statistics\\
       University of Warwick\\
       Coventry CV4 7AL, United Kingdom
       \AND
       \name Julien Stoehr \email stoehr@ceremade.dauphine.fr \\
       \addr CEREMADE\\
       Universit\'e Paris-Dauphine, Universit\'e PSL, CNRS\\
       75016 Paris, France}

\editor{}
\maketitle

\begin{abstract}
Kernel quadrature is used to approximate integrals of smooth functions, with the worst-case error typically decaying at the minimax rate $n^{-\alpha/d}$ for smoothness $\alpha$ in dimension $d$. Existing rate-optimal methods often depend on deterministic point sets, sometimes
tailored to a specific kernel. Here, we study randomised designs for kernel quadrature on compact and unbounded domains 
for kernels whose reproducing kernel Hilbert spaces are norm-equivalent to Sobolev spaces of integer order $\alpha>d/2$.
We bound the worst-case error of prescribed
designs adapted to a sampling density $p$ by $n^{-\alpha/d}$ 
times an explicit functional of the sampling and integration densities. 
Minimising this factor gives $p^\star\propto\pi^\gamma$ with $\gamma=2d/(2\alpha+d)$, where $\pi$ is the integration density. 
The proposal depends on the kernel only through its Sobolev order $\alpha$. 
Under local mass and cell-average conditions, together with a tail condition on $\mathbb{R}^d$, independent samples from $p^\star$ attain the same minimax rate with high probability, up to a logarithmic factor. 
The same bound also controls posterior contraction.
We verify these conditions for Beta-type densities on the cube and, for integrands in $\mathcal{H}^\alpha(\mathbb{R}^d)$, for Gaussian and Student-$t$ measures. 
The deterministic minimax exponent for Student targets is $\min(\alpha,\nu+d/2)/d$, with a logarithmic gap between the bounds at $\alpha=\nu+d/2$.
\end{abstract}

\begin{keywords}
  Probabilistic Integration, Bayesian Quadrature, Uncertainty Quantification, Contraction Rates, Minimax Rates
\end{keywords}

\section{Introduction}
\label{sec:intro}

Computing integrals is central to modern statistics, where expectations, marginal likelihoods, posterior summaries, and predictive quantities rarely admit closed-form solutions. Numerical integration therefore underpins both Bayesian and frequentist computation. In many contemporary applications, integrand evaluations are computationally costly, and approximation errors propagate through complex computational pipelines, making the choice of integration method a key factor of statistical accuracy.

Let $\Pi$ denote the integration measure, and $\pi$ its density
with respect to Lebesgue measure. The aim is to approximate the integral of a function $f$ with respect to $\Pi$ from $n$ pointwise evaluations of $f$.
Monte Carlo methods are widely used due to their simplicity and robustness. The vanilla Monte Carlo estimator achieves a root mean square error (RMSE) of order $\sigma(f)\,n^{-1/2}$, where $\sigma(f)^2$ is the variance of $f$ under $\Pi$ \citep{robert2004monte}. The exponent does not depend on the dimension, although both $\sigma(f)$ and the cost of simulating from $\Pi$ can increase with dimension. 
Moreover, in low and moderate dimensions, methods exploiting the regularity of $f$ can achieve faster convergence than the $1/\sqrt{n}$ rate.

Kernel quadrature provides a principled alternative based on function approximation. 
In its Bayesian formulation, referred to as Bayesian quadrature, a Gaussian process prior on the integrand yields a posterior distribution for the integral given a finite set of function evaluations \citep{briol2019probabilistic, hennig2022probabilistic, mahsereci2026bayesian}. 
The resulting estimator is a weighted sum of function evaluations, where the weights are induced by a positive-definite kernel and the integration measure.

Kernel quadrature may also be viewed as a probabilistic generalization of classical deterministic quadrature. For suitable kernels and design points, the posterior mean recovers standard numerical integration schemes such as the trapezoidal rule and Gauss--Hermite quadrature \citep{sarkka2016,karvonen2017}.
This perspective on Bayesian quadrature links posterior uncertainty to a functional-analytic interpretation of the integration error.
For Lebesgue integration on the unit cube, the classical deterministic
minimax rate for the worst-case error over a Sobolev ball of smoothness
$\alpha>d/2$ is $n^{-\alpha/d}$ \citep{novak1988deterministic}.
Under suitable designs, Bayesian quadrature achieves a worst-case error bounded in probability by
$n^{-\alpha/d+\varepsilon}$ for any $\varepsilon>0$ \citep{briol2019probabilistic}.
For sufficiently small $\varepsilon$ and low-dimensional settings, this bound decays substantially faster than the MC rate.

Despite these advantages, the performance of Bayesian quadrature depends critically on the sampling distribution of evaluation points.
Existing approaches include deterministic design based on posterior variance minimisation \citep{osborne2012active,briol2015frankwolfe,fisher2020locally} or space-filling criteria \citep{briol2019probabilistic,kanagawa2019,llorente2020adaptive,fisher2020locally}, 
while randomized constructions offer improved robustness and practical flexibility \citep{pmlr-v70-briol17a}. 
For a given kernel $k$, and evaluation budget $n$, the problem of optimal sampling can be stated as seeking a sampling distribution that minimizes the worst-case RMSE 
over a class of integrands $\IntegrandClass$
 \citep{pmlr-v70-briol17a}. 
In importance sampling, the variance-minimising proposal is proportional to $\lvert f\rvert\pi$ \citep{robert2004monte}, but it is of no practical use, as computing its normalising constant is an integration problem comparable to the original one. For kernel quadrature, no general tractable solution to the worst-case RMSE problem is known, and theoretical guidance remains limited.
From a practical perspective, \citet{pmlr-v70-briol17a} explored tempered proposals along a geometric path between a reference density and the integration density, using sequential Monte Carlo and an empirical bound on the RMSE to select the temperature. A complementary approach uses sampling distributions determined by the
Mercer spectrum of the kernel \citep{bach2017}.
In dimension one, a related question has been studied for weighted integration of processes with $r$ mean-square derivatives on a bounded interval, where deterministic regular sequences of designs generated by a density proportional to $\pi^{2/(2r+3)}$ are asymptotically optimal  for the mean-square integration error \citep{ritter1996asymptotic,ritter2000average}. Optimal point densities given by a power of the target also arise in quantization theory through Zador's theorem \citep{graf2000foundations}. We discuss these analogies after Theorem~\ref{thm:up-int}.

On bounded domains, controlling how well the evaluation points cover the domain provides a route to near-optimal rates for smooth integrands \citep{chen2025nested,briol2019probabilistic}. The design problem is particularly acute on unbounded domains. A finite design cannot provide uniform coverage of $\mathbb{R}^d$, and sampling from the integration measure can leave regions that contribute to the integral poorly explored. For designs obtained by optimising a kernel-based criterion, the placement of points can also be sensitive to the kernel length-scale \citep{pmlr-v70-briol17a}. For standard Gaussian integration, deterministic rules with rate guarantees have been developed for weighted smoothness spaces \citep{dick2018optimal,kazashi2023suboptimality}. 
However, many statistical applications involve Gaussian measures with a general covariance matrix, which require checking how the affine transformation affects the integrand class to extend their guarantees, and other unbounded integration measures, including heavy-tailed distributions. This motivates a common design strategy for Gaussian and heavy-tailed integration measures.

In this paper, we develop and analyse randomised sampling strategies for Bayesian quadrature on both the unit cube and $\mathbb{R}^d$, using kernels whose
reproducing kernel Hilbert spaces are norm-equivalent to unweighted
isotropic Sobolev spaces of integer order $\alpha>d/2$. 
For designs adapted to a sampling density $p$, we obtain
a worst-case error bound that decays as $n^{-\alpha/d}$ and includes
an explicit factor depending on $p$. Minimising this factor over sampling
densities gives the tempered proposal $p^\star\propto\pi^\gamma$, with
$\gamma=2d/(2\alpha+d)$, whenever $\pi^\gamma$ is integrable.
This gives a theoretical rule for choosing the exponent, complementing the empirical tempering approach of \citet{pmlr-v70-briol17a}. The proposal depends on the kernel only
through its Sobolev order and requires no spectral decomposition unlike the method from \citet{bach2017}.
Under the stated conditions, including a weak doubling condition,
and with suitable tail control on $\mathbb{R}^d$, we show that independent samples from $p$ give the same algebraic exponents with high probability, with an additional logarithmic factor in the upper bounds to ensure coverage of the cells. These bounds hold uniformly over integrands for each realised design. Our analysis also yields posterior contraction conditional on the design and controls the marginal posterior variance, which includes the variability of the posterior mean across designs.
We show that the deterministic minimax rate is achieved for Beta-type and Gaussian integration measures, while for the Student-$t$ measures with $\nu$
degrees of freedom, the exponent is
$\min(\alpha,\nu+d/2)/d$, with a logarithmic gap between the upper and lower
bounds at $\alpha=\nu+d/2$.

\paragraph*{Outline}
Section \ref{sec:notations} fixes notation and Section \ref{sec:background} reviews Bayesian quadrature and existing guarantees.
Section \ref{sec:main} presents our theoretical contributions.
Section \ref{sec:num} illustrates them numerically, and Section \ref{sec:conclusion} discusses extensions.
Proofs are given in Appendix \ref{app:proofs}, details on hyperparameter estimation in Appendix \ref{app:hyper}, and a further experiment in Appendix \ref{app:num}.

\section{Notations}
\label{sec:notations}

Throughout, $\mathbb{N}=\{1,2,\dots\}$ denotes the set of positive integers and $\mathbb{N}_0=\mathbb{N}\cup\{0\}$ the set of non-negative integers.
The zero vector in $\mathbb{R}^d$ and the $d\times d$ identity matrix are denoted by $\zero_d$ and $\id_d$, respectively.
For a probability measure $Q$ and $m\in\mathbb{N}$, $Q^{\otimes m}$ denotes its $m$-fold product measure.
Let $\Xset \subseteq \mathbb{R}^d$, $d\in\mathbb{N}$, $f:\Xset\rightarrow\mathbb{R}$ be a measurable function, and $\Pi$ be a Borel probability measure on $\Xset$. The integral of $f$ with respect to a measure $\Pi$ on $\Xset$ is denoted
\begin{equation*}
\Pi(f) = \int_{\Xset} f(x) \Pi(\,\dd x).
\end{equation*}
For a set $\mathsf{A}\subseteq\Xset$, $\mathds{1}_{\mathsf{A}}$ denotes its indicator function, equal to $1$ on $\mathsf{A}$ and $0$ elsewhere.
In function-space notation, the Lebesgue space $L^p(\Xset)$ refers to Lebesgue measure, whereas $L^p(\Xset,\Pi)$ refers to the integration measure $\Pi$, $1 \le p < \infty$. The $\Lp$-norm and $\Lp[\infty]$-norm of $f$ on $\Xset$ are defined respectively by
\begin{equation*}
\lVert f\rVert_{\Lp} = \left( 
\int_\Xset \lvert f(x)\rvert^p  \Pi(\,\dd x) 
\right)^{1/p},
\quad
\lVert f\rVert_{\Lp[\infty]} = \operatorname*{ess\,sup}_{x \in \Xset} \lvert f(x)\rvert.
\end{equation*}
For $\alpha \in \mathbb{N}_0$, let $\mathcal{H}^\alpha(\Xset)$ denote the Sobolev space of order $\alpha$ on $\Xset$, equipped with norm
\begin{equation*}
\lVert f\rVert_{\Ha} = \left( \sum_{\lVert \beta\rVert_1  \le \alpha} \int_\Xset \vert D^\beta f(x)\vert ^2  \,\dd x \right)^{1/2},
\end{equation*}
where $\beta = (\beta_1, \dots, \beta_d)$ is a multi-index, and $D^\beta f$ is the weak derivative of $f$.

Let $k:\Xset\times\Xset\to\mathbb{R}$ be a symmetric positive-definite kernel, and denote by $\mathcal{H}_k(\Xset)$ its reproducing kernel Hilbert space (RKHS), with inner product $\langle\cdot,\cdot\rangle_{\mathcal{H}_k}$, that is 
for all $f\in\Hk$, $x\in\Xset$, $f(x)=\langle f,k(x,\cdot)\rangle_{\mathcal{H}_k}$.
The norm induced by $\langle\cdot,\cdot\rangle_{\mathcal{H}_k}$ is denoted $\lVert \cdot \rVert_{\Hk}$.

Let $\NodeSet\subseteq\Xset$ be a nonempty finite set of design points and let $\dom\subseteq\Xset$ be a nonempty region. Denoting the Euclidean norm by $\lVert\cdot\rVert_2$, we define the fill distance of $\NodeSet$ on $\dom$ by
\begin{equation*}
\rho(\NodeSet,\dom)=\sup_{x\in\dom}\min_{y\in\NodeSet}\lVert x-y\rVert_2.
\end{equation*}
This quantity may equal $+\infty$.
It is the supremum of the radii of balls centred in $\dom$ that contain no design point in their interior. This quantity, also called the covering radius of the design, is the standard measure of how well a point set covers a domain in scattered data approximation \citep{Wendland_2004}. For a design given as a tuple $\bx_{1:n}=(x_1,\ldots,x_n)\in\Xset^n$, we write $\rho(\bx_{1:n},\dom)$ for $\rho(\{x_1,\ldots,x_n\},\dom)$.

For a sequence of random variables $(X_n)_{n \ge 1}$ and a deterministic sequence $(a_n)_{n \ge 1}$ with $a_n\neq0$, we write $X_n = \mathcal{O}_{\mathbb{P}}(a_n)$ if the sequence $(X_n/a_n)_{n\ge1}$ is uniformly tight, \ie, if for all $\varepsilon > 0$ there exists $B_{\varepsilon} > 0$ such that $\sup_{n \ge 1} \mathbb{P}[\lvert X_n \rvert > B_{\varepsilon} \lvert a_n \rvert] \le \varepsilon$.

\section{Bayesian Quadrature}
\label{sec:background}

\paragraph*{Problem formulation} Let $f\in \Lp[2]$ for which exact pointwise evaluation is possible, and consider approximating the integral $\Pi(f)$. 
In Bayesian quadrature, the integrand $f$ is modeled as a Gaussian process $f^\star \sim \mathcal{GP}(0, k)$, where $k : \Xset \times 
\Xset \to \mathbb{R}$ is a symmetric, positive-definite kernel \citep[\eg,][Chapter 10]{hennig2022probabilistic}. 
Given evaluations $\f_{1:n} = (f(x_1), \ldots, f(x_n))^\top$
at distinct design points $\bx_{1:n}=(x_1, \ldots, x_n) \in \Xset^n$, 
the posterior distribution of the integral value is
\begin{equation*}
\Pi(f^\star) \mid \bx_{1:n} \sim \mathcal{N}\left(\bmean_\Pi\left(\bx_{1:n}\right)^\top \bK^{-1}\f_{1:n},\; 
\Pi^{\otimes 2}(k) - \bmean_\Pi(\bx_{1:n})^\top \bK^{-1} \bmean_\Pi(\bx_{1:n})\right),
\end{equation*}
where $\bK$ is a symmetric positive definite matrix with $\bK_{ij} = k(x_i, x_j)$, and 
\begin{equation*}
    \bmean_\Pi(\bx_{1:n}) = (\Pi(k(\cdot, x_1)), \ldots, \Pi(k(\cdot,x_n)))^{\top}
\end{equation*} is the kernel mean vector. 
The posterior mean defines the 
weighted quadrature estimator
\begin{equation}
\label{eq:bq-estimator}
\widehat{\Pi}_n(f) 
= \mathbb{E}_{f^\star}\left[\Pi(f^\star)\mid \bx_{1:n}\right]
=\sum_{j=1}^n \omega_j f(x_j), 
\qquad \boldsymbol{\omega}_{1:n} = (\omega_1, \ldots, \omega_n)^{\top} = \bK^{-1}\bmean_\Pi(\bx_{1:n}).
\end{equation}
We refer to \citet{KarvonenKanagawaSarkka2019} for properties of the quadrature weights, including their positivity and magnitude.

The RKHS formulation \citep{berlinet2011reproducing} yields a variational characterisation of integration error.
When the kernel mean  embedding $\mu_\Pi \in \Hk$, it  satisfies 
for all $f \in \Hk$,
$\Pi(f) = \langle f, \mu_\Pi \rangle_{\Hk}$.
The estimator \eqref{eq:bq-estimator} then corresponds to replacing $\mu_\Pi$ by the empirical element $\mu_n = \sum_{j=1}^n \omega_j k(\cdot, x_j)$. 
For a fixed design, the quadrature weights minimise the worst-case integration error over the unit ball of $\Hk$, and the posterior variance coincides with this squared minimum \citep{KarvonenKanagawaSarkka2019}, namely
\begin{equation}
\label{eq:var-worst-case}
\var_{f^\star}[\Pi(f^\star) \mid \bx_{1:n}] 
= \lVert \mu_\Pi - \mu_n\rVert_{\Hk}^2 
= \sup_{\lVert f\rVert_{\Hk} \leq 1} \lvert \Pi(f) - \widehat{\Pi}_n(f)\rvert^2.
\end{equation}
This variational characterization has motivated extensive work on selecting design points to directly minimize $\lVert \mu_\Pi - \mu_n\rVert_{\Hk}$. 
Deterministic approaches include greedy variance minimisation via Frank--Wolfe methods
\citep{briol2015frankwolfe,wagstaff_batch_2018,pmlr-v151-hamid22a,adachi_fast_2022},
which enjoy convergence guarantees for the integration error, and active Bayesian quadrature
\citep{osborne2012active,gunter2014},
for which convergence guarantees for a broad class of adaptive methods were established by
\citet{KanHen19}. 
Randomized constructions include Mercer-based sampling
\citep{bach2017},
which provides strong theoretical guarantees but relies on sampling distributions defined through the Mercer decomposition of the kernel, which are difficult to evaluate and to simulate from in practice.
When the integration measure is accessible only through an \iid sample, \citet{chatalic2025efficient} use approximate ridge leverage scores to select a subset of the sampled points at which to evaluate $f$ and obtain convergence rates for integrands in Sobolev spaces on bounded domains.
\citet{belhadji2019kernel} propose determinantal point processes (DPPs) as randomised designs.
Although the repulsive behaviour of DPPs promotes good spatial coverage of $\Xset$ while controlling the fill distance, exact simulation relies on the spectral algorithm of \citet{hough2006determinantal}, which requires the eigendecomposition of the DPP kernel and becomes expensive as $n$ grows. We refer to \citet{gautier2019dppy} for a review of available samplers.
Finally, \citet{hayakawa2022} propose kernel quadrature rules enforcing positive integration weights via subsampling from a larger candidate set. The approach prevents amplification of errors in function evaluations while retaining theoretical guarantees on worst-case error bounds. However, their Mercer variants  require the kernel Mercer decomposition, which is generally unavailable, while their Nystr\"om variants use kernel evaluations and incur an additional term in the error bound. 
In this paper, we draw the evaluation points directly from a density determined by $\pi$ and the Sobolev order $\alpha$, and use the quadrature weights in \eqref{eq:bq-estimator}.

\paragraph*{Fill Distance and Convergence Rates} 
For Lebesgue integration on $[0,1]^d$, if $\Hk$ is norm-equivalent to $\Ha$ with $\alpha>d/2$, the deterministic minimax worst-case error over the unit ball of $\Hk$ is of order $n^{-\alpha/d}$ \citep{novak1988deterministic}. 
The influence of design point placement is captured by the fill distance, which quantifies the largest distance from any point in $\Xset$ to its nearest evaluation point.
Results from scattered data approximation \citep{Wendland_2004} show that
for compact $\Xset$ satisfying the interior cone condition, the kernel interpolant $s_n$ of any $f \in \Hk$ satisfies, for some constant $C$,
\begin{equation*}
\lVert f - s_n\rVert_{L^2(\Xset)} \leq C \rho(\bx_{1:n},\Xset)^\alpha \lVert f\rVert_{\Hk}.
\end{equation*}
Since $\widehat{\Pi}_n(f)=\Pi\left(s_n\right)$, when $\pi\in L^2(\Xset)$, the Cauchy--Schwarz inequality gives
\begin{equation*}
\lvert\Pi(f)-\widehat\Pi_n(f)\rvert=\left\lvert\int_{\Xset}\{f(x)-s_n(x)\}\pi(x)\,\dd x\right\rvert\le\lVert f-s_n\rVert_{L^2(\Xset)}\lVert\pi\rVert_{L^2(\Xset)}.
\end{equation*}
This directly implies that $\lVert \mu_\Pi - \mu_n\rVert_{\Hk} \leq C\lVert\pi\rVert_{L^2(\Xset)} \rho(\bx_{1:n},\Xset)^\alpha$. Hence, optimal control of the fill distance at order $n^{-1/d}$ yields the optimal rate for the worst-case integration error.

\paragraph*{Existing Guarantees on Bounded Domains}
The interplay between fill distance and sampling distribution for $\Xset=\left[0,1\right]^d$ was analysed by \citet{briol2019probabilistic}. If $\pi$ is bounded away from zero on $\Xset$ and the design points are generated by a reversible, uniformly ergodic Markov chain targeting $\Pi$, then $\rho\left(\bx_{1:n},\Xset\right)=\mathcal{O}_{\mathbb{P}}\left(n^{-1/d+\varepsilon}\right)$ for any $\varepsilon>0$. Under Sobolev norm equivalence, they give the worst-case error bound
\begin{equation*}
\left\lVert\mu_\Pi-\mu_n\right\rVert_{\Hk}=\mathcal{O}_{\mathbb{P}}\left(n^{-\alpha/d+\varepsilon}\right)
\end{equation*}
for arbitrarily small $\varepsilon>0$.
Thus, MCMC-based Bayesian quadrature targeting $\Pi$ achieves the minimax-optimal rate on the cube up to an arbitrarily small polynomial loss $n^\varepsilon$, and the posterior contracts super-exponentially in $n$, in probability over the design, towards the integral value.
The condition that $\pi$ is bounded away from zero excludes, for instance, Beta densities with shape parameters greater than one that vanish at the boundary.
Notably, the proof relies only on controlling the fill distance and does not require kernel-specific designs: since $\pi$ is bounded away from zero, points drawn from $\Pi$ cover the whole domain at the optimal scale $n^{-1/d}$, up to the loss $n^\varepsilon$. Throughout the paper rate optimality refers to the worst-case error over the unit ball of $\Hk$ for a realised design. For a fixed integrand, the error can decay much faster, and designs adapted to the particular $f$, such as the active and adaptive schemes discussed above, can exploit this. Designing for a specific $f$ is, however, a different problem from minimising the worst-case error: choosing the nodes adaptively from observed values of $f$ does not improve the worst-case rate of deterministic rules for integration over a Sobolev ball on the cube. Randomised error criteria are discussed in Section~\ref{sec:lower-bounds}.

This guarantee concerns the rate rather than the leading constant.
It can be extended to other compact domains either by establishing the required fill-distance and interpolation bounds directly or by using a smooth reparametrisation onto the cube.

When $\pi$ is highly non-uniform, sampling from $\Pi$ can still be inefficient among designs with the same convergence rate, because the accuracy at a given evaluation budget depends on how the points are distributed relative to $\pi$.
Separately, many statistical applications involve measures with unbounded support, such as Gaussian or heavy-tailed distributions on $\mathbb{R}^d$, where global fill-distance arguments do not directly apply. Indeed, the domain cannot be covered by finitely many balls of a fixed radius $h$, and sampling from the target may concentrate points in high-density regions while leaving large portions of the effective support underexplored. This motivates the local coverage analysis developed in Section~\ref{sec:main}. The sensitivity of sampling design was demonstrated empirically by \citet{pmlr-v70-briol17a}, who numerically illustrate substantial deterioration in RMSE under poorly chosen sampling distributions, independently of the asymptotic rate.

\section{Sampling Theory for Kernel Quadrature}
\label{sec:main}

\subsection{Limitations of Global Fill-Distance Control}
\label{sec:limit}

The guarantees recalled above bound the worst-case error through a single fill distance $\rho\left(\bx_{1:n},\Xset\right)$ on the whole domain.
While control of the fill distance gives a worst-case error bound of order $n^{-\alpha/d}$ on a compact domain, this bound measures coverage at the same resolution throughout $\Xset$.
However, the integration error is local. The contribution of a region $\dom\subset\Xset$ is controlled by the product of the supremum of the absolute interpolation error on $\dom$ and the mass $\Pi(\dom)$.
Regions where $\pi$ is large contribute more per unit volume and therefore require a finer mesh, whereas low-density regions can tolerate a coarser spacing. 
A uniform fill-distance bound does not capture how this spatial variation affects the accuracy of the rule, \ie, the resulting design can be rate-optimal with suboptimal leading constant.
The mismatch is more severe on unbounded domains, where the global fill distance is infinite for every finite design since $\Xset=\mathbb{R}^d$ cannot be covered by finitely many balls of any fixed radius. A natural attempt to recover the compact theory is to map the problem back to $\left[0,1\right]^d$ by a change of variables.

Let $F_\Pi^{-1}\colon\left(0,1\right)^d\to\Xset$ be a measurable map that sends the uniform distribution on $\left(0,1\right)^d$ to $\Pi$. In dimension one, we take the usual quantile function. Then
\begin{equation*}
\int_{\Xset} f(x)\,\Pi(\,\dd x)=\int_{\left(0,1\right)^d}f\left(F_\Pi^{-1}(u)\right)\,\dd u.
\end{equation*}
If $f\circ F_\Pi^{-1}$ belongs to a RKHS that is norm-equivalent to a Sobolev space of integer order $\alpha>d/2$, then Theorem~1 of~\citet{briol2019probabilistic} yields rate $n^{-\alpha/d+\varepsilon}$, for arbitrarily small $\varepsilon>0$, in probability on the unit ball of that RKHS.
However, this regularity assumption typically fails on unbounded domains, even in simple settings such as Gaussian integration.

\paragraph*{Example}
Let $\Pi$ be the standard Gaussian measure on $\mathbb{R}$, with CDF $\Phi$ and density $\phi$, and consider $f(x)=e^{-x^2/4}$, which belongs to $\mathcal{H}^\alpha(\mathbb{R})$ for every $\alpha$. By the change of variables $u=\Phi(x)$,
\begin{equation*}
\int_0^1\left\lvert\left(f\circ\Phi^{-1}\right)'(u)\right\rvert^2\,\dd u=\int_{-\infty}^{\infty}\frac{f'(x)^2}{\phi(x)}\,\dd x=\frac{\sqrt{2\pi}}{4}\int_{-\infty}^{\infty}x^2\,\dd x=\infty.
\end{equation*}
Hence $f\circ\Phi^{-1}\notin\mathcal{H}^\alpha\left(\left[0,1\right]\right)$ for every $\alpha\geq1$. Although $f$ is infinitely smooth and belongs to the class considered in this paper, the quantile transformation destroys the Sobolev regularity needed to apply the bounded-domain theory.

Since a single fill distance does not adapt to $\pi$, and a direct reduction to the compact setting may fail to preserve Sobolev regularity, in what follows we work directly on $\Xset$. We control the interpolation error locally, on cells whose diameters reflect the sampling density, and recombine these bounds against the local mass of $\Pi$.
This approach gives worst-case error bounds on both compact and unbounded domains, and identifies sampling densities adapted to $\pi$ . We also provide lower bounds that determine when these rates are optimal.

\subsection{Main idea and overview}
\label{sec:heuristic}

The idea of the main result is as follows.
We partition a bounded region of $\Xset$ into cells of comparable mass under a sampling density $p$. Heuristically, near a point $x$ where $p$ varies slowly on the scale of a cell, a cell of mass of order $M^{-1}$ has diameter of order $\left\{M p(x)\right\}^{-1/d}$.
A local Sobolev sampling inequality bounds the interpolation error on each cell. Combining these bounds using the cell masses under $\Pi$ leads to a discrete criterion based on cell averages, whose continuous counterpart is
\begin{equation}
\label{eq:J-continuous}
\Jcrit(p)=\int_{\Xset}p(x)^{-2\alpha/d}\pi(x)^2\,\dd x.
\end{equation}
By placing one node in each required subcell, we construct a deterministic design with $M$ proportional to $n$. Under the assumptions below, this gives a worst-case error bound of order $n^{-\alpha/d}\Jcrit(p)^{1/2}$, after accounting for the contribution outside the covered region on $\mathbb{R}^d$. The explicit factor $\Jcrit(p)$ is minimised by $p^\star\propto\pi^\gamma$, with $\gamma=2d/\left(2\alpha+d\right)$, whenever $\pi^\gamma$ is integrable.
Whether $p^\star$ also minimises the leading constant of the worst-case error is discussed after Theorem~\ref{thm:up-int}.

Section~\ref{sec:optimal} makes this precise and states the assumptions.
Section~\ref{sec:upper-posterior} gives the upper bounds for deterministic and random designs, then derives the consequences for posterior contraction and marginal variance. Section~\ref{sec:examples} applies these results to Beta-type, Gaussian and Student targets. Section~\ref{sec:lower-bounds} gives the corresponding lower bounds and identifies the optimal algebraic exponents.
Any density with $\Jcrit(p)<\infty$ that satisfies the geometric assumptions gives the rate $n^{-\alpha/d}$.

\subsection{Stopping partitions and the error functional}
\label{sec:optimal}

Let $\Pi$ have Lebesgue density $\pi$, and let $p$ be a probability density on $\Xset$ used to construct the design, with design measure $P=p\,\dd x$. For a cube $\dcube$ and a locally integrable function $h$, write 
\begin{equation*}
\left\langle h\right\rangle_\dcube=\frac{1}{\operatorname{vol}(\dcube)}\int_\dcube h(x)\,\dd x.
\end{equation*}
The dependence of the bound on $p$ is described by the functional $\Jcrit(p)$ in \eqref{eq:J-continuous}, with the convention that the integrand is $+\infty$ where $\pi>0$ and $p=0$, and $0$ where $\pi=0$. We also set for any $\alpha \geq 0$ and $x\in\Xset$
\begin{equation}
\label{eq:optimal-density}
p^\star(x)=\frac{\pi(x)^\gamma}{Z_\gamma},
\qquad Z_\gamma=\int_{\Xset}\pi(x)^\gamma\,\dd x,\qquad 
\gamma=\frac{2d}{2\alpha+d}.
\end{equation}
Whenever $Z_\gamma<\infty$, $p^\star$ is the probability density that minimises $\Jcrit(p)$. We denote the corresponding probability measure by $P^\star$, with $P^\star(\dd x)=p^\star(x)\,\dd x$. On $\left[0,1\right]^d$, if $\alpha\geq d/2$, then $0<\gamma\leq1$ and Jensen's inequality gives $Z_\gamma\leq1$.

\paragraph*{Dyadic forests}
Throughout, a cube is a set of the form $a+\left[0,\ell\right]^d$, with $a\in\mathbb{R}^d$ and $\ell>0$, that is, an axis-parallel cube. A cube $\Kcube=a+\left[0,\ell\right]^d$ generates the dyadic tree $\DyadicFamily(\Kcube)$ of cubes obtained by recursively bisecting every side of $\Kcube$. It is the image of the tree of standard dyadic subcubes of $\left[0,1\right]^d$ under $u\mapsto a+\ell u$ \citep{Stein1970}.
Every $\dcube\in\DyadicFamily(\Kcube)\setminus\left\{\Kcube\right\}$ has a unique parent $\widehat{\dcube}\in\DyadicFamily(\Kcube)$, and two cubes of $\DyadicFamily(\Kcube)$ are either nested or have disjoint interiors.
The subdivision may start from several cubes at once. 

A root family is a finite family $\RootFamily$ of cubes, in the above sense, with pairwise disjoint interiors, whose elements we call roots. Its forest $\DyadicFamily(\RootFamily)$ is the union of the dyadic trees of the roots, and its union $\Kregion$ is the region covered by the roots, namely
\begin{equation*}
\DyadicFamily(\RootFamily) =\bigcup_{\Kcube\in\RootFamily}\DyadicFamily(\Kcube),
\qquad
\Kregion = \bigcup_{\Kcube\in\RootFamily}\Kcube
\end{equation*}
On the unit cube $\Xset=\left[0,1\right]^d$, the domain itself is the only root, so $\RootFamily=\left\{\Xset\right\}$. On $\mathbb{R}^d$, the root family depends on $M$, and the region it covers grows with $M$. The constructions in Section~\ref{sec:examples} use either a single expanding cube $\left[-\rad_M,\rad_M\right]^d$ or a growing family of unit cubes.

\paragraph*{Stopping partitions}
For a root family $\RootFamily$ with union $\Kregion\subseteq\Xset$ and an integer $M\geq2$, set $\tau_M=P(\Kregion)/M$. The stopping partition $\StopPartition_M=\StopPartition_M\left(\RootFamily,P\right)$ is the set of maximal cubes $\dcube\in\DyadicFamily(\RootFamily)$ with $P(\dcube)\leq\tau_M$ \citep[Sec.~5.1]{lerner2019}.
By construction, each cell has $P$-mass at most $\tau_M$.
The associated discrete version of \eqref{eq:J-continuous} is
\begin{equation}
\label{eq:J-discrete}
\Jcrit_M(p) =\sum_{\dcube\in\StopPartition_M}\operatorname{vol}(\dcube)\left\langle p\right\rangle_\dcube^{-2\alpha/d}\left\langle\pi\right\rangle_\dcube^2.
\end{equation}

\begin{definition}
\label{def:doubling}
Let $D\geq1$ and $M\geq 2$. The measure $P$ is $D$-doubling at level $M$ on $\RootFamily$ if for every $\dcube\in\DyadicFamily(\RootFamily)\setminus\RootFamily$ with $P(\dcube)\leq\tau_M$
\begin{equation}
\label{eq:doubling}
P\left(\widehat{\dcube}\right)\leq D\,P(\dcube)\qquad
\end{equation}
\end{definition}
Ordinary dyadic doubling requires \eqref{eq:doubling} for every non-root cube. Here only cubes below the stopping mass are compared.

\paragraph*{Assumptions}
\begin{hyp}
\label{hyp:k}
The kernel $k$ satisfies $\sup_{x\in\Xset} k\left(x,x\right)=C_k<\infty$.
\end{hyp}
\begin{hyp}
\label{hyp:sobolev}
$\Hk$ is norm-equivalent to the Sobolev space $\Ha$ with $\alpha>d/2$, $\alpha\in\mathbb{N}_0$.
\end{hyp}
\begin{hyp}
\label{hyp:admissible-design}
There exist $M_0\geq2$, $D\geq1$ and a sequence of root families $\left(\RootFamily_M\right)_{M\geq M_0}$. For each $M\geq M_0$, let $\Kregion_M$ be the union of the roots in $\RootFamily_M$ and set $\tau_M=P\left(\Kregion_M\right)/M$. For every $M\geq M_0$, \emph{(i)} $P\left(\Kregion_M\right)\geq1/2$, \emph{(ii)} every root in $\RootFamily_M$ has $P$-mass greater than $\tau_M$, \emph{(iii)} $P$ is $D$-doubling at level $M$ on $\RootFamily_M$.
Moreover, either
\begin{enumerate}
\item[\textbf{(B)}] $\Xset=\left[0,1\right]^d$ and $\RootFamily_M=\left\{\Xset\right\}$ for every $M$; or
\item[\textbf{(U)}] $\Xset=\mathbb{R}^d$, and the sets $\Kregion_M$ are non-decreasing with union $\mathbb{R}^d$.
\end{enumerate}
\end{hyp}
\begin{hyp}
\label{hyp:tail}
In case \textbf{(U)}, 
\begin{equation*}
T=\sup_{M\geq M_0}M^{2\alpha/d}\int_{\Kregion_M^c}\pi(x)^2\,\dd x<\infty.
\end{equation*}
\end{hyp}
\begin{hyp}
\label{hyp:cell-average}
$\Jcrit(p)<\infty$. As $M\to\infty$, $\Jcrit_M(p)$, as defined in \eqref{eq:J-discrete} with the stopping partition $\StopPartition_M\left(\RootFamily_M,P\right)$, converges to $\Jcrit(p)$.
\end{hyp}

Condition \hypref{hyp:k}, as in \citet{briol2019probabilistic}, makes functions in $\Hk$ bounded, and \hypref{hyp:sobolev} fixes the local approximation order. Condition \hypref{hyp:admissible-design} is the only geometric requirement on $p$. For a single root, as in case \textbf{(B)}, the condition on the roots is automatic since, when $M\geq 2$ and $P(\Kregion_M) >0$, $P(\Kregion_M)>P(\Kregion_M)/M=\tau_M$.
Under \hypref{hyp:admissible-design}, every root has $P$-mass greater than $\tau_M$. Thus, along each branch, the first cube with $P$-mass at most $\tau_M$ is reached after at least one subdivision, and its parent still has $P$-mass greater than $\tau_M$. The doubling inequality therefore gives $\tau_M/D<P(\dcube)\leq\tau_M$ for every $\dcube\in\StopPartition_M$. The cells have comparable sampling masses and are smaller where the average value of $p$ is higher, up to the factor $D$. The lower bound controls the number of cells in the deterministic construction. For independent sampling, iterating the doubling inequality also gives a lower bound on the probability of each subcell required for coverage.
Condition \hypref{hyp:tail} controls the contribution outside $\Kregion_M$ through the $L^2(\Xset)$ tail of $\pi$, and \hypref{hyp:cell-average} states that the discrete functional is consistent.

These conditions differ in severity. Conditions \hypref{hyp:k} and \hypref{hyp:sobolev} hold, for instance, for Matérn kernels of suitable smoothness (Section~\ref{sec:num}), but the finite-order Sobolev equivalence in \hypref{hyp:sobolev} excludes infinitely smooth kernels such as the Gaussian kernel.
Condition \hypref{hyp:admissible-design} is the main geometric constraint on the sampling measure. It can be verified for the designs in Section~\ref{sec:examples}, but it can fail when the density vanishes faster than any polynomial, as shown at the end of that section.
Condition \hypref{hyp:tail} compares the tail of $\pi$ with the growth of $\Kregion_M$. It is what limits the rate for heavy-tailed targets (Proposition~\ref{prop:student-target}(b)).
For $p^\star$, Lemma~\ref{lem:J-limit} verifies \hypref{hyp:cell-average} under positivity and a uniform reverse H\"older inequality on the stopping cells. Alternatively, a uniform bound on $\Jcrit_M(p)$ would suffice for the same algebraic rate.
If $\Jcrit(p)=\infty$, which happens for every $p$ when $Z_\gamma=\infty$, slower rates can still be obtained with a truncated density, as in Proposition~\ref{prop:student-target}(b).

\subsection{Error bounds and posterior guarantees}
\label{sec:upper-posterior}

\subsubsection{Upper bound}
\label{sec:bounded}

We first state the bound for deterministic designs built from the stopping partition of $P$. Random designs are treated in Section~\ref{sec:random}.

\begin{theorem}
\label{thm:up-int}
Assume \hypref{hyp:k}, \hypref{hyp:sobolev}, \hypref{hyp:admissible-design} and \hypref{hyp:cell-average}, and, in case \textbf{(U)}, \hypref{hyp:tail}. There exist $C_0>0$ and $n_0\in\mathbb{N}$ such that, for every $n\geq n_0$, there exists a design $\bx_{1:n}$ with
\begin{equation}
\label{eq:thmup-int}
\sup_{\left\lVert f\right\rVert_{\Hk}\leq1}\left\lvert\Pi(f)-\widehat{\Pi}_n(f)\right\rvert\leq C_0\,n^{-\alpha/d}\Jcrit(p)^{1/2}.
\end{equation}
Moreover, with the notation of \eqref{eq:optimal-density}, if $Z_\gamma<\infty$, then for every probability density $p$
\begin{equation*}
\Jcrit(p)\geq Z_\gamma^{2/\gamma}
\end{equation*}
with equality if, and only if, $p=p^\star$ almost everywhere. Otherwise, $\Jcrit(p)=\infty$ for every $p$.
\end{theorem}

\begin{remark}
Both $D$ and $n_0$ may depend on $p$. The constant $C_0$ depends only on $d$, $\alpha$, $D$, the norm-equivalence constants and, in case \textbf{(U)}, also on the ratio $T/Z_\gamma^{2/\gamma}$.
\end{remark}

The construction places a fixed number of nodes in each cell of $\StopPartition_{M_n}$, with $M_n$ proportional to $n$ as specified in Appendix~\ref{app:thm}, and the bound follows from a local sampling inequality on each cell (Appendix~\ref{app:cell}).
Before absorbing the tail contribution, the proof bounds the worst-case error by
\begin{equation*}
c_{\mathrm{loc}}M_n^{-\alpha/d}\Jcrit_{M_n}(p)^{1/2}+C_{\mathrm{eq}}\left(\int_{\Kregion_{M_n}^c}\pi(x)^2\,\dd x\right)^{1/2},
\end{equation*}
where $c_{\mathrm{loc}}$ and $C_{\mathrm{eq}}$ are the constants from Appendix~\ref{app:cell}. The second term is absent in case \textbf{(B)}. In case \textbf{(U)}, \hypref{hyp:tail} bounds it by $C_{\mathrm{eq}}T^{1/2}M_n^{-\alpha/d}$, which is absorbed into the constant in \eqref{eq:thmup-int} using $\Jcrit(p)\geq Z_\gamma^{2/\gamma}$.

Theorem~\ref{thm:up-int} identifies $p^\star$ as the minimiser of the factor $\Jcrit(p)$, not of the worst-case error itself. The constant $C_0$ depends on $p$ through the doubling constant $D$ and, in case \textbf{(U)}, may also depend on $p$ through the choice of roots entering $T$, and we do not derive a lower bound of the form $c\,n^{-\alpha/d}\Jcrit(p)^{1/2}$ for designs distributed according to $p$. Nevertheless, the exponent $\gamma$ suggests a possible connection with two classical settings. In dimension one, asymptotically optimal design densities for weighted integration of processes with $r$ mean-square derivatives are proportional to $\pi^{2/\left(2r+3\right)}$ \citep{ritter1996asymptotic}, which has the same form as $\pi^\gamma$ with $\alpha=r+1$. In quantisation theory, Zador's theorem gives optimal point densities proportional to a power of the target \citep{graf2000foundations}. Applied formally to the measure $\pi^2\,\dd x$ with order $2\alpha$, the weighting appearing in $\Jcrit(p)$, it yields $(\pi^2)^{d/\left(d+2\alpha\right)}=\pi^\gamma$. These analogies suggest that $p^\star$ might also be optimal from an asymptotic point of view, in the sense of minimising the leading constant of the worst-case error. We leave this question open.

\subsubsection{Random designs}
\label{sec:random}

We now show that independent samples from $p$ attain the same rate, up to a logarithmic factor, with high probability.
The logarithmic factor arises because all subcells must receive a sample point. Taking $M$ of order $n/\log n$ ensures this coverage with high probability.
 
\begin{corollary}
\label{cor:up-int}
Assume \hypref{hyp:k}, \hypref{hyp:sobolev}, \hypref{hyp:admissible-design} and \hypref{hyp:cell-average}, and, in case \textbf{(U)}, \hypref{hyp:tail}. Let $\bx_{1:n}$ consist of $n$ independent samples from $p$. For every $c>0$, there exist $C_1>0$ and $n_c\in\mathbb{N}$ such that, for every $n\geq n_c$, with probability at least $1-n^{-c}$ over the design $\bx_{1:n}$,
\begin{equation}
\label{eq:randomized-rate}
\sup_{\left\lVert f\right\rVert_{\Hk}\leq1}\left\lvert\Pi(f)-\widehat{\Pi}_n(f)\right\rvert\leq C_1\left(\frac{n}{\log n}\right)^{-\alpha/d}\Jcrit(p)^{1/2}.
\end{equation}
\end{corollary}

\begin{remark}
The constant $C_1$ has the same dependencies as in Theorem~\ref{thm:up-int},
with an additional dependence on $c$.
\end{remark}

\subsubsection{Posterior contraction}
\label{sec:posterior-contraction}

Conditionally on the design, the posterior of $\Pi\left(f^\star\right)$ is Gaussian with standard deviation equal to the worst-case error \eqref{eq:var-worst-case}.
Corollary~\ref{cor:up-int} thus bounds both the posterior variance and the error of the posterior mean with high probability over the design. A Gaussian tail bound then gives posterior contraction around $\Pi(f)$.
 
\begin{corollary}
\label{cor:postvar-main}
Assume \hypref{hyp:k}, \hypref{hyp:sobolev}, \hypref{hyp:admissible-design} and \hypref{hyp:cell-average}, and, in case \textbf{(U)}, \hypref{hyp:tail}.
Let $f\in\Hk$ and let $\bx_{1:n}$ consist of $n$ independent samples from $p$.
For every $c>0$, there exists $C_2>0$ such that, for every $\delta>0$, there is a threshold $n_{c,\delta}\in\mathbb{N}$ for which the bound
\begin{equation}
\label{eq:posterior-contraction-main}
\mathbb{P}\left[\left\lvert\Pi\left(f^\star\right)-\Pi(f)\right\rvert\geq\delta\,\middle|\,\bx_{1:n}\right]\leq 2\exp\left\{-\frac{C_2\,\delta^2}{\Jcrit(p)}\left(\frac{n}{\log n}\right)^{2\alpha/d}\right\}
\end{equation}
holds for every $n\geq n_{c,\delta}$ with probability at least $1-n^{-c}$ over the design $\bx_{1:n}$.
\end{corollary}

\begin{remark}
The constant $C_2$ has the same dependencies as in Corollary~\ref{cor:up-int}. The threshold $n_{c,\delta}$ also depends on $\left\lVert f\right\rVert_{\Hk}$ and $\Jcrit(p)$.
\end{remark}

\subsubsection{Marginal posterior variance}
\label{sec:var-decomposition}

Corollary~\ref{cor:postvar-main} controls the posterior for each design on a high-probability event. To account also for the randomness of the design, fix $f\in\Hk$ and let $Y_n$ be a draw from the Bayesian quadrature posterior obtained by first sampling $\bx_{1:n}$ and then observing $\f_{1:n}$. Conditionally on $\bx_{1:n}$,
\begin{equation*}
Y_n\mid\bx_{1:n}\sim\mathcal{N}\left(\widehat{\Pi}_n(f),V_n\right),\qquad V_n=\var_{f^\star}\left[\Pi\left(f^\star\right)\mid\bx_{1:n}\right],
\end{equation*}
and the law of total variance gives
\begin{equation}
\label{eq:var-decomposition}
\var\left[Y_n\right]=\mathbb{E}_{\bx_{1:n}}\left[V_n\right]+\var_{\bx_{1:n}}\left[\widehat{\Pi}_n(f)\right].
\end{equation}
The second term is the variability of the posterior mean across designs.

\begin{proposition}
\label{prop:variance-control}
Assume \hypref{hyp:k}. For any random design for which the quadrature rule is well defined almost surely and every $f\in\Hk$,
\begin{equation}
\label{eq:marginal-variance-control}
\var\left[Y_n\right]\leq\left(1+\left\lVert f\right\rVert_{\Hk}^2\right)\mathbb{E}_{\bx_{1:n}}\left[V_n\right].
\end{equation}
\end{proposition}

The proof of Proposition~\ref{prop:variance-control} also yields the mean squared error bound
\begin{equation*}
\mathbb{E}_{\bx_{1:n}}\left[\left\lvert\widehat{\Pi}_n(f)-\Pi(f)\right\rvert^2\right]\leq\left\lVert f\right\rVert_{\Hk}^2\mathbb{E}_{\bx_{1:n}}\left[V_n\right]. 
\end{equation*}
To pass from the high-probability bound to the expectation of $V_n$, the proof uses $V_n\leq C_k$ on designs that fail the coverage condition and chooses the failure exponent large enough to make their contribution negligible.

\begin{corollary}
\label{cor:marginal-postvar}
Assume the hypotheses of Corollary~\ref{cor:up-int}. There exist $C_3>0$ and $n_3\in\mathbb{N}$ such that, for every $f\in\Hk$ and every $n\geq n_3$,
\begin{equation}
\label{eq:marginal-postvar-rate}
\var\left[Y_n\right]\leq C_3\left(1+\left\lVert f\right\rVert_{\Hk}^2\right)\Jcrit(p)\left(\frac{n}{\log n}\right)^{-2\alpha/d}.
\end{equation}
\end{corollary}

\begin{remark}
The constant $C_3$ has the same dependencies as in Theorem~\ref{thm:up-int}.
\end{remark}

\subsection{Beta-type, Gaussian and Student targets}
\label{sec:examples}
 
We assume \hypref{hyp:k} and \hypref{hyp:sobolev} throughout this section and we show that assumptions \hypref{hyp:admissible-design}, and \hypref{hyp:cell-average}, as well as \hypref{hyp:tail} in case \textbf{(U)}, hold for three different families of measures. For the Student target, we do so in the regime $\alpha<\nu+d/2$ of Proposition~\ref{prop:student-target}(a). Theorem~\ref{thm:up-int} and Corollaries~\ref{cor:up-int}, \ref{cor:postvar-main} and~\ref{cor:marginal-postvar} therefore apply in these settings.

\subsubsection{Beta-type targets on the unit cube}

\begin{proposition}
\label{prop:cube-beta}
Let $\Xset=\left[0,1\right]^d$ and suppose that $\Pi$ has density
\begin{equation*}
\pi(x)=c\,\varphi(x)\prod_{j=1}^d x_j^{a_j-1}\left(1-x_j\right)^{b_j-1},
\end{equation*}
where $a_j,b_j>0$ for $j=1,\ldots,d$, $c$ is the normalising constant, and $0<m_\varphi\leq\varphi(x)\leq M_\varphi<\infty$ for almost every $x\in\Xset$.
Then the proposal $p^\star$ defined in \eqref{eq:optimal-density} satisfies \hypref{hyp:admissible-design}\textbf{(B)} and \hypref{hyp:cell-average}.
\end{proposition}

The Beta-type family allows densities that vanish or diverge at the boundary. The factor $\varphi$ may couple the coordinates, so the family also includes non-product densities.

\subsubsection{Gaussian targets}

\begin{proposition}
\label{prop:gauss-target}
Let $\Xset=\mathbb{R}^d$ and $\Pi=\mathcal{N}\left(\zero_d,\id_d\right)$. The proposal $p^\star$ defined in \eqref{eq:optimal-density} is the density of $\mathcal{N}\left(\zero_d,\gamma^{-1}\id_d\right)$.
For each $M\geq3$, take $\RootFamily_M$ to be the family of unit cubes $\prod_{i=1}^d\left[j_i,j_i+1\right]$, $\mathbf{j}=(j_1, \ldots, j_d)^\top\in\mathbb{Z}^d$, that intersect the open Euclidean ball $\EuclideanBall\left(\zero_d,\rad_M\right)$ in $\mathbb{R}^d$, where
\begin{equation*}
\rad_M^2=\frac{2\alpha}{d}\log M+\frac{d}{2}\log\log M.
\end{equation*}
With these root families, $p^\star$ satisfies \hypref{hyp:admissible-design}\textbf{(U)}, \hypref{hyp:tail} and \hypref{hyp:cell-average}.

For $\Pi=\mathcal{N}\left(\mu,\bSigma\right)$ with $\bSigma$ positive definite, $P^\star = \mathcal{N}\left(\mu,\gamma^{-1}\bSigma\right)$. The change of variables $x=\mu+\bSigma^{1/2}y$ preserves the conclusions of Theorem~\ref{thm:up-int} and Corollaries~\ref{cor:up-int}, \ref{cor:postvar-main} and~\ref{cor:marginal-postvar}, with constants that may also depend on $\bSigma$.
\end{proposition}

We use unit cubes rather than a single cube to control the variation of the Gaussian design density uniformly on the stopping cells and their parents, hence providing the doubling bound \eqref{eq:doubling}. The union of the unit cubes covers the ball of radius $\rad_M$ needed for the tail bound while avoiding the Gaussian low-density corners of $\left[-\rad_M,\rad_M\right]^d$, for $d \geq 2$.

\subsubsection{Student targets}

\begin{proposition}
\label{prop:student-target}

Let $\Xset=\mathbb{R}^d$, $\Pi=t_\nu\left(\zero_d,\id_d\right)$ be the multivariate Student law with $\nu>0$ degrees of freedom, and
\begin{equation}
\label{eq:student-Rn}
\RootFamily_M=\left\{\left[-\rad_M,\rad_M\right]^d\right\},
\qquad
\rad_M=\begin{cases}
M^{\alpha/\left(d\left(\nu+d/2\right)\right)}, & \alpha<\nu+d/2,\\
M^{1/d}, & \alpha\geq\nu+d/2.
\end{cases}
\end{equation}
\begin{enumerate}
\item[(a)] If $\alpha<\nu+d/2$, then the proposal $p^\star$ defined in \eqref{eq:optimal-density} is the density of $t_{\nu^\star}\left(\zero_d,\frac{\nu}{\nu^\star}\id_d\right)$ with $\nu^\star=\gamma\left(\nu+d\right)-d>0$, and $p^\star$ satisfies \hypref{hyp:admissible-design}\textbf{(U)}, \hypref{hyp:tail} and \hypref{hyp:cell-average}.
\item[(b)] If $\alpha\geq\nu+d/2$, then $Z_\gamma=\infty$. Define the tempered density on the truncated region $\Kregion_M$ by
\begin{equation*}
p_M=\frac{\pi^\gamma}{Z_M}\mathds{1}_{\Kregion_M},
\qquad
Z_M=\int_{\Kregion_M}\pi(x)^\gamma\,\dd x.
\end{equation*} 
Set $\kappa=1/\gamma$ if $\alpha=\nu+d/2$ and $\kappa=0$ if $\alpha>\nu+d/2$.

There exist constants $C_\nu>0$ and $n_\nu\in\mathbb{N}$ such that, for every $n\geq n_\nu$, there is a deterministic design $\bx_{1:n}$ with
\begin{equation*}
\sup_{\left\lVert f\right\rVert_{\Hk}\leq1}\left\lvert\Pi(f)-\widehat{\Pi}_n(f)\right\rvert\leq C_\nu\,n^{-\left(\nu+d/2\right)/d}(\log n)^\kappa.
\end{equation*}
Moreover, for every $c>0$, there exist $A_\nu>0$, $\widetilde{C}_\nu>0$ and $n_c\in\mathbb{N}$ such that, for every $n\geq n_c$, a random design $\bx_{1:n}$ obtained from $n$ independent samples of $p_{\widetilde M_n}$, with $\widetilde M_n=\left\lfloor n/\left(A_\nu\log n\right)\right\rfloor$, satisfies
\begin{equation*}
\sup_{\left\lVert f\right\rVert_{\Hk}\leq1}\left\lvert\Pi(f)-\widehat{\Pi}_n(f)\right\rvert\leq \widetilde{C}_\nu\left(\frac{n}{\log n}\right)^{-\left(\nu+d/2\right)/d}(\log n)^\kappa
\end{equation*}
with probability at least $1-n^{-c}$ over the design.
\end{enumerate}
\end{proposition}
 
In case (a), Theorem~\ref{thm:up-int} gives the rate $n^{-\alpha/d}$ for the proposal $p^\star$.
In case (b), the construction uses $p_M$ on the expanding region $\Kregion_M$ and balances the growth of $Z_M$ against the exterior error. 
The interpolation and tail bounds in \eqref{eq:student-finite-bound} have the same order $M^{-\left(\nu+d/2\right)/d}$, with an extra factor $(\log M)^{1/\gamma}$ at $\alpha=\nu+d/2$ due to the growth of $Z_M$. The extra $\log n$ factor in the partition level for independent sampling ensures coverage with high probability.

\subsubsection{Failure of the doubling condition}

The assumptions used to prove Theorem~\ref{thm:up-int} are general but not automatic: for example, doubling can fail when the density vanishes faster than any polynomial at the boundary. On $[0,1]$, for $\pi(x)=c\exp(-1/x)$, the dyadic interval $\dcube_j=[0,2^{-j}]$, $j\in\mathbb{N}$, and its parent satisfy, for some constant $c'>0$ independent of $j$, 
\begin{equation*}
\frac{P^\star\left(\widehat{\dcube}_j\right)}{P^\star(\dcube_j)} 
\ge c'4\exp(\gamma2^{j-1})\xrightarrow[j\to\infty]{} \infty.
\end{equation*}
Hence the doubling condition \eqref{eq:doubling} fails. The local scale of such a density changes faster than polynomially near the boundary, which is uncommon among standard parametric families.

\subsection{Lower bounds}
\label{sec:lower-bounds}

The upper bounds above apply to deterministic and independent sampling designs. We now study deterministic lower bounds to determine when these upper bounds attain the minimax rate.

Let $\QuadClass_n$ denote the class of deterministic linear quadrature rules with at most $n$ fixed nodes,
\begin{equation*}
\QuadRule_n(f)=\sum_{i=1}^n w_i f\left(x_i\right),\qquad x_i\in\Xset,\quad w_i\in\mathbb{R},
\end{equation*}
and define the $n$-th minimal worst-case error for integration against $\Pi$ by
\begin{equation}
\label{eq:minimal-wce}
e_n\left(\Pi\,;\,\Ha\right)=\inf_{\QuadRule_n\in\QuadClass_n}\sup_{\left\lVert f\right\rVert_{\Ha}\leq1}\left\lvert\Pi(f)-\QuadRule_n(f)\right\rvert.
\end{equation}
Bayesian quadrature with $n$ nodes belongs to $\QuadClass_n$, so lower bounds on \eqref{eq:minimal-wce} apply to it.
The infimum in \eqref{eq:minimal-wce} is over deterministic rules, and the minimax rates below refer to this setting, \ie, our upper bounds control the worst-case error of each realised design. Randomised algorithms, whose error is averaged over their own randomness for each fixed integrand, can attain faster rates \citep{novak1988deterministic}, for instance by correcting a deterministic approximation with Monte Carlo sampling of the residual. In a kernel setting, this amounts to using the kernel interpolant as a control variate, as in the control functionals of \citet{oates2017control}.
Combining our designs with such a Monte Carlo correction is a natural extension, which we do not pursue here.

The following benchmark is classical \citep[Proposition~1, Section~1.3.12]{novak1988deterministic}. Let $\mathrm{Leb}$ denote Lebesgue measure on $\left[0,1\right]^d$.

\begin{theorem}
\label{thm:novak-lower}
Let $\alpha>d/2$ be an integer, and $\mathrm{Leb}$ denote Lebesgue measure on $\left[0,1\right]^d$. There exist constants $0<c\leq C<\infty$, depending only on $\alpha$ and $d$, such that, for every $n\in\mathbb{N}$,
\begin{equation}
\label{eq:novak-lower}
c\,n^{-\alpha/d}\leq e_n\left(\mathrm{Leb}\,;\,\mathcal{H}^\alpha\left(\left[0,1\right]^d\right)\right)\leq C\,n^{-\alpha/d}.
\end{equation}
\end{theorem}

The lower bound extends to every density that is bounded below on some cube.

\begin{corollary}
\label{cor:lower-local}
Let $\alpha>d/2$ be an integer, and suppose that $\Pi$ has a density $\pi$ with $\pi\geq m>0$ almost everywhere on a cube $\LocalCube\subseteq\Xset$. There exists $c>0$, depending only on $\alpha$, $d$, $m$ and the side length of $\LocalCube$, such that, for every $n\in\mathbb{N}$,
\begin{equation}
\label{eq:lower-local}
e_n\left(\Pi\,;\,\Ha\right)\geq c\,n^{-\alpha/d}.
\end{equation}
\end{corollary}

Together with Theorem~\ref{thm:up-int} and Propositions~\ref{prop:cube-beta}--\ref{prop:gauss-target}, Corollary~\ref{cor:lower-local} shows that $n^{-\alpha/d}$ is the minimax rate for the Beta-type and Gaussian targets. For heavy-tailed targets, the tail of $\Pi$ can limit the rate.

\begin{theorem}
\label{thm:lower-student}
Let $\alpha>d/2$ be an integer, let $\nu>0$, and suppose that $\Pi$ has a density $\pi_\nu$ on $\mathbb{R}^d$ such that, for some $c_1>0$,
\begin{equation}
\label{eq:student-density-comparison}
\pi_\nu(x)\geq c_1\left(1+\left\lVert x\right\rVert_2^2\right)^{-\left(\nu+d\right)/2},\qquad x\in\mathbb{R}^d.
\end{equation}
There exists $c>0$, depending only on $\alpha$, $d$, $\nu$ and $c_1$, such that, for every $n\in\mathbb{N}$,
\begin{equation}
\label{eq:lower-student}
e_n\left(\Pi\,;\,\Ha\right)\geq c\,n^{-\min\left(\alpha,\nu+d/2\right)/d}.
\end{equation}
\end{theorem}

Theorem~\ref{thm:lower-student} applies to $t_\nu\left(\zero_d,\id_d\right)$. Together with Proposition~\ref{prop:student-target}, it shows that  the minimax exponent for the Student target is $\min\left(\alpha,\nu+d/2\right)/d$ when $\alpha\neq\nu+d/2$. At $\alpha=\nu+d/2$ there is a gap of $(\log n)^{1/\gamma}$ between the lower and upper bounds.

\section{Numerical Experiments}
\label{sec:num}

\subsection{Experimental design}

We illustrate the behaviour of the proposed sampling strategies in Bayesian quadrature on $\mathbb{R}$ and on $\left[0,1\right]$. Additional experiments are presented in Appendix~\ref{app:num}.

\paragraph*{Matérn kernels}
We use the Matérn $3/2$ and $7/2$ kernels. Let $r\in\mathbb{N}_0$ and consider the Matérn kernel of smoothness $r+1/2$ with marginal variance $\sigma_f^2>0$ and length-scale $\ell_f>0$. Let $v_r=\sqrt{2r+1}/\ell_f$ be its decay rate. The half-integer Matérn kernel factorises as
\begin{equation*}
k_{r+1/2}\left(x,x'\right)=\sigma_f^2 S_r\left(v_r\left\lvert x-x'\right\rvert\right)e^{-v_r\left\lvert x-x'\right\rvert},
\qquad
S_r(t)=\sum_{j=0}^r s_{r,j}t^j,
\quad
s_{r,j} = 2^j\,\frac{\left(2r-j\right)!}{j!\,\left(r-j\right)!}\frac{r!}{\left(2r\right)!}.
\end{equation*}
In dimension one, the associated RKHS is norm-equivalent to $\mathcal{H}^{r+1}(\Xset)$ \citep{Porcu2024Matern}. The Sobolev orders used in \eqref{eq:optimal-density} are therefore $\alpha=2$ for Matérn $3/2$ and $\alpha=4$ for Matérn $7/2$.

\paragraph*{Kernel hyperparameters} 
For each target, we draw $n_{\rm pilot}=100$ independent pilot points and evaluate $f$ at those points. The pilot density is set to the tempered proposal proportional to $\pi^\gamma$ with $\gamma = 2/5$ for the Matérn $3/2$ and $\gamma = 2/9$ for the Matérn 7/2.
They allocate more pilot points to the tails of the Gaussian and Student targets and to regions near the endpoints of $\left[0,1\right]$ for the Beta target.
For each kernel, we use the Metropolis-within-Gibbs algorithm of Appendix~\ref{app:hyper} to sample the posterior distribution of $\left(\sigma_f^2,\ell_f\right)$ conditional on these pilot evaluations. We place an $\mathrm{IG}\left(2,2\right)$ prior on $\sigma_f^2$ and a normal prior with mean $0$ and standard deviation $100$ on $\log\ell_f$. We run the algorithm for $1000$ iterations, with the first $200$ discarded. The hyperparameters are set to the resulting posterior means.
The pilot points are used only to fit the hyperparameters. For each target and kernel, the fitted values are then held fixed across all sampling strategies and replications. The reported Bayesian quadrature posteriors are conditional on these values.

The conditional posterior variance depends on the choice of kernel, including its fitted variance parameter and length-scale. By \eqref{eq:var-worst-case}, it equals the squared worst-case integration error over the unit ball of the fitted RKHS. 
Holding the fitted kernel fixed therefore allows us to compare the sampling strategies on the same function class. For deterministic integrands, \citet{karvonen2020maximum} study how maximum likelihood estimation of the kernel scale affects the ratio of the integration error to the posterior standard deviation. Under our pilot-fitting protocol, a separate analysis would be needed to determine how often the conditional posterior bands contain the true integral $\Pi(f)$ for a fixed $f$ across repeated draws of the design.

\paragraph*{Sampling strategies}
We compare three sampling strategies while keeping the integrand, the kernel hyperparameters, and the Bayesian quadrature update rule fixed. In all three, the evaluation points $\bx_{1:n} = (x_1,\ldots,x_n)$ are independent samples from a proposal distribution: 
\emph{(i)} the integration density $\pi$, \emph{(ii)} the tempered proposal with density $p^\star$ from \eqref{eq:optimal-density}, and \emph{(iii)} the approximation of the regularised leverage-score distribution of \citet{bach2017}.

For the latter, let $(\mu_j,e_j)_{j\geq1}$ be the Mercer eigenpairs of the kernel integral operator on $L^2\left(\Xset,\Pi\right)$, with the $e_j$ orthonormal, and define the effective dimension $d_\lambda=\sum_{j\geq1}\mu_j/(\mu_j+\lambda)$. Combining the normalised definition in \citet[Eq.~(10)]{bach2017} with its quadrature specialisation in \citet[Eq.~(12), \S 4.4]{bach2017}, the optimised distribution has density, relative to $\Pi$,
\begin{equation}
\label{eq:bach-proposal}
\frac{\,\dd q_\lambda^\star}{\,\dd\Pi}(x)
= \frac{1}{d_\lambda}\sum_{j\geq 1}\frac{\mu_j}{\mu_j+\lambda}\,e_j(x)^2,
\end{equation}
a regularised leverage score that gives less weight to spectral directions whose eigenvalues $\mu_j$ are small relative to $\lambda$.

We approximate and sample from \eqref{eq:bach-proposal} using the finite-support construction described 
in \citet[\S 4.2]{bach2017}. 
Let $F_\Pi^{-1}$ be the quantile function of $\Pi$.
For $m=1,\ldots,N_q$, let $\mathsf{I}_m=\left(\left(m-1\right)/N_q,m/N_q\right)$ and let $z_m=F_\Pi^{-1}\left(\left(m-\frac{1}{2}\right)/N_q\right)$ be the representative point of the corresponding quantile cell. Each cell has probability $1/N_q$ under $\Pi$.
We compute an orthonormal eigendecomposition of the rescaled Gram matrix
\begin{equation*}
\widehat{\mathbf A}=\frac{1}{N_q}\left[k\left(z_m,z_{m'}\right)\right]_{m,m'=1}^{N_q}
=\widehat{\mathbf E}\operatorname{diag}\left(\widehat{\mu}_1,\ldots,\widehat{\mu}_{N_q}\right)\widehat{\mathbf E}^{\top}.
\end{equation*}
The empirical leverage probability assigned to the $m$-th quantile cell is
\begin{equation}
\label{eq:bach-discrete-probability}
\widehat q_m=\frac{1}{\widehat d_\lambda}\sum_{j=1}^{N_q}\frac{\widehat{\mu}_j}{\widehat{\mu}_j+\lambda}\widehat{\mathbf E}_{mj}^2,
\qquad
\widehat d_\lambda=\sum_{j=1}^{N_q}\frac{\widehat{\mu}_j}{\widehat{\mu}_j+\lambda}.
\end{equation}
To sample from the approximate proposal $\widehat{q}_\lambda^{\,\,\star}$ of $q_\lambda^\star$, we first draw a cell index $B$ with probabilities $\left(\widehat q_1,\ldots,\widehat q_{N_q}\right)$, then draw $U$ uniformly on $\mathsf{I}_B$, and return $X=F_\Pi^{-1}\left(U\right)$. 
In the $m$-th quantile cell, the resulting Lebesgue density is $\widehat{q}_\lambda^{\,\,\star} = N_q\widehat q_m\pi$. This gives an absolutely continuous proposal, so independent draws are distinct almost surely. We use $N_q=1500$ and fix $\lambda=n_{\max}^{-2\alpha/d}$, where $n_{\max}$ is the final quadrature budget. We construct the proposal once per fitted kernel, at a cost of $\mathcal{O}\left(N_q^3\right)$ for the eigendecomposition, and draw a fresh \iid sample from it in each replication.

We use $\widehat{q}_\lambda^{\,\,\star}$ as a kernel-dependent benchmark. Sampling from $\pi$ ignores the kernel entirely, while $p^\star$ depends on $\pi$ and on the smoothness index $\alpha$ but not on the fitted kernel. By contrast, $\widehat{q}_\lambda^{\,\,\star}$ is built from the Mercer spectrum of the kernel used for Bayesian quadrature, so it provides a natural competitor against which to assess whether the simpler, kernel-free proposal $p^\star$ is already competitive.

\paragraph*{Posterior computation} All experiments use Bayesian quadrature with Gaussian process prior $f^\star\sim\mathcal{GP}(0,k_{r+1/2})$, $r\in\{1, 3\}$.
We report the marginal posterior mean and the expected posterior variance, namely
\begin{equation*}
\mathbb{E}_{\bx_{1:n}}\left[
  \mathbb{E}_{f^\star}\left[\Pi(f^\star)\mid \bx_{1:n}\right]
\right],
\qquad
\mathbb{E}_{\bx_{1:n}}\left[
  \var_{f^\star}\left[\Pi(f^\star)\mid\bx_{1:n}\right]
\right].
\end{equation*}
For each target, kernel and sampling strategy, the integrand $f$ is fixed, and the experiment is repeated $N=1000$ times with independent designs $\bx_{1:n}$, to account for the randomness of the design points.
If $\widehat\Pi_n^{(i)}$ and $\widehat{V}_n^{(i)}$ denote the posterior mean and variance in the $i$-th repetition, $1\le i \le N$, we estimate
the marginal posterior mean and the expected posterior variance by Monte Carlo as
\begin{equation*}
\overline{\Pi}_{n, N} = \frac{1}{N} \sum_{i=1}^N \widehat{\Pi}_n^{(i)}, \qquad
\overline{V}_{n,N} = \frac{1}{N} \sum_{i=1}^N \widehat{V}_n^{(i)} 
\end{equation*}
Note that the posterior variance estimate is computed using the prior variance $\Pi^{\otimes 2}(k_{r+1/2})$ and the identity
\begin{equation*}
\var[\Pi(f^\star)\mid \bx_{1:n}]
=\Pi^{\otimes2}(k) - \bmean_\Pi(\bx_{1:n})^\top
\bK^{-1}
\bmean_\Pi(\bx_{1:n}).
\end{equation*}

We also report $95\%$ bands for the estimation of the marginal posterior mean.
For sampled design $\bx_{1:n}^{(i)}$ at the $i$-th repetition, Bayesian quadrature yields the conditional Gaussian posterior
\begin{equation*}
\Pi(f^\star)\mid \bx_{1:n}^{(i)}
\sim
\mathcal{N}\left(
\widehat{\Pi}_n^{(i)},\widehat{V}_n^{(i)}
\right).
\end{equation*}
Marginalizing over the random design distribution therefore gives the marginal distribution
estimated by the empirical measure
\begin{equation*}
\frac{1}{N}
\sum_{i = 1}^N
\mathcal{N}\left(
\widehat\Pi_n^{(i)},\widehat{V}_n^{(i)}
\right).
\end{equation*}
The bands are the empirical quantiles obtained by drawing $100$ values from each of the $N$ conditional posteriors.

\subsection{Gaussian and Beta experiments}

\subsubsection{Gaussian target}
\label{sec:num-gaussian}

We consider $\Pi=\mathcal{N}\left(0,1\right)$ and the integrand $f$ defined on $\mathbb{R}$ by
\begin{equation}
\label{eq:f-gauss}
f(x)=\sqrt{3}e^{-x^2}+\frac{\sin\left(2\pi x\right)}{1+x^2},
\end{equation}
which is in $\Ha[\mathbb{R}]$ for every integer $\alpha\geq1$. It is similar to the integrand considered by \citet{pmlr-v70-briol17a}, but includes decay at infinity. The exact integral is $\Pi(f) = 1$.

The normalized proposal of Theorem~\ref{thm:up-int} is the density of
\begin{equation*}
    P^\star
    =
    \begin{cases}
        \mathcal{N}(0,5/2),&
        \text{Matérn }3/2,\quad \gamma=2/5,\\[2pt]
        \mathcal{N}(0,9/2),&
        \text{Matérn }7/2,\quad \gamma=2/9.
    \end{cases}
\end{equation*}
The kernel mean and the prior variance admit closed forms. For $y\in\mathbb{R}$ define, for $j\in\mathbb{N}_0$,
\begin{equation*}
I_j(y)=\int_0^\infty s^j e^{-y s-s^2/2}\,\dd s,
\qquad
I_0(y)=\sqrt{\tfrac\pi2}\,e^{y^2/2}\operatorname{erfc}\!\left(\tfrac{y}{\sqrt2}\right),
\qquad
I_1(y)=1-y I_0(y),
\end{equation*}
which satisfy the two-term recursion
\begin{equation*}
I_j(y)=(j-1)\,I_{j-2}(y)-y\,I_{j-1}(y),\qquad j\ge2.
\end{equation*}
The kernel mean is
\begin{equation*}
m_{\Pi,r}(x)=\sigma_f^2\frac{e^{-x^2/2}}{\sqrt{2\pi}}\left[W^{(r)}_{\ell_f}(v_r-x)+W^{(r)}_{\ell_f}(v_r+x)\right],
\qquad
W^{(r)}_{\ell_f}(y)=\sum_{j=0}^r s_{r,j}\,v_r^{\,j}\,I_j(y).
\end{equation*}
The prior variance is obtained by integrating $m_{\Pi,r}$ once more against $\pi$.
Since, for independent standard Gaussian variables $X$ and $X'$, $X-X'\sim\mathcal{N}(0,2)$, the relevant integrals carry the half-step weight $e^{-s^2/4}$, giving the prior variance
\begin{equation*}
\Pi^{\otimes2}(k_{r+1/2})
=\sigma_f^2\,\frac{1}{\sqrt\pi}\sum_{j=0}^r 2^{\left(j+1\right)/2} s_{r,j}\,v_r^{\,j} I_j\left(\sqrt{2}\,v_r\right).
\end{equation*}
With the fitted hyperparameters, the prior integral variances are $7.27\times10^{-1}$ for the Matérn $3/2$ kernel and $2.87\times10^{-1}$ for the Matérn $7/2$ kernel.

\begin{figure}[!t]
    \centering
    \begin{subfigure}{0.32\textwidth}
        \centering
        \includegraphics[trim={0.2cm 0.2cm 0 0}, clip, width=\textwidth]{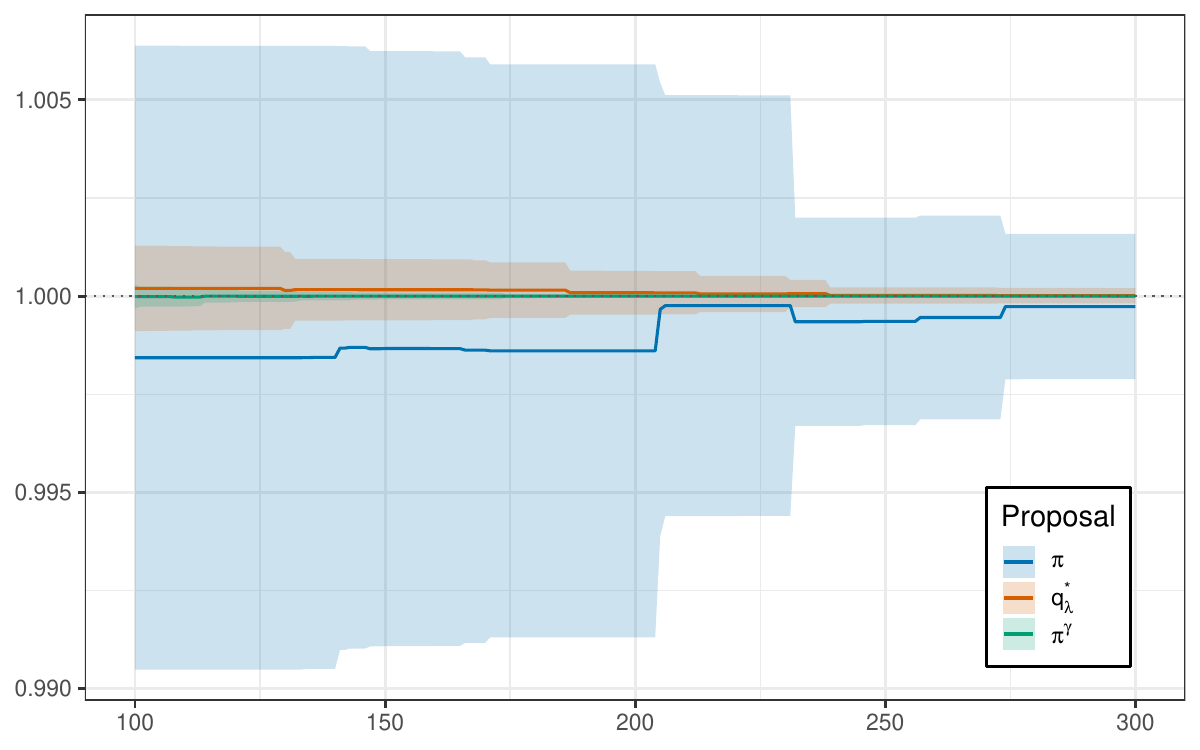}
    \end{subfigure}
    \hfill
    \begin{subfigure}{0.32\textwidth}
        \centering
        \includegraphics[trim={0.2cm 0.2cm 0 0}, clip, width=\textwidth]{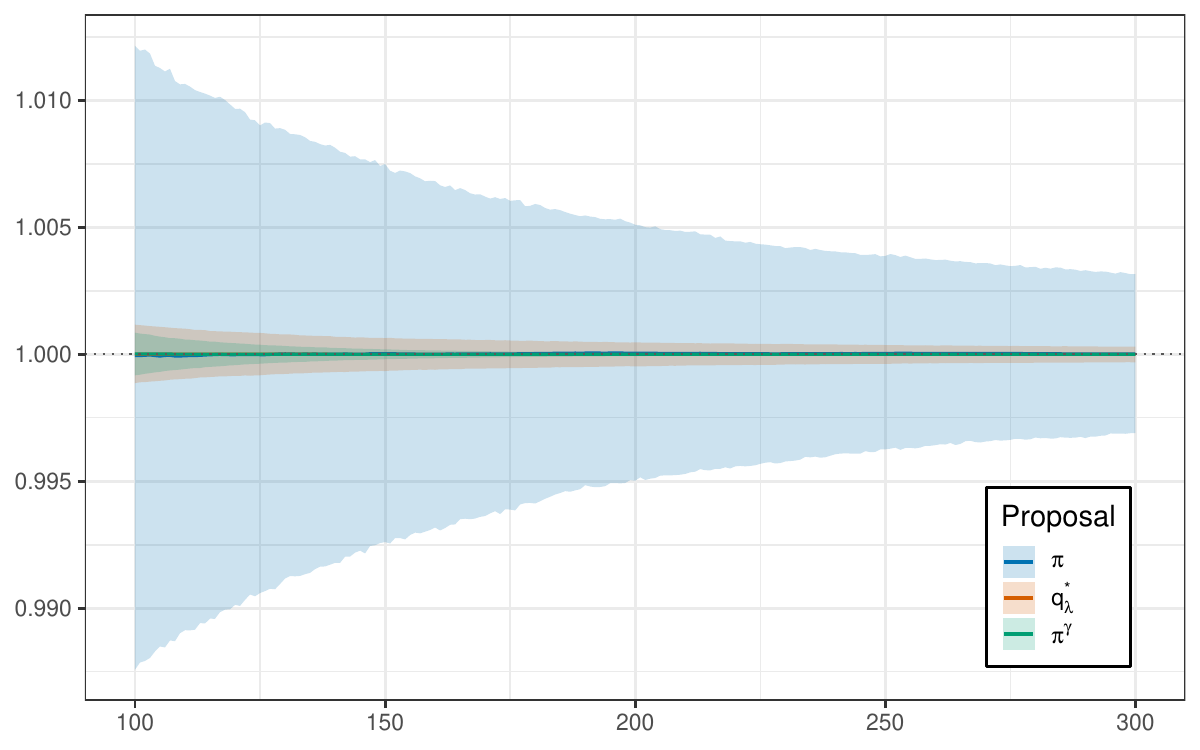}
    \end{subfigure}
    \hfill
    \begin{subfigure}{0.32\textwidth}
        \centering
        \includegraphics[trim={0.2cm 0.2cm 0 0}, clip, width=\textwidth]{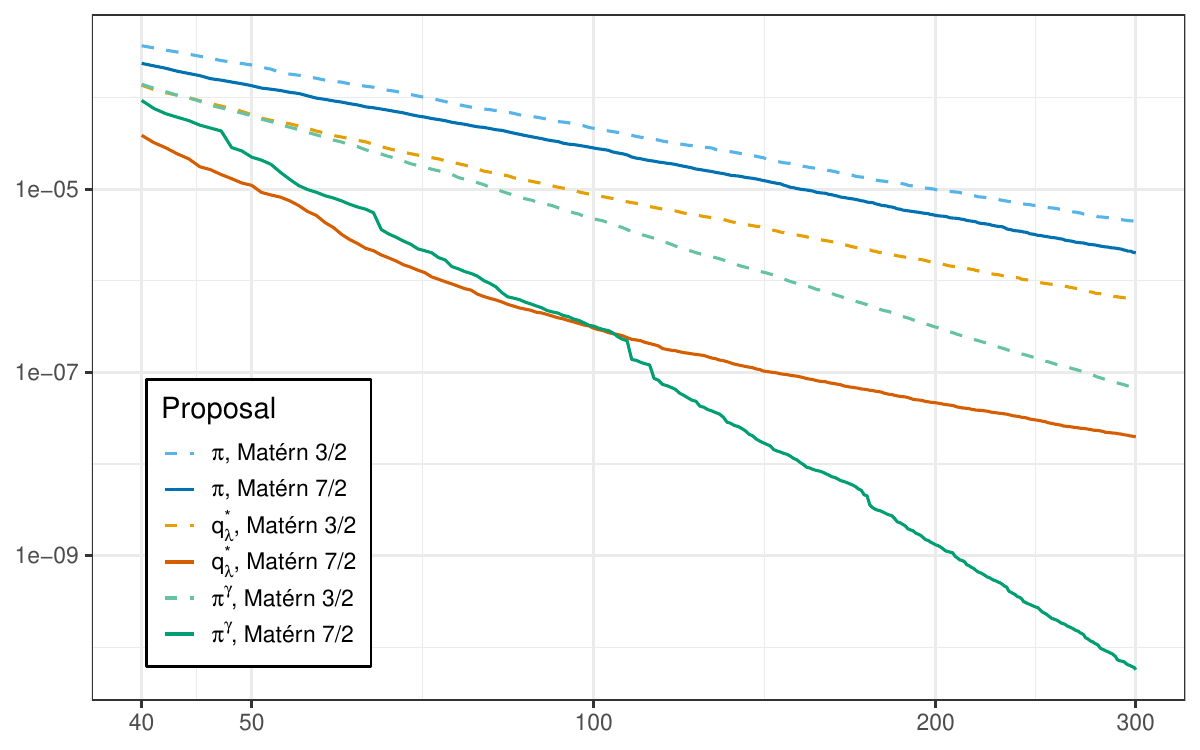}
    \end{subfigure}
    \caption{Posterior summaries for $f$ as in \eqref{eq:f-gauss} under proposals $\pi$, $\widehat{q}_\lambda^{\,\,\star}$  and $p^\star\propto \pi^\gamma$. Left: single-run posterior mean (solid line) $\pm 2$ standard deviation bands for the Matérn $7/2$ kernel against the number $n$ of evaluations; the horizontal dashed line is the true integral value. 
    Centre: marginal posterior mean and $95\%$ bands for the Matérn $7/2$ kernel, estimated from $1000$ independent runs.
    Right: average posterior variance across the same $1000$ runs, for the Matérn $3/2$ (dashed) and Matérn $7/2$ (solid) kernels (log--log scale, from $n=40$).}
    \label{fig:gauss}
\end{figure}

Figure~\ref{fig:gauss} shows that sampling from $p^\star$ substantially improves both point estimation and uncertainty contraction.
On the left panel of Figure \ref{fig:gauss} showing one run per strategy with the Matérn $7/2$ kernel, at $n = 300$ the absolute errors under $\pi$, $\widehat{q}_\lambda^{\,\,\star}$ and $p^\star$ are $2.64\times10^{-4}$, $1.26\times10^{-5}$ and $4.33\times10^{-10}$, and the corresponding posterior variances are $3.70\times10^{-6}$, $4.15\times10^{-8}$ and $1.02\times10^{-10}$.
The marginal posterior mean in the centre panel shows that the accuracy observed in the single run is stable across the $1000$ independent designs.
On the other hand, the right panel shows that tempering gives the smallest averaged posterior variance across replications for both kernels at final iterations.
It also shows that the benefit becomes more pronounced for smoother kernels.
Under $P^\star$, the final average posterior variance decreases from $6.67\times10^{-8}$ for Matérn $3/2$ to $5.71\times10^{-11}$ for Matérn $7/2$.
The corresponding values for $\widehat{q}_\lambda^{\,\,\star}$ are $6.22\times10^{-7}$ and $1.99\times10^{-8}$, while those for target sampling $\pi$ are $4.43\times10^{-6}$ and $2.02\times10^{-6}$.
This example illustrates that wider tail exploration may be particularly advantageous when the additional kernel smoothness permits faster posterior contraction.

\subsubsection{Beta target}
\label{sec:num-beta}

We consider $\Pi=\operatorname{Beta}\left(3,3\right)$ and the integrand $f$ defined on $\left[0,1\right]$ by
\begin{equation}
\label{eq:f-beta}
f(x)=1+\sin\left(2\pi x\right).
\end{equation}
This is the same integrand as the one considered by \citet{pmlr-v70-briol17a}. As in the Gaussian example, the exact integral is $\Pi(f)=1$.

Since the density $\pi$ is proportional to $x^2(1-x)^2$, the normalized proposal of Theorem~\ref{thm:up-int} is the density of
\begin{equation*}
    P^\star =
    \begin{cases}
        \operatorname{Beta}(9/5,9/5),&
        \text{Matérn }3/2,\quad \gamma=2/5,\\[2pt]
        \operatorname{Beta}(13/9,13/9),&
        \text{Matérn }7/2,\quad \gamma=2/9.
    \end{cases}
\end{equation*}
Both proposals flatten the central peak of $\pi$ and allocate more nodes near the endpoints. 
For the Beta measure, we compute the kernel means and the prior variances by adaptive numerical quadrature with relative tolerance $10^{-11}$.
With the fitted hyperparameters, the prior integral variances are $4.81\times10^{-1}$ for the Matérn $3/2$ kernel and $6.36\times10^{-2}$ for the Matérn $7/2$ kernel.

Figure~\ref{fig:beta} shows that flattening the target substantially improves performance relative to sampling from $\pi$.
For the Matérn $7/2$ single run, at $n=80$ the absolute errors under $\pi$, $\widehat{q}_\lambda^{\,\,\star}$ and $p^\star$ are $1.89\times10^{-4}$, $4.23\times10^{-9}$ and $1.10\times10^{-6}$, respectively, while the corresponding posterior variances are $7.84\times10^{-9}$, $9.46\times10^{-13}$ and $1.27\times10^{-12}$.
Under $\pi$, the absolute error is about twice the posterior standard deviation. 
This illustrates the issue studied by \citet{karvonen2020maximum}. For a deterministic integrand, the fitted kernel amplitude influences the width of the posterior bands relative to the actual integration error.
The leverage-score proposal is best in this single run, and it remains slightly better on average: across the $1000$ independent designs, the final mean absolute errors are $2.95\times10^{-5}$, $3.36\times10^{-7}$ and $4.35\times10^{-7}$, in the same order.
Both proposals improve on sampling from $\pi$ by about two orders of magnitude, and $p^\star$ comes close to the leverage-score proposal while being explicit and depending on the kernel only through $\alpha$.

\begin{figure}[!t]
    \centering
    \begin{subfigure}{0.32\textwidth}
        \centering
        \includegraphics[trim={0.2cm 0.2cm 0 0}, clip, width=\textwidth]{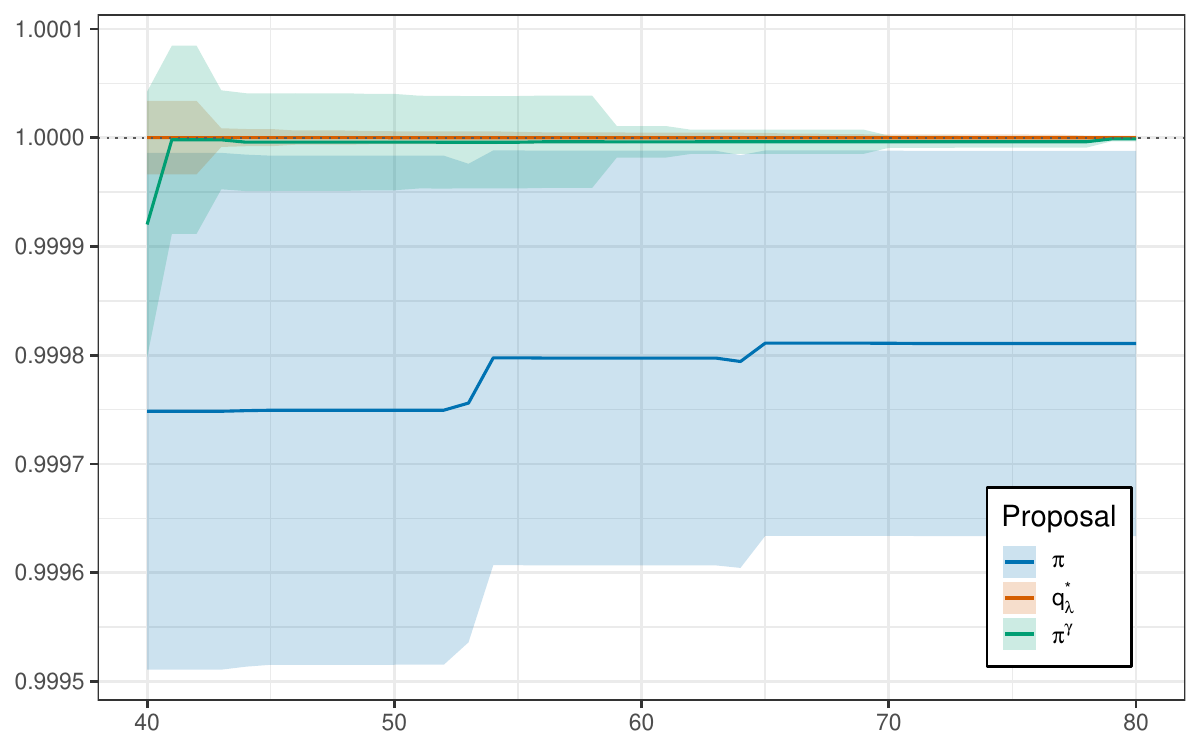}
    \end{subfigure}
    \hfill
    \begin{subfigure}{0.32\textwidth}
        \centering
        \includegraphics[trim={0.2cm 0.2cm 0 0}, clip, width=\textwidth]{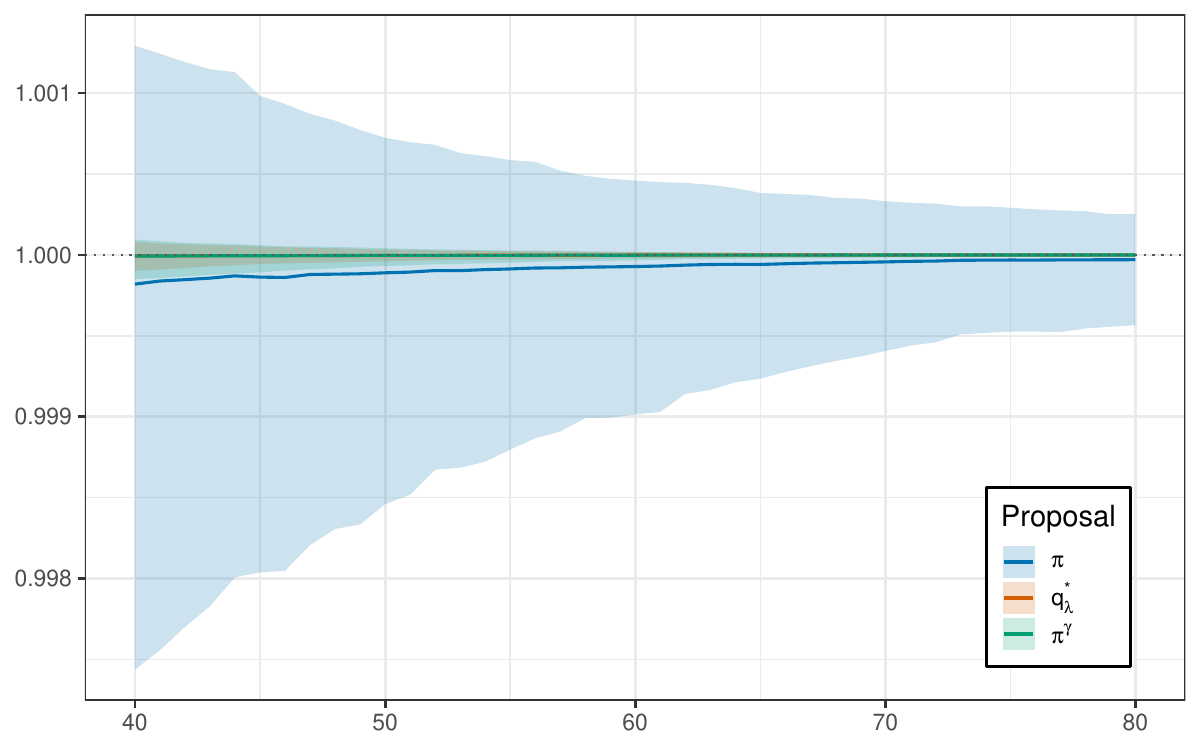}
    \end{subfigure}
    \hfill
    \begin{subfigure}{0.32\textwidth}
        \centering
        \includegraphics[trim={0.2cm 0.2cm 0 0}, clip, width=\textwidth]{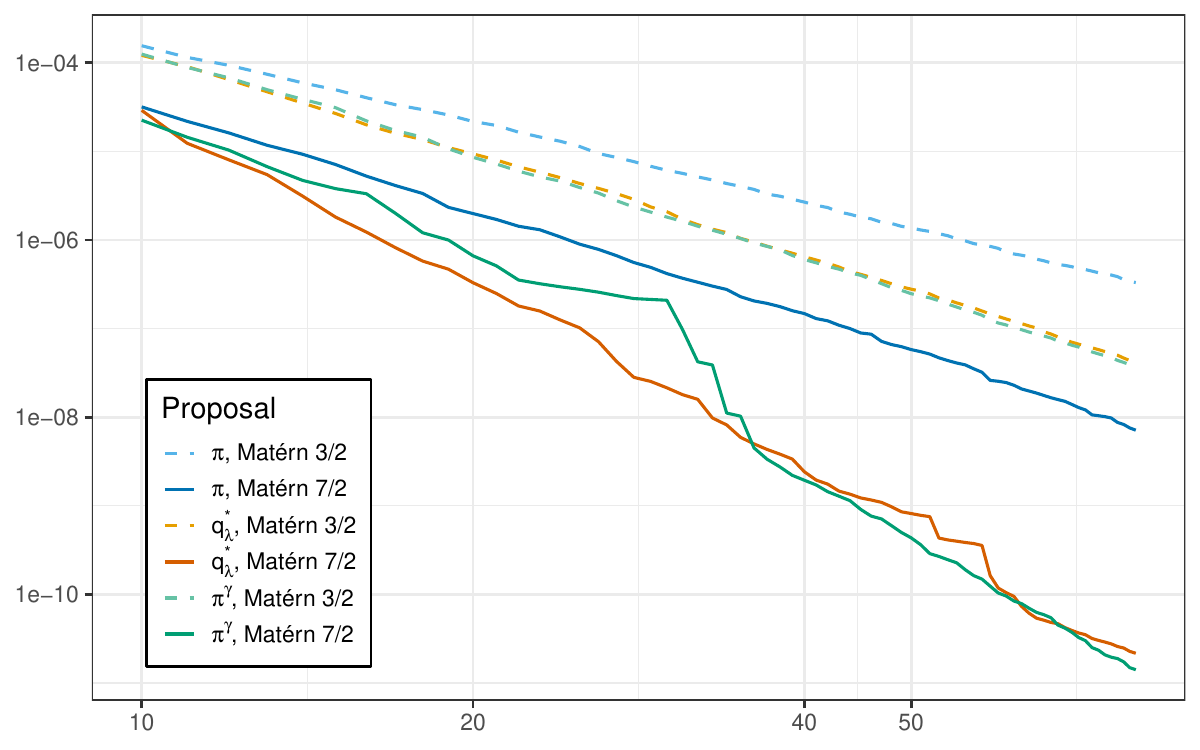}
    \end{subfigure}
    \caption{Posterior summaries for $f$ as in \eqref{eq:f-beta} under proposals $\pi$, $\widehat{q}_\lambda^{\,\,\star}$  and $p^\star\propto \pi^\gamma$. Left: single-run posterior mean (solid line) $\pm 2$ standard deviation bands for the Matérn $7/2$ kernel against the number $n$ of evaluations; the horizontal dashed line is the true integral value. 
    Centre: marginal posterior mean and $95\%$ bands for the Matérn $7/2$ kernel, estimated from $1000$ independent runs.
    Right: average posterior variance across the same $1000$ runs, for the Matérn $3/2$ (dashed) and Matérn $7/2$ (solid) kernels (log--log scale, from $n=10$).}
    \label{fig:beta}
\end{figure}

The right panel gives the corresponding uncertainty comparison.
Under $P^\star$, the final average conditional posterior variance decreases from $3.73\times10^{-8}$ for Matérn $3/2$ to $1.41\times10^{-11}$ for Matérn $7/2$.
The corresponding values for $\widehat{q}_\lambda^{\,\,\star}$ are $4.16\times10^{-8}$ and $2.16\times10^{-11}$, while those for $\pi$ are $3.31\times10^{-7}$ and $7.13\times10^{-9}$.
Hence the leverage-score proposal and $p^\star$ both improve markedly on target sampling and achieve similar results.
The proposal $p^\star$ is nevertheless considerably cheaper to construct and sample from. Here it is an explicit Beta distribution, whereas $\widehat{q}_\lambda^{\,\,\star}$ requires a numerical Mercer eigendecomposition and the construction of leverage-score probabilities before sampling.

\subsubsection{General observations and computational cost}
\label{sec:num-comparison}

\paragraph{Leverage scores and tempering}

The distribution $q_\lambda^\star$ is optimised to control the $L^2\left(\Xset,\Pi\right)$ approximation error for functions $f$ in the unit ball of $\Hk$ 
\citep{bach2017}. 
By duality, this controls the error in approximating $\int_{\Xset}fg\,\dd\Pi$ uniformly over all $g$ with $\left\lVert g\right\rVert_{L^2\left(\Xset,\Pi\right)}\leq1$. Our analysis focuses on the functional $\Pi(f)$, which corresponds to $g=1$, and Theorem~\ref{thm:up-int} identifies $p^\star$ as the minimiser of $\Jcrit(p)$. As $\lambda$ decreases over the range examined, the Lebesgue density of $\widehat{q}_\lambda^\star$ becomes closer to the normalised power $\pi^{d/(2\alpha)}$. This power defines a different density from $p^\star$ and therefore gives a strictly larger value of $\Jcrit$ by Theorem~\ref{thm:up-int}. This difference in the bound helps to explain the gains observed with $p^\star$ in the numerical experiments.

\paragraph{Applicability of the general bound}

For the target density $\pi$, $\Jcrit(\pi)$ is infinite for the Gaussian target, and the Student target of Appendix \ref{app:num}, whenever $\alpha\ge d$, and for the Beta$(3,3)$ target whenever $\alpha\ge5/4$. 
For fixed $N_q$, the density of the leverage-score approximation is a constant multiple of $\pi$ inside each outer quantile cell. It therefore inherits the same divergence in these experiments. 
Therefore, Theorem~\ref{thm:up-int} does not apply to either density in these experiments, so it provides no sufficient quadrature budget $n$ for a prescribed worst-case error tolerance.

For the standard Gaussian target in dimension one, a bound can still be obtained for $p=\pi$ by truncation, as in Proposition~\ref{prop:student-target}(b). Taking the root interval $\left[-\rad_M,\rad_M\right]$ with $\rad_M^2=2\log M-\log\log M$, the bounds on $\Jcrit_M(\pi)$ and the $L^2$ tail give a worst-case error of order at most $M^{-1}(\log M)^{1/4}$ on the coverage event. Taking $M$ of order $n/\log n$ and controlling the exceptional event as in Corollary~\ref{cor:marginal-postvar} gives an upper bound of order $n^{-2}(\log n)^{5/2}$ on the average posterior variance for $\alpha\geq1$. This is consistent with the similar decay of the average posterior variance for the two Matérn kernels under target sampling $\pi$ in Figure~\ref{fig:gauss}.

\paragraph{Computational cost} On a standard laptop, constructing $\widehat{q}_\lambda^{\,\,\star}$ with $N_q=1500$ takes about $4$ seconds, while drawing $n$ points from $p^\star$ takes less than $1$ millisecond, since $p^\star$ is a Gaussian, Beta or Student distributions in our examples. The Bayesian quadrature update itself, of cost $\mathcal{O}(n^3)$, is common to all strategies and takes about $0.15$ seconds at the largest $n$.

\section{Discussion}
\label{sec:conclusion}

This work develops a theory for choosing sampling densities in kernel quadrature on bounded and unbounded domains, including for heavy-tailed integration measures. The local error analysis yields an explicit sampling density determined by the integration measure and the Sobolev order of the kernel. This makes tempering a low-cost alternative to kernel-specific leverage-score designs. For the target families studied, the deterministic upper bounds match the minimax lower bounds, up to a logarithmic factor in the Student target setting at the boundary $\alpha=\nu+d/2$. Independent sampling achieves the same rates with high probability, with an additional logarithmic factor for cell coverage. In our experiments, tempering reduces the average posterior variance relative to target sampling for both Gaussian and Beta targets, and achieves an integration accuracy close to that of the leverage-score design for the Beta target. On unbounded domains, the guarantees also depend on controlling the integration error beyond the region covered by the nodes.

On unbounded domains, the analysis in this paper assumes that the integrand belongs to the global unweighted Sobolev space of order $\alpha$ underlying the kernel quadrature guarantees. This requires the integrand and its weak derivatives up to order $\alpha$ to be square-integrable on $\mathbb{R}^d$.
This condition is used to bound the exterior error by controlling both the $L^2$ norm of the interpolation residual on $\mathbb{R}^d$ and the $L^2$ tail of $\pi$ outside the expanding compact sets.

However, it is plausible that these assumptions can be relaxed.
In particular, one may expect similar behaviour for functions that are smooth locally and remain bounded (or increase at a given rate), without requiring global decay.
A possible extension would be to consider the broader class
\begin{equation*}
f \in \mathcal{H}^\alpha_{\mathrm{loc}}(\mathbb{R}^d) \cap L^\infty(\mathbb{R}^d),
\end{equation*}
where $\mathcal{H}^\alpha_{\mathrm{loc}}(\mathbb{R}^d)$ denotes functions that belong to $\mathcal{H}^\alpha(\dom)$ for every compact set $\dom\subset\mathbb{R}^d$, \ie, functions whose weak derivatives up to order $\alpha$ are locally square-integrable.
Such a class would also allow changes of measure. For a reference measure $Q$ with density $q$ against which the kernel mean embedding and the prior variance are available in closed form, for instance a Gaussian, the identity $\Pi(f)=Q(f\pi/q)$ holds when $\Pi\ll Q$. Extending the error analysis to this broader class could provide guarantees when $f\pi/q$ also belongs to the class.

Extending the theoretical guarantees to this setting would likely require direct control of the approximation error outside compact regions, potentially through uniform bounds on the Gaussian Process posterior mean $s_n$.
Such an extension would likely affect the achievable convergence rates.
\begin{figure}[!t]
    \centering
    \begin{subfigure}{0.32\textwidth}
        \centering
        \includegraphics[trim={0.2cm 0.2cm 0 0}, clip, width=\textwidth]{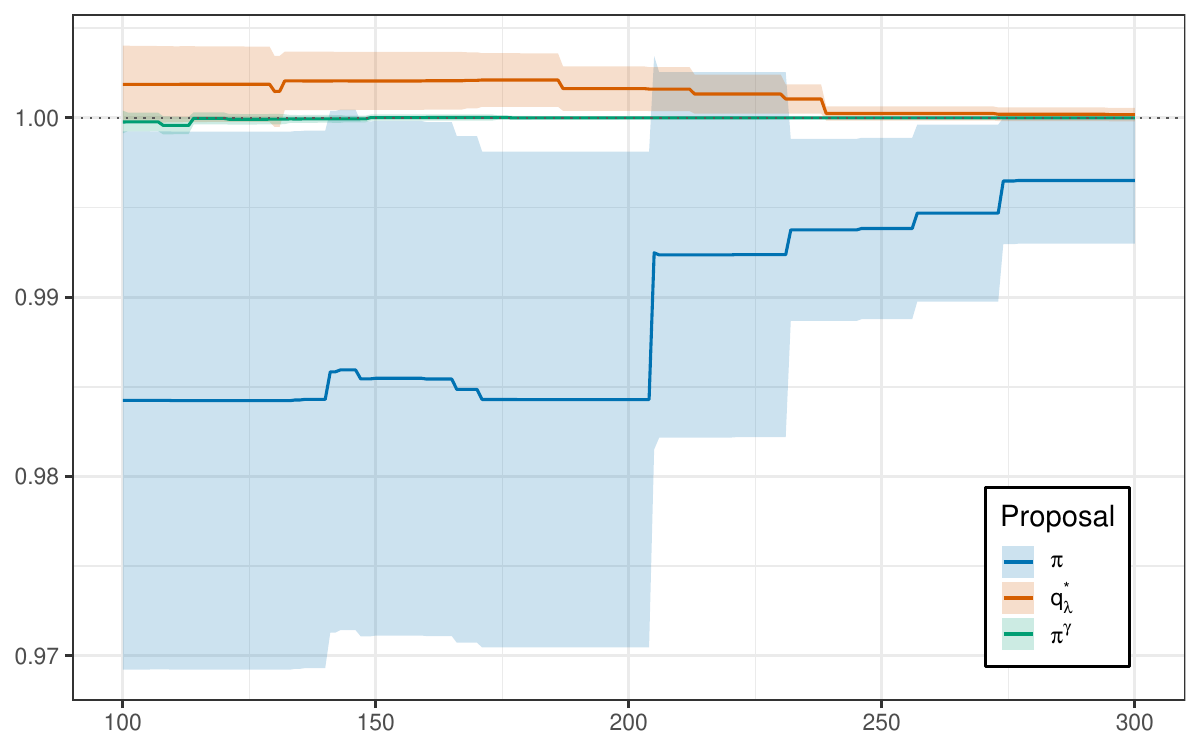}
    \end{subfigure}
    \hfill
    \begin{subfigure}{0.32\textwidth}
        \centering
        \includegraphics[trim={0.2cm 0.2cm 0 0}, clip, width=\textwidth]{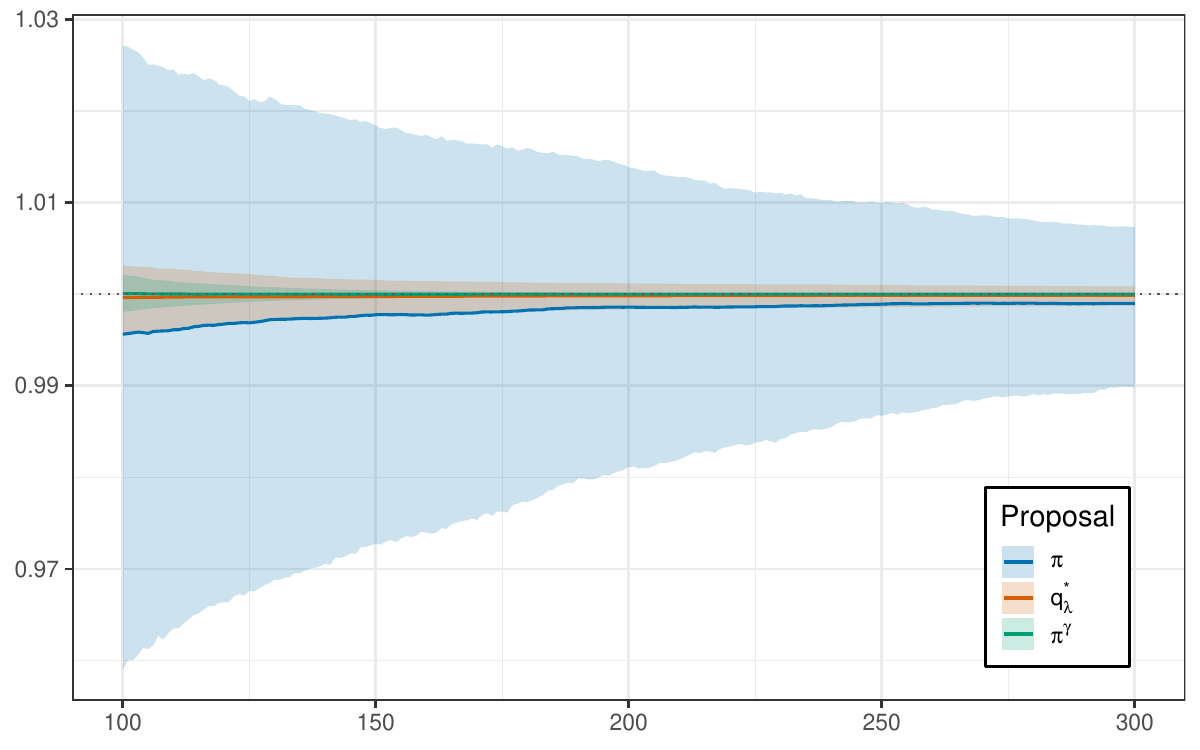}
    \end{subfigure}
    \hfill
    \begin{subfigure}{0.32\textwidth}
        \centering
        \includegraphics[trim={0.2cm 0.2cm 0 0}, clip, width=\textwidth]{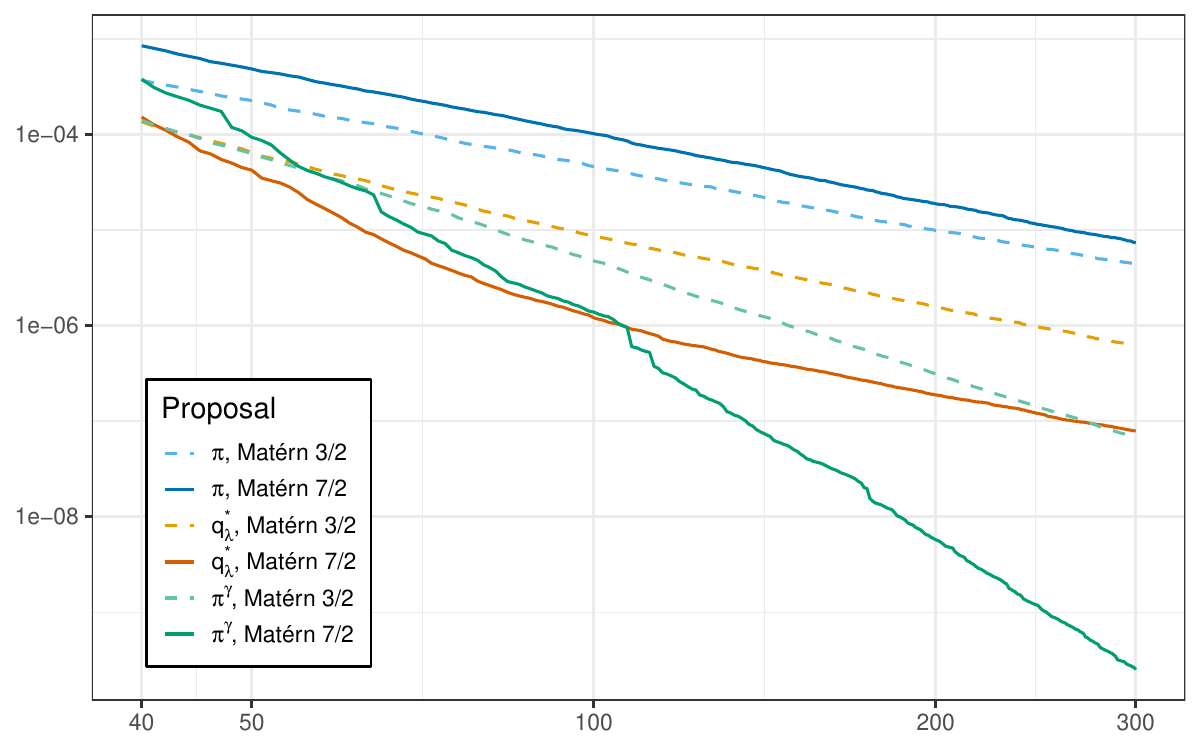}
    \end{subfigure}
    \caption{Posterior summaries for $f(x)=1+\sin(2\pi x)$ and the Gaussian target $\Pi=\mathcal{N}(0,1)$, under proposals $\pi$, $\widehat{q}_\lambda^{\,\,\star}$  and $p^\star\propto\pi^\gamma$. 
    Left: single-run posterior mean (solid line) $\pm 2$ standard deviation bands for the Matérn $7/2$ kernel against the number $n$ of evaluations; the horizontal dashed line is the true integral value. 
    Centre: marginal posterior mean and $95\%$ bands for the Matérn $7/2$ kernel, estimated from $1000$ independent runs.
    Right: average posterior variance across the same $1000$ runs, for the Matérn $3/2$ (dashed) and Matérn $7/2$ (solid) kernels (log--log scale, from $n=40$).}
    \label{fig:discussion}
\end{figure}

To provide preliminary support for the
above, we consider $\Pi = \mathcal{N}(0,1)$ and the function $f(x)=1+\sin(2\pi x)$ which is bounded and infinitely smooth but is periodic with infinite $L^2(\mathbb{R})$ norm and therefore lies outside the global Sobolev class.
We approximate $\Pi(f)$ on $\mathbb{R}$ using Bayesian quadrature under the same experimental setup as in Section~\ref{sec:num}, with both Matérn $3/2$ and $7/2$ kernels.

Figure~\ref{fig:discussion} shows that the qualitative behaviour remains consistent with the mass-decaying setting.
Although these experiments do not constitute theoretical evidence, they suggest that the practical benefits of the proposed sampling strategy may extend beyond the class of globally Sobolev integrands considered in this paper.

The conditions on $\pi$ used in Section~\ref{sec:examples} are sufficient but not necessary. On bounded domains, they allow densities that vanish or blow up at the boundary at a polynomial rate, as for the Beta-type targets of Proposition~\ref{prop:cube-beta}. 
Near $x=0$, a density proportional to $e^{-1/x}$ can violate the doubling condition, while a density proportional to $\left\{x\log^2\left(1/x\right)\right\}^{-1}$ can violate the reverse H\"older inequality~\eqref{eq:RH}.
We expect that such densities can be treated by a finer construction of the partition.

Our results also suggest other directions for extensions.
First, the present analysis focuses on exact function evaluations.
Extensions to noisy evaluations or numerical error would be important for statistical applications where likelihoods or simulators are themselves approximate.
Second, the framework may be extended beyond Gaussian process priors to alternative surrogate models, including Student--$t$ processes \citep{Shah2014StudentT,Pruher2017StudentT}.

Finally, our results do not extend to infinitely smooth kernels simply by letting $\alpha$ grow to infinity. 
On a bounded domain, as $\gamma\downarrow0$, the density $p^\star$ converges to the uniform density on the support of $\pi$.
Such kernels call for sampling inequalities with exponentially small error terms \citep{rieger2010sampling}.
On bounded domains, the rates in these inequalities improve when the points cluster near the boundary \citep{rieger2014improved}.
Gauss quadrature points behave in the same way: in one dimension, as the Gaussian kernel becomes increasingly flat, worst-case optimal kernel quadrature rules with optimally chosen nodes converge to the Gauss rule of the integration measure in the flat-kernel limit \citep{KarvonenSarkka2020}.
For any regular measure on an interval, the Gauss nodes are asymptotically distributed according to the arcsine law, which concentrates near the endpoints \citep{KuijlaarsVanAssche1999}.
This suggests that for infinitely smooth kernels the optimal design may be governed by the geometry of the domain rather than by a power of $\pi$.
Making this precise for a fixed lengthscale, as the number of nodes grows, remains open.

\section{Code availability}

The code used in this article is available at \url{https://github.com/edoardobandoni00/Opt-BQ-sampling}.

\phantomsection\label{supplementary-material}
\bigskip

\acks{Edoardo Bandoni and Christian P. Robert are partly supported by the European Union under the ERC Synergy Grant 101071601 (OCEAN, 2023--2030). Views and opinions expressed are solely those of the authors and do not necessarily reflect those of the European Union or the European Research Council Executive Agency. Neither the European Union nor the granting authority can be held responsible for them. Christian P. Robert is also partly supported by a PR[AI]RIE-PSAI Chair funded by the Agence Nationale de la Recherche (ANR-23-IACL-0008). The authors are grateful to F.-X. Briol and V. Elvira for their comments.}



\newpage

\appendix
\section{Proofs}
\label{app:proofs}

All the upper bounds of Section~\ref{sec:main} rest on one deterministic estimate, Proposition~\ref{prop:cell-bound}, which bounds the part of the integration error carried by a union of dyadic cubes in terms of the discrete functional $\Jcrit_M(p)$.
The proofs are arranged around it: we first collect the three tools it needs, then prove it, and then add, one at a time, the ingredients that turn it into the statements of Section~\ref{sec:main}.
 
\subsection{Preliminaries}
\label{app:prelim}
 
The only analytic input of the upper bounds is a sampling inequality: a Sobolev function that vanishes on a sufficiently dense finite set is uniformly small.
We need it on cubes of arbitrary size and position, with constants that depend on neither, so we state it in a scale-invariant form.
A set $\OpenSet\subset\mathbb{R}^d$ satisfies an interior cone condition if there exist $\varrho>0$ and $\theta\in(0,\pi/2)$ such that, for every $x\in\OpenSet$, there is a unit vector $\xi(x)$ for which the cone
\begin{equation*}
\ConeSet\left(x,\xi(x),\theta,\varrho\right)
=
\left\{
x+\lambda y:\ \lVert y\rVert_2=1,\ y^\top\xi(x)\ge\cos\theta,\ \lambda\in[0,\varrho]
\right\}
\end{equation*}
is contained in $\OpenSet$~\citep{Wendland_2004}.
The radius $\varrho$ and the aperture $\theta$ are fixed for the whole set, so cusps and thin spikes are excluded.
Every cube satisfies the condition with an aperture depending only on $d$ and a radius proportional to its side length.
For an integer $\alpha$, write the Sobolev seminorm as
\begin{equation*}
    \lvert u\rvert_{\mathcal{H}^\alpha(\OpenSet)}^2=\sum_{\lVert\beta\rVert_1=\alpha}\lVert D^\beta u\rVert_{L^2(\OpenSet)}^2,
\end{equation*}
so that $\lvert u\rvert_{\mathcal{H}^\alpha(\OpenSet)}\le\lVert u\rVert_{\mathcal{H}^\alpha(\OpenSet)}$.
 
\begin{lemma}
\label{lem:local-sampling}
Let $\alpha>d/2$ be an integer. There are constants $h_0\in(0,1)$ and $c_{\mathrm{samp}}>0$, depending only on $\alpha$ and $d$, such that the following holds for every cube $\dcube\subset\mathbb{R}^d$.
If $\NodeSet\subset \dcube$ is finite and non empty with $\rho(\NodeSet,\dcube)\le h_0\operatorname{diam}(\dcube)$, and $u\in\mathcal{H}^\alpha(\dcube)$ vanishes on $\NodeSet$, then
\begin{equation}
\label{eq:local-sampling}
    \lVert u\rVert_{L^\infty(\dcube)}
    \le
    c_{\mathrm{samp}}\operatorname{diam}(\dcube)^{\alpha-d/2}
    \lvert u\rvert_{\mathcal{H}^\alpha(\dcube)}\le c_{\mathrm{samp}}\operatorname{diam}(\dcube)^{\alpha-d/2}
    \lVert u\rVert_{\mathcal{H}^\alpha(\dcube)}.
\end{equation}
\end{lemma}
 
\begin{proof}
Let $\dcube_0$ be a cube of unit diameter. It is a bounded Lipschitz domain satisfying an interior cone condition. The sampling inequality for functions vanishing on a discrete set \citep[Theorem~2.12]{NarcowichWardWendland2005}, see also \citet[Theorem~2.1.1]{rieger2008sampling}, gives constants $h_0\in(0,1)$ and $C>0$, depending only on $\alpha$ and $d$, such that
\begin{equation*}
    \lVert v\rVert_{L^\infty(\dcube_0)}\le Ch^{\alpha-d/2}\lvert v\rvert_{\mathcal{H}^\alpha(\dcube_0)}\le C\lvert v\rvert_{\mathcal{H}^\alpha(\dcube_0)}
\end{equation*}
whenever $v\in\mathcal{H}^\alpha(\dcube_0)$ vanishes on a set of fill distance $h\le h_0$ in $\dcube_0$.
Every cube is of the form $\dcube=a+\delta \dcube_0$ with $\delta=\operatorname{diam}(\dcube)$. Put $v(y)=u(a+\delta y)$. Then $v$ vanishes on $(\NodeSet-a)/\delta$, whose fill distance in $\dcube_0$ is $\rho(\NodeSet,\dcube)/\delta\le h_0$. The chain rule and a change of variables give
\begin{equation*}
    \lVert v\rVert_{L^\infty(\dcube_0)}=\lVert u\rVert_{L^\infty(\dcube)},\qquad \lvert v\rvert_{\mathcal{H}^\alpha(\dcube_0)}=\delta^{\alpha-d/2}\lvert u\rvert_{\mathcal{H}^\alpha(\dcube)}.
\end{equation*}
This proves \eqref{eq:local-sampling} with $c_{\mathrm{samp}}=C$.
\end{proof}
 
The factor $\operatorname{diam}(\dcube)^{\alpha-d/2}$ in \eqref{eq:local-sampling} is what converts the local mesh size into the rate $n^{-\alpha/d}$: summed over cells, it produces the weights $\langle p\rangle_\dcube^{-2\alpha/d}$ in $\Jcrit_M(p)$.
For this to happen, the cells must be adapted to the design density, so that each receives the same number of nodes.
The stopping partition achieves this by splitting the roots into cubes of nearly equal $P$-mass.
 
Fix a root family $\RootFamily$ with union $\Kregion$, a density $p$ with measure $P=p\,\dd x$, and $M\ge2$, and recall that $\tau_M=P(\Kregion)/M$.
For every $x\in \Kregion$ outside the boundaries of the cubes of $\DyadicFamily(\RootFamily)$, which form a countable union of hyperplane pieces and hence a Lebesgue-null set, $x$ lies in a unique root $\Kcube\in\RootFamily$, and there is a unique chain
\begin{equation*}
    \Kcube=\dcube^{(0)}(x)\supset \dcube^{(1)}(x)\supset \dcube^{(2)}(x)\supset\cdots,
    \qquad
    x\in \dcube^{(j)}(x),\quad \operatorname{vol}\left(\dcube^{(j)}(x)\right)=2^{-dj}\operatorname{vol}(\Kcube).
\end{equation*}
A cube $\dcube\in\DyadicFamily(\RootFamily)$ belongs to $\StopPartition_M$ exactly when it is the first cube along such a chain with $P(\dcube)\le\tau_M$.
 
\begin{lemma}
\label{lem:dyadic-stopping}
Suppose that every root $\Kcube\in\RootFamily$ satisfies $P(\Kcube)>\tau_M$ and that $P$ is $D$-doubling at level $M$ on $\RootFamily$. Then $\StopPartition_M$ is a finite partition of $\Kregion$ up to a Lebesgue-null set, and

\noindent
\begin{minipage}[t]{0.62\linewidth}
\begin{equation}
\label{eq:stop-mass}
\frac{P(\Kregion)}{DM}<P(\dcube)\leq\frac{P(\Kregion)}{M},
\qquad \dcube\in\StopPartition_M,
\end{equation}
\end{minipage}\hfill
\begin{minipage}[t]{0.35\linewidth}
\begin{equation}
\label{eq:stop-count}
\vphantom{\frac{P(\Kregion)}{DM}}
M\leq\lvert\StopPartition_M\rvert<DM.
\end{equation}
\end{minipage}
\end{lemma}
 
\begin{proof}
Since $P$ is absolutely continuous, continuity from above along the chain gives \begin{equation*}
P\left(\dcube^{(j)}(x)\right) \xrightarrow[j\to\infty]{} P(\{x\})=0,
\end{equation*}
so every $x$ off the boundaries lies in a cube of the chain with mass at most $\tau_M$, and the first such cube belongs to $\StopPartition_M$.
Two distinct cubes of $\StopPartition_M$ cannot be nested, by maximality, and therefore have disjoint interiors.
Hence $\StopPartition_M$ covers $\Kregion$ up to a null set and is a partition; it is countable because $\DyadicFamily(\RootFamily)$ is, since $\RootFamily$ is finite and each tree contains finitely many cubes of each generation.
 
The upper bound in \eqref{eq:stop-mass} is the stopping rule.
Since every root has mass larger than $\tau_M$, every $\dcube\in\StopPartition_M$ is a strict descendant of a root, its parent $\widehat \dcube$ lies in $\DyadicFamily(\RootFamily)$, and maximality gives $P(\widehat \dcube)>\tau_M$. As $P(\dcube)\le\tau_M$, \eqref{eq:doubling} applies and yields $P(\dcube)\ge D^{-1}P(\widehat \dcube)>\tau_M/D$.
Summing the two bounds in \eqref{eq:stop-mass} over the partition, whose masses add up to $P(\Kregion)$, gives \eqref{eq:stop-count}; in particular, $\StopPartition_M$ is finite.
\end{proof}

Lemma~\ref{lem:dyadic-stopping} is the only consequence of the doubling condition used for deterministic designs.
The upper bound in \eqref{eq:stop-mass} enters the error estimate, while the lower bound controls the number of cells, and hence the number of nodes.
Both are statements at a fixed level $M$.
To identify the limit of $\Jcrit_M(p)$ as $M\to\infty$, we also need the cells to shrink, so that averages over them converge to point values.
 
Under \hypref{hyp:admissible-design}, write $\StopPartition_M=\StopPartition_M(\RootFamily_M,P)$, and let $\BoundarySet$ be the union over $M\ge M_0$ of the boundaries of the cubes of $\DyadicFamily(\RootFamily_M)$. As a countable union of pieces of hyperplanes, $\BoundarySet$ is Lebesgue-null. For $x\in \Kregion_M\setminus \BoundarySet$, let $\dcube_M(x)$ denote the unique cell of $\StopPartition_M$ that contains $x$. For a locally integrable $h$, define the averaging operator 
\begin{equation*}
T_Mh=\sum_{\dcube\in\StopPartition_M}\langle h\rangle_\dcube\mathds{1}_\dcube,
\end{equation*}
which vanishes off $\Kregion_M$, so that $T_Mh(x)=\langle h\rangle_{\dcube_M(x)}$.

\begin{lemma}
\label{lem:shrink}
Assume \hypref{hyp:admissible-design} and that $p>0$ almost everywhere on $\Xset$. For every $x\in\Xset\setminus \BoundarySet$, $x\in \Kregion_M$ for all large $M$, and
\begin{equation*}
    \lim_{M\to\infty}\operatorname{diam}(\dcube_M(x))=0.
\end{equation*}
Consequently, for every $h\in L^1_{\mathrm{loc}}(\Xset)$, $T_Mh(x)\xrightarrow[M\to\infty]{} h(x)$ for almost every $x\in\Xset$.
\end{lemma}
 
\begin{proof}
Fix $x\in\Xset\setminus \BoundarySet$. By \hypref{hyp:admissible-design}, $x\in \Kregion_M$ for all large $M$.
Suppose that $\operatorname{diam}(\dcube_M(x))$ does not tend to zero as $M\to\infty$. Then there are $s>0$ and a sequence $M_k\to\infty$, as $k\to\infty$, such that the side of $\dcube_{M_k}(x)$ is at least $s$. Each $\dcube_{M_k}(x)$ therefore contains a cube $\AuxCube_k=a_k+[0,s]^d$ with $x\in \AuxCube_k\subseteq\Xset$. The corners $a_k$ lie in the compact set $\prod_{i=1}^d[x_i-s,x_i]$, so, along a subsequence as $k\to\infty$, $a_k\to a$ and, by dominated convergence, $P(\AuxCube_k)\to P(a+[0,s]^d)$. This limit is positive, since $p>0$ almost everywhere. On the other hand, the stopping rule gives
\begin{equation*}
    P(\AuxCube_k)\le P\left(\dcube_{M_k}(x)\right)\le\tau_{M_k}\le\frac{1}{M_k}\xrightarrow[k\to\infty]{}0,
\end{equation*}
a contradiction.
 
For the second claim, recall that $x$ is a Lebesgue point of $h$ if
\begin{equation*}
    \lim_{\delta\to0}\frac{1}{\operatorname{vol}(\EuclideanBall(x,\delta))}\int_{\EuclideanBall(x,\delta)}\lvert h(y)-h(x)\rvert\,\dd y=0,
\end{equation*}
and that, by the Lebesgue differentiation theorem, almost every $x$ is a Lebesgue point of $h\in L^1_{\mathrm{loc}}$ \citep[Section~3.4]{folland1999real}.
Let $x\in\Xset\setminus \BoundarySet$ be such a point, and put $\delta_M=\operatorname{diam}(\dcube_M(x))$, which tends to zero as $M\to\infty$ by the first part. Since $x\in \dcube_M(x)$, the cell is contained in the ball $\EuclideanBall(x,\delta_M)$, and
  \begin{equation*}
    \operatorname{vol}(\dcube_M(x))=d^{-d/2}\delta_M^d=c_d\operatorname{vol}\left(\EuclideanBall(x,\delta_M)\right),
    \qquad
    c_d=\frac{d^{-d/2}}{\operatorname{vol}(\EuclideanBall(\zero_d,1))}.
\end{equation*} 
Therefore
\begin{equation*}
    \left\lvert T_Mh(x)-h(x)\right\rvert
    \le\frac{1}{c_d\operatorname{vol}(\EuclideanBall(x,\delta_M))}\int_{\EuclideanBall(x,\delta_M)}\lvert h(y)-h(x)\rvert\,\dd y ,
\end{equation*}
and the right-hand side tends to zero as $M\to\infty$.
\end{proof}

\subsection{The cellwise estimate}
\label{app:cell}
 
Throughout, $C_{\mathrm{eq}}$ denotes a constant such that $\lVert g\rVert_{\Ha}\le C_{\mathrm{eq}}\lVert g\rVert_{\Hk}$ for all $g\in\Hk$, which exists by \hypref{hyp:sobolev}.
Let $h_0$ and $c_{\mathrm{samp}}$ be the constants of Lemma~\ref{lem:local-sampling}, and set
\begin{equation*}
    \ell_0=\left\lceil\log_2(1/h_0)\right\rceil,
    \qquad
    m_0=2^{d\ell_0}.
\end{equation*}
\begin{definition}
A design $\bx_{1:n}$ is said to fill $\StopPartition_M$ if each of the $2^{d\ell_0}$ descendants at depth $\ell_0$ of every $\dcube\in\StopPartition_M$ contains at least one node.    
\end{definition}
Every point of such a $\dcube$ then lies within distance $2^{-\ell_0}\operatorname{diam}(\dcube)\le h_0\operatorname{diam}(\dcube)$ of a node in $\dcube$, and placing one node in each descendant gives a filling design with $m_0\lvert\StopPartition_M\rvert$ nodes.
For a root family $\RootFamily$, we use the averaging operator $T_M$ defined before Lemma~\ref{lem:shrink}, with $\StopPartition_M=\StopPartition_M(\RootFamily,P)$, so that
\begin{equation*}
\Jcrit_M(p)=\int_{\Kregion}\left\{(T_Mp)(x)\right\}^{-2\alpha/d}\left\{(T_M\pi)(x)\right\}^2\,\dd x.
\end{equation*}
 
The proof of the next proposition has three steps.
The sampling inequality bounds the residual on each cell by its local Sobolev seminorm, multiplied by a power of the cell diameter.
The Cauchy--Schwarz inequality over the partition then separates these seminorms, whose squares add up to at most the global norm, from a purely geometric sum.
Finally, the stopping rule expresses the cell volumes through their masses, $P(\dcube)\le\tau_M$, which turns the geometric sum into $M^{-2\alpha/d}\Jcrit_M(p)$.

\begin{proposition}
\label{prop:cell-bound}
Assume \hypref{hyp:k} and \hypref{hyp:sobolev}. Let $\RootFamily$ be a root family with union $\Kregion\subseteq\Xset$, let $M\ge2$, and suppose that every root has $P$-mass larger than $\tau_M$ and that $P$ is $D$-doubling at level $M$ on $\RootFamily$. Let $\bx_{1:n}\in\Xset^n$ fill $\StopPartition_M=\StopPartition_M(\RootFamily,P)$.
For $f\in\Hk$ with $\lVert f\rVert_{\Hk}\le1$, let $s_n$ be its kernel interpolant at $\bx_{1:n}$ and $\zeta_n=f-s_n$. Then
\begin{equation}
\label{eq:interior-bound}
    \left\lvert\int_{\Kregion}\zeta_n(x)\pi(x)\,\dd x\right\rvert
    \le
    c_{\mathrm{loc}}\,M^{-\alpha/d}\,
    \Jcrit_M(p)^{1/2},
    \qquad
    c_{\mathrm{loc}}=c_{\mathrm{samp}}\,C_{\mathrm{eq}}\,d^{\alpha/2-d/4}.
\end{equation}
\end{proposition}
 
\begin{proof}
Since $s_n$ is the orthogonal projection of $f$ onto $\operatorname{span}\{k(\cdot,x_i):1\le i\le n\}$,
\begin{equation}
\label{eq:projection-contraction}
    \lVert \zeta_n\rVert_{\Hk}\le\lVert f\rVert_{\Hk}\le1,
\end{equation}
and therefore $\lVert \zeta_n\rVert_{\Ha}\le C_{\mathrm{eq}}$.
For $\dcube\in\StopPartition_M$, because the design fills $\StopPartition_M$, the nodes in $\dcube$ have fill distance at most $h_0\operatorname{diam}(\dcube)$ in $\dcube$, and $\zeta_n$ vanishes on them. Lemma~\ref{lem:local-sampling} gives
\begin{equation*}
    \lVert \zeta_n\rVert_{L^\infty(\dcube)}\le c_{\mathrm{samp}}\operatorname{diam}(\dcube)^{\alpha-d/2}\lvert \zeta_n\rvert_{\mathcal{H}^\alpha(\dcube)}.
\end{equation*}
The cells have disjoint interiors, so
  \begin{equation*}
    \sum_{\dcube\in\StopPartition_M}\lvert \zeta_n\rvert_{\mathcal{H}^\alpha(\dcube)}^2\le\lVert \zeta_n\rVert_{\Ha}^2\le C_{\mathrm{eq}}^2.
\end{equation*} 
By the Cauchy--Schwarz inequality over the partition,
\begin{equation}
\label{eq:local-recombination}
    \left\lvert\int_{\Kregion}\zeta_n\pi\,\dd x\right\rvert
    \le
    \sum_{\dcube\in\StopPartition_M}\lVert \zeta_n\rVert_{L^\infty(\dcube)}\Pi(\dcube)
    \le
    c_{\mathrm{samp}}C_{\mathrm{eq}}
    \left(\sum_{\dcube\in\StopPartition_M}\operatorname{diam}(\dcube)^{2\alpha-d}\,\Pi(\dcube)^2\right)^{1/2}.
\end{equation} 
For a cube, $\operatorname{diam}(\dcube)^d=d^{d/2}\operatorname{vol}(\dcube)$. Moreover, $\Pi(\dcube)=\operatorname{vol}(\dcube)\langle\pi\rangle_\dcube$ and 
\begin{equation*}
\operatorname{vol}(\dcube)=P(\dcube)/\langle p\rangle_\dcube\le(M\langle p\rangle_\dcube)^{-1},
\end{equation*}
by \eqref{eq:stop-mass} and $P(\Kregion)\le1$; note that $\langle p\rangle_\dcube>0$ by the lower bound in \eqref{eq:stop-mass}. Hence
\begin{equation*}
    \operatorname{diam}(\dcube)^{2\alpha-d}\,\Pi(\dcube)^2
    =d^{\alpha-d/2}\operatorname{vol}(\dcube)^{2\alpha/d}\,\operatorname{vol}(\dcube)\langle\pi\rangle_\dcube^2
    \le
    d^{\alpha-d/2}M^{-2\alpha/d}\operatorname{vol}(\dcube)\langle p\rangle_\dcube^{-2\alpha/d}\langle\pi\rangle_\dcube^2 .
\end{equation*}
Summing over $\dcube\in\StopPartition_M$ and using \eqref{eq:J-discrete} gives
\begin{equation*}
    \sum_{\dcube\in\StopPartition_M}\operatorname{diam}(\dcube)^{2\alpha-d}\Pi(\dcube)^2\le d^{\alpha-d/2}M^{-2\alpha/d}\Jcrit_M(p),
\end{equation*}
which together with \eqref{eq:local-recombination} proves \eqref{eq:interior-bound}.
\end{proof}

Proposition~\ref{prop:cell-bound} holds at a fixed level $M$, for any root family and any density for which the stopping partition is well defined.
It involves neither the tail of $\pi$, nor the limit $M\to\infty$, nor the way the nodes are generated.
These three ingredients are added separately: the tail in the proof of Theorem~\ref{thm:up-int}, the limit through \hypref{hyp:cell-average}, and random nodes in Section~\ref{app:random}.

\subsection{Proof of Theorem~\ref{thm:up-int}}
\label{app:thm}
 
The proof proceeds in three steps. We first show that $p^\star$ minimises $\Jcrit$, which also gives the positive lower bound $\Jcrit(p)\ge Z_\gamma^{2/\gamma}$ used to absorb the tail. We then prove the bound on $[0,1]^d$, where it follows from Proposition~\ref{prop:cell-bound} by choosing the level $M$ so that a filling design fits within the budget of $n$ nodes, and by replacing $\Jcrit_M(p)$ with $\Jcrit(p)$ by means of \hypref{hyp:cell-average}. Finally, we adapt the argument to $\mathbb{R}^d$.
 
\begin{lemma}
\label{lem:J-min}
For every probability density $p$ on $\Xset$, $\Jcrit(p)\ge Z_\gamma^{2/\gamma}$; in particular, $\Jcrit(p)=\infty$ if $Z_\gamma=\infty$. If $Z_\gamma<\infty$, equality holds if and only if $p=p^\star$ almost everywhere.
\end{lemma}
 
\begin{proof}
If $\Jcrit(p)=\infty$, there is nothing to prove, so assume $\Jcrit(p)<\infty$. Since $\gamma=2d/(2\alpha+d)$, we have $2\alpha\gamma/d=2-\gamma$.
H\"older's inequality with conjugate exponents $2/\gamma$ and $2/(2-\gamma)$ gives
\begin{equation*}
    Z_\gamma
    =\int_{\Xset}
      \left(p^{-2\alpha/d}\pi^2\right)^{\gamma/2}
      p^{\alpha\gamma/d}\,\dd x
    \le
      \left(\int_{\Xset}p^{-2\alpha/d}\pi^2\,\dd x\right)^{\gamma/2}
      \left(\int_{\Xset}p\,\dd x\right)^{(2-\gamma)/2}
     =\Jcrit(p)^{\gamma/2}.
\end{equation*}
Therefore $\Jcrit(p)\ge Z_\gamma^{2/\gamma}$, with equality if and only if $p^{-2\alpha/d}\pi^2$ and $p$ are proportional, that is, $p^{1+2\alpha/d}\propto\pi^2$, or $p=p^\star$ almost everywhere. Direct substitution gives
\begin{equation*}
    \Jcrit(p^\star)=Z_\gamma^{2\alpha/d}\int_{\Xset}\pi(x)^{2-2\alpha\gamma/d}\,\dd x=Z_\gamma^{1+2\alpha/d}=Z_\gamma^{2/\gamma}.
\end{equation*}
\end{proof}
 
The lower bound does not depend on the design: it identifies the density for which the leading factor $\Jcrit(p)$ is smallest, and the H\"older equality case shows that this density is unique.
 
The proofs below start from the same identity. By \hypref{hyp:k}, the kernel mean embedding $\mu_\Pi$ exists, and the definition of the Bayesian quadrature weights gives $\widehat\Pi_n(f)=\Pi(s_n)$. Hence, with $\zeta_n=f-s_n$,
\begin{equation}
\label{eq:error-identity}
    \Pi(f)-\widehat\Pi_n(f)=\int_\Xset \zeta_n(x)\pi(x)\,\dd x .
\end{equation}
Set $M_n=\lfloor n/(Dm_0)\rfloor$. By \hypref{hyp:cell-average}, there is $n_0$ such that, for $n\ge n_0$, $M_n\ge\max\{2,M_0\}$, $M_n\ge n/(2Dm_0)$ and $\Jcrit_{M_n}(p)\le2\Jcrit(p)$. By \eqref{eq:stop-count}, placing one node in each depth-$\ell_0$ descendant of each cell of $\StopPartition_{M_n}$ uses at most $m_0DM_n\le n$ nodes; the remaining nodes, if any, are placed arbitrarily. This design fills $\StopPartition_{M_n}$.

\begin{proposition}
\label{prop:bounded}
Assume the hypotheses of Theorem~\ref{thm:up-int} in case \textbf{(B)}. For every $n\ge n_0$, the design above satisfies
\begin{equation*}
    \sup_{\lVert f\rVert_{\Hk}\le1}
    \left\lvert\Pi(f)-\widehat\Pi_n(f)\right\rvert
    \le \sqrt2\,c_{\mathrm{loc}}(2Dm_0)^{\alpha/d}\,n^{-\alpha/d}\,\Jcrit(p)^{1/2}.
\end{equation*}
\end{proposition}
 
\begin{proof}
In case \textbf{(B)}, $\Kregion_{M_n}=\Xset$, so \eqref{eq:error-identity} and Proposition~\ref{prop:cell-bound} give, for $\lVert f\rVert_{\Hk}\le1$,
\begin{equation*}
    \left\lvert\Pi(f)-\widehat\Pi_n(f)\right\rvert
    \le c_{\mathrm{loc}}M_n^{-\alpha/d}\Jcrit_{M_n}(p)^{1/2}
    \le \sqrt2\,c_{\mathrm{loc}}(2Dm_0)^{\alpha/d}\,n^{-\alpha/d}\,\Jcrit(p)^{1/2}.
\end{equation*}
\end{proof}
 
On $\mathbb{R}^d$, the region covered by the roots grows with $n$, and the only new term is the contribution of $\Pi$ outside it.
Since the $L^2(\mathbb{R}^d)$-norm of the residual is controlled by its Sobolev norm, one application of the Cauchy--Schwarz inequality suffices, and \hypref{hyp:tail} ensures that this term decays at the same rate as the interior term.

\vspace{\baselineskip}
\begin{proof}[\textbf{Theorem~\ref{thm:up-int}}]
In case \textbf{(B)}, the bound is Proposition~\ref{prop:bounded}, and the last statement is Lemma~\ref{lem:J-min}.
In case \textbf{(U)}, split the integral in \eqref{eq:error-identity} over $\Kregion_{M_n}$ and its complement.
On $\Kregion_{M_n}$, Proposition~\ref{prop:cell-bound} gives the bound $\sqrt2\,c_{\mathrm{loc}}(2Dm_0)^{\alpha/d}n^{-\alpha/d}\Jcrit(p)^{1/2}$, exactly as in the proof of Proposition~\ref{prop:bounded}.
On the complement, the Cauchy--Schwarz inequality, $\lVert \zeta_n\rVert_{L^2(\mathbb{R}^d)}\le\lVert \zeta_n\rVert_{\Ha}$, \eqref{eq:projection-contraction} and \hypref{hyp:tail} give
\begin{equation}
\label{eq:tail-bound-proof}
    \left\lvert\int_{\Kregion_{M}^c}\zeta_n\pi\,\dd x\right\rvert
    \le
    \lVert \zeta_n\rVert_{L^2(\Kregion_{M}^c)}
    \left(\int_{\Kregion_{M}^c}\pi^2\,\dd x\right)^{1/2}
    \le
    C_{\mathrm{eq}}\left(\int_{\Kregion_{M}^c}\pi^2\,\dd x\right)^{1/2}
    \le
    C_{\mathrm{eq}}T^{1/2}M^{-\alpha/d}
\end{equation}
for every $M\ge M_0$. With $M=M_n\ge n/(2Dm_0)$, this is at most $C_{\mathrm{eq}}(2Dm_0)^{\alpha/d}T^{1/2}n^{-\alpha/d}$. By Lemma~\ref{lem:J-min}, \begin{equation*}
T^{1/2}\le\left(\frac{T}{Z_\gamma^{2/\gamma}}\right)^{1/2}\Jcrit(p)^{1/2}.
\end{equation*}
Adding the two bounds proves \eqref{eq:thmup-int} with
\begin{equation*}
    C_0=(2Dm_0)^{\alpha/d}\left\{\sqrt2\,c_{\mathrm{loc}}+C_{\mathrm{eq}}\left(T/Z_\gamma^{2/\gamma}\right)^{1/2}\right\}.
\end{equation*}
\end{proof}

\subsection{Random designs}
\label{app:random}
 
The deterministic designs of Section~\ref{app:thm} place one node in each depth-$\ell_0$ descendant of each cell.
An independent sample need not do so, but it does with high probability if every descendant has mass of order $1/M$ and the sample size exceeds $M$ by a logarithmic factor.
This coupon-collector argument is the source of the factor $\log n$ in Corollary~\ref{cor:up-int}.
The lower bound on the masses of the descendants follows by iterating the doubling condition down a fixed number of generations below a stopping cell, which is allowed because every cube along the way has mass at most $\tau_M$.
 
\begin{lemma}
\label{lem:min_occupancy}
Let $\RootFamily$ be a root family with union $\Kregion$, let $M\ge2$, and suppose that $P(\Kregion)\ge1/2$, that every root has $P$-mass larger than $\tau_M$, and that $P$ is $D$-doubling at level $M$ on $\RootFamily$. Every descendant $\dcube'$ at depth $\ell_0$ of a cell $\dcube\in\StopPartition_M$ satisfies
\begin{equation}
\label{eq:prob-desc-mass}
    P(\dcube')\ge\frac aM,
    \qquad
    a=\frac{1}{2D^{\ell_0+1}}.
\end{equation}
Consequently, if $\bx_{1:n}$ consists of $n$ independent samples from $p$,
\begin{equation*}
    \mathbb{P}\left[\bx_{1:n}\text{ does not fill }\StopPartition_M\right]
    \le
    2^{d\ell_0}DM\,e^{-an/M}.
\end{equation*}
\end{lemma}
 
\begin{proof}
Along the chain from $\dcube$ to $\dcube'$, every cube has mass at most $P(\dcube)\le\tau_M$ and every parent lies in $\DyadicFamily(\RootFamily)$, so \eqref{eq:doubling} applies at each of the $\ell_0$ generations. With \eqref{eq:stop-mass}, this gives
\begin{equation*}
P(\dcube')\ge D^{-\ell_0}P(\dcube)>D^{-\ell_0-1}P(\Kregion)/M\ge a/M.
\end{equation*}
A fixed descendant receives no sample with probability
\begin{equation*}
(1-P(\dcube'))^n\le e^{-nP(\dcube')}\le e^{-an/M}
\end{equation*}
By \eqref{eq:stop-count}, there are at most $2^{d\ell_0}DM$ descendants, and a union bound concludes.
\end{proof}

On the event that the sample fills the stopping partition, the deterministic argument applies verbatim, with $\widetilde M_n$ in place of $M_n$.

\vspace{\baselineskip}
\begin{proof}[\textbf{Corollary~\ref{cor:up-int}}]
Fix $c>0$, set $A=(c+2)/a$ and $\widetilde M_n=\lfloor n/(A\log n)\rfloor$, and let $n_c$ be such that, for $n\ge n_c$, $\widetilde M_n\ge\max\{2,M_0\}$, $\widetilde M_n\ge n/(2A\log n)$, $2^{d\ell_0}D\le n$ and $\Jcrit_{\widetilde M_n}(p)\le2\Jcrit(p)$.
Since $\widetilde M_n\le n/(A\log n)$ and $\widetilde M_n\le n$, Lemma~\ref{lem:min_occupancy} with $\RootFamily=\RootFamily_{\widetilde M_n}$ gives
\begin{equation}
\label{eq:fill-prob}
    \mathbb{P}\left[\bx_{1:n}\text{ does not fill }\StopPartition_{\widetilde M_n}\right]
    \le
    2^{d\ell_0}D\,n\,e^{-aA\log n}
    =2^{d\ell_0}D\,n^{-c-1}
    \le n^{-c}.
\end{equation}
On the complementary event, Proposition~\ref{prop:cell-bound} applies on $\Kregion_{\widetilde M_n}$; nodes outside a cell play no role in it. In case \textbf{(U)}, \eqref{eq:tail-bound-proof} with $M=\widetilde M_n$ bounds the integral over $\Kregion_{\widetilde M_n}^c$. Hence, for $\lVert f\rVert_{\Hk}\le1$,
\begin{equation*}
    \left\lvert\Pi(f)-\widehat\Pi_n(f)\right\rvert
    \le
    \left(\sqrt2\,c_{\mathrm{loc}}\Jcrit(p)^{1/2}+C_{\mathrm{eq}}T^{1/2}\right)\widetilde M_n^{-\alpha/d}
    \le
    C_1\left(\frac{n}{\log n}\right)^{-\alpha/d}\Jcrit(p)^{1/2},
\end{equation*}
with $T=0$ in case \textbf{(B)} and 
\begin{equation*}
C_1=(2A)^{\alpha/d}\left\{\sqrt2\,c_{\mathrm{loc}}+C_{\mathrm{eq}}(T/Z_\gamma^{2/\gamma})^{1/2}\right\},
\end{equation*}
where we used Lemma~\ref{lem:J-min} as in the proof of Theorem~\ref{thm:up-int}.
\end{proof}

Posterior contraction requires no further geometric argument.
By \eqref{eq:var-worst-case}, the posterior standard deviation equals the worst-case error, and the posterior mean is within the same distance of the true integral, up to the factor $\lVert f\rVert_{\Hk}$.
The high-probability bound of Corollary~\ref{cor:up-int} therefore transfers directly to the posterior, and the dependence on $\Jcrit(p)$ is inherited through the worst-case error.

\vspace{\baselineskip}
\begin{proof}[\textbf{Corollary~\ref{cor:postvar-main}}]
Let
\begin{equation*}
\eta_n=V_n^{1/2}=\lVert\mu_\Pi-\mu_n\rVert_{\Hk}=\sup_{\lVert g\rVert_{\Hk}\le1}\lvert\Pi(g)-\widehat\Pi_n(g)\rvert.
\end{equation*}
By \eqref{eq:var-worst-case}, $\eta_n^2$ is the posterior variance of $\Pi(f^\star)$ given $\bx_{1:n}$, and $\lvert\widehat\Pi_n(f)-\Pi(f)\rvert\le\eta_n\lVert f\rVert_{\Hk}$.
By Corollary~\ref{cor:up-int}, for every $c>0$ there is $C_1$ and $n_c$ such that, for $n\ge n_c$, with probability at least $1-n^{-c}$,
\begin{equation}
\label{eq:wce-rate-contraction}
    \eta_n\le C_1\,\Jcrit(p)^{1/2}\varepsilon_n,
    \qquad
    \varepsilon_n=\left(\frac{n}{\log n}\right)^{-\alpha/d}.
\end{equation}
On this event, the conditional posterior is $\mathcal{N}(\widehat\Pi_n(f),\eta_n^2)$. Increasing $n_c$ to some $n_{c,\delta}$ so that $C_1\Jcrit(p)^{1/2}\varepsilon_n\lVert f\rVert_{\Hk}\le\delta/2$, Gaussian tail bounds give
\begin{equation*}
    \mathbb{P}\!\left[
        \left\lvert\Pi(f^\star)-\Pi(f)\right\rvert\ge\delta
        \,\middle|\,\bx_{1:n}
    \right]
    \le
    \mathbb{P}\!\left[\lvert\mathcal{N}(0,\eta_n^2)\rvert\ge\delta/2\right]
    \le
    2\exp\!\left\{-\frac{\delta^2}{8C_1^2\Jcrit(p)\varepsilon_n^2}\right\},
\end{equation*}
which is \eqref{eq:posterior-contraction-main} with $C_2=1/(8C_1^2)$; see also Lemma~3 of \citet{briol2019probabilistic}.
\end{proof}

\subsection{Proofs for the marginal posterior variance}

The marginal variance adds to the expected conditional posterior variance the variability of the posterior mean across designs.
Since the error of the posterior mean is bounded by the worst-case error, which is the posterior standard deviation, the second term is controlled by the first.
The rate then follows from the high-probability bound of Corollary~\ref{cor:up-int}.

\vspace{\baselineskip}
\begin{proof}[\textbf{Proposition~\ref{prop:variance-control}}]
By \eqref{eq:var-worst-case},
\begin{equation*}
V_n=\lVert\mu_\Pi-\mu_n\rVert_{\Hk}^2,
\qquad
\widehat\Pi_n(f)-\Pi(f)=-\langle f,\mu_\Pi-\mu_n\rangle_{\Hk}.
\end{equation*}
Since $\Pi(f)$ does not depend on the design, the variance of $\widehat\Pi_n(f)$ is at most its mean squared distance to $\Pi(f)$, and the Cauchy--Schwarz inequality gives
\begin{equation*}
    \var_{\bx_{1:n}}\!\left[\widehat\Pi_n(f)\right]
    \le
    \mathbb{E}_{\bx_{1:n}}
    \left[\langle f,\mu_\Pi-\mu_n\rangle_{\Hk}^2\right]
    \le
    \lVert f\rVert_{\Hk}^2\,
    \mathbb{E}_{\bx_{1:n}}[V_n].
\end{equation*}
Substitution into \eqref{eq:var-decomposition} proves \eqref{eq:marginal-variance-control}.
\end{proof}

Corollary~\ref{cor:up-int} bounds the worst-case error outside an event of polynomially small probability, while \hypref{hyp:k} bounds the posterior variance everywhere. Choosing the failure exponent large enough makes the exceptional event negligible.

\vspace{\baselineskip}
\begin{proof}[\textbf{Corollary~\ref{cor:marginal-postvar}}]
Write $\varepsilon_n=(n/\log n)^{-\alpha/d}$, and recall that $\Jcrit(p)\ge Z_\gamma^{2/\gamma}>0$ by Lemma~\ref{lem:J-min}.
Apply Corollary~\ref{cor:up-int} with the failure exponent $c=2\alpha/d+1$, and let $C_1$ be the resulting constant.
By \eqref{eq:var-worst-case}, $V_n^{1/2}$ is the worst-case error, so $V_n\le C_1^2\Jcrit(p)\varepsilon_n^2$ outside an event of probability at most $n^{-c}$.
Moreover, conditioning cannot increase the prior variance, so \hypref{hyp:k} gives
\begin{equation*}
    0\le V_n\le\Pi^{\otimes2}(k)\le C_k.
\end{equation*}
Hence
\begin{equation*}
    \mathbb{E}_{\bx_{1:n}}[V_n]
    \le C_1^2\Jcrit(p)\varepsilon_n^2+C_kn^{-c}
    \le C_3\Jcrit(p)\varepsilon_n^2
\end{equation*}
for all $n\ge n_3$, where $C_3=2C_1^2$ and $n_3$ is chosen large enough to absorb the exceptional-event contribution, since $n^{-c}=o(\varepsilon_n^2)$ by the choice of $c$.
Proposition~\ref{prop:variance-control} then proves \eqref{eq:marginal-postvar-rate}.
\end{proof}

\subsection{Densities proportional to \texorpdfstring{$\pi^\gamma$}{pi gamma}}
\label{app:power}
 
Theorem~\ref{thm:up-int} assumes \hypref{hyp:cell-average}, namely $\Jcrit(p)<\infty$ and $\Jcrit_M(p)\to\Jcrit(p)$ as $M\to\infty$. This subsection provides a tool to verify it for the targets of Propositions~\ref{prop:cube-beta}--\ref{prop:student-target}, as well as a fixed-level bound that replaces it when $\pi^\gamma$ is not integrable.
In all these examples, the design density is proportional to $\pi^\gamma$: it is $p^\star=\pi^\gamma/Z_\gamma$ when $Z_\gamma<\infty$, and the normalised restriction $p_M=\pi^\gamma\mathds{1}_{\Kregion_M}/Z_M$ for the heavy-tailed Student target.
 
For such densities, each term of $\Jcrit_M(p)$ involves only the two averages $\langle\pi\rangle_\dcube$ and $\langle\pi^\gamma\rangle_\dcube$. By Jensen's inequality, $\langle\pi^\gamma\rangle_\dcube\le\langle\pi\rangle_\dcube^\gamma$, since $\gamma<1$. The reverse inequality, up to a constant,
\begin{equation}
\label{eq:RH}
    \langle\pi\rangle_\dcube^{\gamma}
    \le
    C_{\mathrm{RH}}\,\langle\pi^\gamma\rangle_\dcube,
\end{equation}
is a reverse H\"older inequality. It states that $\pi$ does not vary too much on $\dcube$; for instance, it holds when $\sup_\dcube\pi^\gamma\le C\inf_\dcube\pi^\gamma$. When \eqref{eq:RH} holds on every stopping cell, each term of $\Jcrit_M(p)$ is bounded by a multiple of 
\begin{equation*}
\operatorname{vol}(\dcube)\langle\pi^\gamma\rangle_\dcube=\int_{\dcube}\pi(x)^\gamma\,\dd x,
\end{equation*}
and summing over the cells gives a bound by the normalising constant. This is Lemma~\ref{lem:J-finite}.
If, moreover, $Z_\gamma<\infty$, the same domination, combined with the shrinking of the cells, yields the convergence $\Jcrit_M(p^\star)\to\Jcrit(p^\star)$ as $M\to\infty$, which is Lemma~\ref{lem:J-limit}.
 
These lemmas are used in Sections~\ref{app:beta} and~\ref{app:radial} as follows. For the Beta-type densities, the Jacobi-weight estimates give the doubling condition and \eqref{eq:RH} on every dyadic subcube, and Lemma~\ref{lem:J-limit} then gives \hypref{hyp:cell-average}. For the radial targets, Lemma~\ref{lem:radial-check} verifies both the doubling condition and \eqref{eq:RH} for radial densities whose profile oscillates boundedly at the scale of the stopping cells. Lemma~\ref{lem:J-limit} then gives \hypref{hyp:cell-average} for the Gaussian target and for the Student target with $\alpha<\nu+d/2$. When $\alpha\ge\nu+d/2$, so that $Z_\gamma=\infty$, Theorem~\ref{thm:up-int} does not apply, and the proof of Proposition~\ref{prop:student-target}(b) combines Proposition~\ref{prop:cell-bound} at each level $M$ with the bound of Lemma~\ref{lem:J-finite} for $Z=Z_M$.
 
\begin{lemma}
\label{lem:J-finite}
Fix $M\ge2$ and a root family $\RootFamily_M$ with union $\Kregion_M$ whose roots have $P$-mass larger than $\tau_M$, and suppose that $P$ is $D$-doubling at level $M$ on $\RootFamily_M$. Suppose also that $p=Z^{-1}\pi^\gamma$ on $\Kregion_M$ for some $Z\ge Z_M=\int_{\Kregion_M}\pi(x)^\gamma\,\dd x$, and that \eqref{eq:RH} holds for every $\dcube\in\StopPartition_M(\RootFamily_M,P)$, for some $C_{\mathrm{RH}}\ge1$. Then the first bound below holds on $\Kregion_M$
\begin{equation}
\label{eq:J-domination}
    (T_Mp)^{-2\alpha/d}(T_M\pi)^2
    \le
    C_{\mathrm{RH}}^{2/\gamma}Z^{2\alpha/d}\,T_M(\pi^\gamma),
    \qquad\text{and}\qquad
    \Jcrit_M(p)\le(C_{\mathrm{RH}}Z)^{2/\gamma}.
\end{equation}
\end{lemma}
 
\begin{proof}
Recall that $\gamma=2d/(2\alpha+d)$, so that $2/\gamma-2\alpha/d=1$. Fix a cell $\dcube\in\StopPartition_M$. On $\dcube$, the operator $T_M$ takes the constant values $T_Mp=\langle p\rangle_\dcube=Z^{-1}\langle\pi^\gamma\rangle_\dcube$, $T_M\pi=\langle\pi\rangle_\dcube$ and $T_M(\pi^\gamma)=\langle\pi^\gamma\rangle_\dcube$. Raising \eqref{eq:RH} to the power $2/\gamma$ gives $\langle\pi\rangle_\dcube^2\le C_{\mathrm{RH}}^{2/\gamma}\langle\pi^\gamma\rangle_\dcube^{2/\gamma}$, and therefore, on $\dcube$,
\begin{equation*}
    (T_Mp)^{-2\alpha/d}(T_M\pi)^2
    =Z^{2\alpha/d}\langle\pi^\gamma\rangle_\dcube^{-2\alpha/d}\langle\pi\rangle_\dcube^2
    \le C_{\mathrm{RH}}^{2/\gamma}Z^{2\alpha/d}\langle\pi^\gamma\rangle_\dcube^{2/\gamma-2\alpha/d}
    =C_{\mathrm{RH}}^{2/\gamma}Z^{2\alpha/d}\langle\pi^\gamma\rangle_\dcube .
\end{equation*}
Since the cells partition $\Kregion_M$, this is the first inequality in \eqref{eq:J-domination}. Integrating it over $\Kregion_M$ and using that $T_M$ preserves integrals on $\Kregion_M$,
\begin{equation*}
    \Jcrit_M(p)
    \le C_{\mathrm{RH}}^{2/\gamma}Z^{2\alpha/d}\sum_{\dcube\in\StopPartition_M}\operatorname{vol}(\dcube)\langle\pi^\gamma\rangle_\dcube
    =C_{\mathrm{RH}}^{2/\gamma}Z^{2\alpha/d}Z_M
    \le C_{\mathrm{RH}}^{2/\gamma}Z^{1+2\alpha/d}
    =(C_{\mathrm{RH}}Z)^{2/\gamma}.
\end{equation*}
\end{proof}
 
The bound of Lemma~\ref{lem:J-finite} holds at a fixed level and does not require $Z_\gamma<\infty$. When $\pi^\gamma$ is not integrable, it is applied with $Z=Z_M$, which grows with $M$, and this growth is what limits the rate in Proposition~\ref{prop:student-target}(b). When $Z_\gamma<\infty$, we take $Z=Z_\gamma$, and the bound is uniform in $M$. The domination in \eqref{eq:J-domination} is then also the key to the convergence of $\Jcrit_M(p^\star)$: the dominating functions $T_M(\pi^\gamma)$ converge to $\pi^\gamma$, together with their integrals as $M\to\infty$, so that a dominated convergence argument applies even though the dominating function changes with $M$.
 
\begin{lemma}
\label{lem:J-limit}
Assume \hypref{hyp:admissible-design} for $p=p^\star$, with $Z_\gamma<\infty$ and $\pi>0$ almost everywhere, and suppose that \eqref{eq:RH} holds for every $\dcube\in\StopPartition_M(\RootFamily_M,P)$ and every $M\ge M_0$, with $C_{\mathrm{RH}}$ independent of $M$. Then \hypref{hyp:cell-average} holds for $p^\star$, with $\Jcrit(p^\star)=Z_\gamma^{2/\gamma}$.
\end{lemma}
 
\begin{proof}
Define, for $M\ge M_0$,
\begin{equation*}
    g_M=(T_Mp^\star)^{-2\alpha/d}(T_M\pi)^2\mathds{1}_{\Kregion_M},
    \qquad
    g=(p^\star)^{-2\alpha/d}\pi^2=Z_\gamma^{2\alpha/d}\pi^\gamma,
\end{equation*}
so that $\Jcrit_M(p^\star)=\int_{\Xset} g_M(x)\,\dd x$ and
\begin{equation*}
\Jcrit(p^\star)=\int_{\Xset} g(x)\,\dd x=Z_\gamma^{1+2\alpha/d}=Z_\gamma^{2/\gamma}<\infty.
\end{equation*}
We show that $\int_{\Xset} g_M(x)\,\dd x\to\int_{\Xset} g(x)\,\dd x$ as $M\to\infty$ in three steps.
 
First, $g_M\to g$ almost everywhere as $M\to\infty$. By Lemma~\ref{lem:shrink}, applied to $h=p^\star$ and $h=\pi$, we have $T_Mp^\star\to p^\star$ and $T_M\pi\to\pi$ almost everywhere as $M\to\infty$, and every $x$ lies in $\Kregion_M$ for all large $M$. Since $p^\star>0$ almost everywhere, the map $(s,t)\mapsto s^{-2\alpha/d}t^2$ is continuous at $(p^\star(x),\pi(x))$ for almost every $x$, which gives the claim.
 
Second, $g_M$ is dominated by a convergent sequence. Lemma~\ref{lem:J-finite} with $Z=Z_\gamma$ gives
\begin{equation*}
    0\le g_M\le h_M=C_{\mathrm{RH}}^{2/\gamma}Z_\gamma^{2\alpha/d}\,T_M(\pi^\gamma),
\end{equation*}
where $T_M(\pi^\gamma)$ vanishes off $\Kregion_M$. By Lemma~\ref{lem:shrink}, 
\begin{equation*}
h_M\xrightarrow[M\to\infty]{} h=C_{\mathrm{RH}}^{2/\gamma}Z_\gamma^{2\alpha/d}\pi^\gamma   
\end{equation*}
almost everywhere, and, since $T_M$ preserves integrals on $\Kregion_M$,
\begin{equation*}
    \int_{\Xset} h_M(x)\,\dd x=C_{\mathrm{RH}}^{2/\gamma}Z_\gamma^{2\alpha/d}\int_{\Kregion_M}\pi(x)^\gamma\,\dd x
    \xrightarrow[M\to\infty]{} C_{\mathrm{RH}}^{2/\gamma}Z_\gamma^{2\alpha/d}Z_\gamma=\int_{\Xset} h(x)\,\dd x<\infty,
\end{equation*}
by monotone convergence, because the sets $\Kregion_M$ are non-decreasing with union $\Xset$.
 
Third, we conclude by Fatou's lemma, applied twice; the ordinary dominated convergence theorem does not apply directly, because the dominating function $h_M$ depends on $M$. Fatou's lemma applied to $g_M\ge0$ gives  
\begin{equation*}
\int_{\Xset}g(x)\,\dd x\le\liminf_{M\to\infty}\int_{\Xset}g_M(x)\,\dd x.
\end{equation*}
Fatou's lemma applied to $h_M-g_M\ge0$ gives
\begin{equation*}
    \int_{\Xset}h(x)\,\dd x-\int_{\Xset}g(x)\,\dd x
    \le\liminf_{M\to\infty}\int_{\Xset}(h_M(x)-g_M(x))\,\dd x
    =\int_{\Xset}h(x)\,\dd x-\limsup_{M\to\infty}\int_{\Xset}g_M(x)\,\dd x.
\end{equation*} 
where we used   $\int_{\Xset}h_M(x)\,\dd x\xrightarrow[M\to\infty]{}\int_{\Xset}h(x)\,\dd x<\infty$. Hence 
\begin{equation*}
\limsup_{M\to\infty}\int_{\Xset}g_M(x)\,\dd x\le\int_{\Xset}g(x)\,\dd x,
\end{equation*}
and therefore  
\begin{equation*}
\Jcrit_M(p^\star)=\int_{\Xset}g_M(x)\,\dd x\xrightarrow{M\to\infty}\int_{\Xset}g(x)\,\dd x=\Jcrit(p^\star).
\end{equation*} 
\end{proof}

\subsection{Proof of Proposition~\ref{prop:cube-beta}}
\label{app:beta}

On $[0,1]^d$, the verification reduces to a two-sided estimate for averages of one-dimensional Jacobi weights. It gives the doubling condition and the reverse H\"older inequality \eqref{eq:RH} on every dyadic subcube, so that \hypref{hyp:cell-average} follows from Lemma~\ref{lem:J-limit}, exactly as for the radial targets.
 
For $\theta,\eta>-1$, let $j_{\theta,\eta}(x)=x^\theta(1-x)^\eta$ on $(0,1)$, a shifted Jacobi weight, and, for an interval $\JacobiInterval=[s,s+\ell]\subseteq[0,1]$, put $m_\JacobiInterval=(s+\ell)^\theta(1-s)^\eta$; see \citet{Gautschi2004} for an introduction to Jacobi weights. Write $\theta_+=\max\{\theta,0\}$ and $\theta_-=\min\{\theta,0\}$. The change of variables $x=s+\ell t$ gives
\begin{equation*}
    \frac{\langle j_{\theta,\eta}\rangle_\JacobiInterval}{m_\JacobiInterval}
    =\int_0^1
    \left(\frac{s+\ell t}{s+\ell}\right)^\theta
    \left(\frac{1-s-\ell t}{1-s}\right)^\eta
    \,\dd t,
\end{equation*}
where the two fractions lie in $[t,1]$ and $[1-t,1]$, respectively, since $\ell\le1-s$. Bounding them below by $1/4$ for $t\in[1/4,3/4]$, and by $t$ and $1-t$ on $[0,1]$, gives
\begin{equation}
\label{eq:jacobi-average}
    c_{\theta,\eta}\,m_\JacobiInterval
    \le\langle j_{\theta,\eta}\rangle_\JacobiInterval
    \le C_{\theta,\eta}\,m_\JacobiInterval,
    \qquad
    c_{\theta,\eta}=\tfrac12\,4^{-\theta_+-\eta_+},
    \quad
    C_{\theta,\eta}=\mathrm{B}\left(1+\theta_-,1+\eta_-\right),
\end{equation}
where $\mathrm{B}$ is the beta function. Two consequences follow.
 
First, $j_{\theta,\eta}\,\dd x$ is doubling. The parent of $\JacobiInterval$ is $\,\widehat{\JacobiInterval}=[s',s'+2\ell]$ with $s'\in\{s-\ell,s\}$, and the ratios $(s'+2\ell)/(s+\ell)$ and $(1-s')/(1-s)$ lie in $[1,2]$, so $m_{\,\widehat{\JacobiInterval}}\le2^{\lvert\theta\rvert+\lvert\eta\rvert}m_\JacobiInterval$. By \eqref{eq:jacobi-average},
\begin{equation*}
    \int_{\,\widehat{\JacobiInterval}}j_{\theta,\eta}\le2^{1+\lvert\theta\rvert+\lvert\eta\rvert}\frac{C_{\theta,\eta}}{c_{\theta,\eta}}\int_{\JacobiInterval} j_{\theta,\eta}.
\end{equation*}
 
Second, $j_{\theta,\eta}$ satisfies a reverse H\"older inequality. If $q>1$ and $\theta,\eta>-1/q$, then $j_{\theta,\eta}^q=j_{q\theta,q\eta}$ with $q\theta,q\eta>-1$, and applying \eqref{eq:jacobi-average} to both weights gives
\begin{equation*}
    \langle j_{\theta,\eta}^q\rangle_\JacobiInterval^{1/q}
    \le C_{q\theta,q\eta}^{1/q}\,m_\JacobiInterval
    \le C_{q\theta,q\eta}^{1/q}c_{\theta,\eta}^{-1}\,\langle j_{\theta,\eta}\rangle_\JacobiInterval .
\end{equation*}
On $[0,1]^d$, averages of products over dyadic cubes factor coordinatewise, so both properties hold for products of Jacobi weights on dyadic subcubes of $[0,1]^d$; they are preserved under multiplication by a function bounded above and away from zero.

\begin{proof}[\textbf{Proposition~\ref{prop:cube-beta}}]
Set $\theta_j=\gamma(a_j-1)$, $\eta_j=\gamma(b_j-1)$ and $q=1/\gamma>1$. Then
\begin{equation*}
    w=\pi^\gamma
    =c^\gamma\varphi^\gamma\prod_{j=1}^d j_{\theta_j,\eta_j}(x_j),
    \qquad
    m_\varphi^\gamma\le\varphi^\gamma\le M_\varphi^\gamma,
\end{equation*}
so that $p^\star=w/Z_\gamma$ has the stated form, with $\widetilde a_j=1+\theta_j$ and $\widetilde b_j=1+\eta_j$, and $Z_\gamma\le1$ by Jensen's inequality.
Since $a_j,b_j>0$, we have $\theta_j,\eta_j>-\gamma=-1/q$.
By the two consequences of \eqref{eq:jacobi-average}, $w\,\dd x$ is doubling on every dyadic subcube of $[0,1]^d$, which gives \hypref{hyp:admissible-design}\textbf{(B)} with $M_0=2$, since $P^\star([0,1]^d)=1$; and $\langle w^{q}\rangle_\dcube^{1/q}\le C\langle w\rangle_\dcube$ for every dyadic cube $\dcube$. As $w^q=\pi$, the latter is \eqref{eq:RH} with $C_{\mathrm{RH}}=C$, independent of $M$.
Since moreover $\pi>0$ almost everywhere, Lemma~\ref{lem:J-limit} gives \hypref{hyp:cell-average}, with $\Jcrit(p^\star)=Z_\gamma^{2/\gamma}$.
\end{proof}

\subsection{Radial densities on \texorpdfstring{$\mathbb{R}^d$}{Rd}}
\label{app:radial}
 
Gaussian and Student targets are radial, and their stopping cells are large far from the origin, where the density is small.
Both the doubling condition and \eqref{eq:RH} hold as soon as the density varies by at most a constant factor over each cell and its parent.
Since the side of a stopping cell at radius $r$ is at most $\eta_M(r)=\{M\varrho(r)\}^{-1/d}$, this reduces to a bound on the oscillation of $\log\varrho$ over that scale, where $\varrho$ is the radial profile of $p$.
 
\begin{lemma}
\label{lem:radial-check}
Let $\RootFamily$ be a root family whose union $\Kregion$ lies in the closed ball $\overline \EuclideanBall(\zero_d,r_{\max})$, let $M\ge2$, and let $p$ be a probability density with $P(\Kregion)\ge1/2$ and $p(x)=\varrho(\lVert x\rVert_2)$ on $\Kregion$, where $\varrho$ is positive and non-increasing on $[0,r_{\max}]$. Set
\begin{equation*}
    \eta_M(r)=\{M\varrho(r)\}^{-1/d},
    \qquad 0\le r\le r_{\max},
\end{equation*}
and suppose that, for some $B<\infty$,
\begin{equation}
\label{eq:radial-local-osc}
    \left\lvert\log\varrho(s)-\log\varrho(r)\right\rvert\le B,
\end{equation}
whenever $0\le r,s\le r_{\max}$ and $\lvert s-r\rvert\le2\sqrt d\,\eta_M(r)$.
Then, for every $\dcube\in\DyadicFamily(\RootFamily)\setminus\RootFamily$ with $P(\dcube)\le\tau_M$,
\begin{equation}
\label{eq:radial-cell-osc}
    \sup_{\widehat \dcube}p\le e^{2B}\inf_{\widehat \dcube}p .
\end{equation}
Consequently, $P$ is $D$-doubling at level $M$ on $\RootFamily$ with $D=2^de^{2B}$. If, moreover, every root has mass larger than $\tau_M$ and $\pi^\gamma=Zp$ on $\Kregion$ for some $Z>0$, then \eqref{eq:RH} holds on $\StopPartition_M(\RootFamily,P)$ with $C_{\mathrm{RH}}=e^{2B}$.
\end{lemma}
 
\begin{proof}
Let $r_\dcube=\max_{x\in \dcube}\lVert x\rVert_2$. Since $\varrho$ is non-increasing, $p\ge\varrho(r_\dcube)$ on $\dcube$, and $P(\dcube)\le\tau_M\le1/M$ gives $\varrho(r_\dcube)\operatorname{vol}(\dcube)\le1/M$. Hence the side of $\dcube$ is at most $\eta_M(r_\dcube)$, and $\operatorname{diam}(\widehat \dcube)\le2\sqrt d\,\eta_M(r_\dcube)$.
The maximum $r_\dcube$ is attained at a point of $\dcube\subset\widehat \dcube$, so, by the triangle inequality,
\begin{equation*}
    \lvert\lVert x\rVert_2-r_\dcube\rvert\le\operatorname{diam}(\widehat \dcube)\le2\sqrt d\,\eta_M(r_\dcube),\qquad x\in\widehat \dcube.
\end{equation*}
Moreover, $\widehat \dcube\subseteq \Kregion$, so $\lVert x\rVert_2\le r_{\max}$.
Condition \eqref{eq:radial-local-osc} then gives $\lvert\log p(x)-\log\varrho(r_\dcube)\rvert\le B$ on $\widehat \dcube$, which implies \eqref{eq:radial-cell-osc}.
Since $\operatorname{vol}(\widehat \dcube)=2^d\operatorname{vol}(\dcube)$,
\begin{equation*}
    P(\widehat \dcube)
    \le\operatorname{vol}(\widehat \dcube)\sup_{\widehat \dcube}p
    \le2^de^{2B}\operatorname{vol}(\dcube)\inf_\dcube p
    \le2^de^{2B}P(\dcube).
\end{equation*}
Finally, when the roots have mass larger than $\tau_M$, every $\dcube\in\StopPartition_M$ is a strict descendant of a root with $P(\dcube)\le\tau_M$, so \eqref{eq:radial-cell-osc} holds on $\dcube\subset\widehat \dcube$. As $\pi=(Zp)^{1/\gamma}$ on $\Kregion$,
\begin{equation*}
    \langle\pi\rangle_\dcube^\gamma
    \le\left(\sup_\dcube\pi\right)^\gamma
    =Z\sup_{\dcube}p
    \le e^{2B}Z\inf_{\dcube}p
    \le e^{2B}\langle\pi^\gamma\rangle_\dcube ,
\end{equation*}
which is \eqref{eq:RH}.
\end{proof}
 
Each of the two propositions then amounts to three computations: the local scale $\eta_M$ of the profile, a choice of roots for which this scale is small compared with the scale on which $\log\varrho$ varies, and the tail of $\pi^2$ outside the roots.
For the Gaussian, $\log\varrho$ varies on the scale $1/r$, so the roots must stay close to the ball $\EuclideanBall(\zero_d,\rad_M)$; this is why unit cubes are used.

\vspace{\baselineskip}
\begin{proof}[\textbf{Proposition~\ref{prop:gauss-target}}]
Here $p^\star(x)=\varrho(\lVert x\rVert_2)$ with $\varrho(r)=(\gamma/2\pi)^{d/2}e^{-\gamma r^2/2}$, so
\begin{equation*}
    \eta_M(r)=(2\pi/\gamma)^{1/2}M^{-1/d}e^{\gamma r^2/(2d)},
\end{equation*}
which is increasing in $r$.
Every root of $\RootFamily_M$ meets $\EuclideanBall(\zero_d,\rad_M)$ and has diameter $\sqrt d$, so $\EuclideanBall(\zero_d,\rad_M)\subseteq \Kregion_M\subseteq\overline \EuclideanBall(\zero_d,\radext_M)$ with $\radext_M=\rad_M+\sqrt d$.
The definition of $\rad_M$ gives
\begin{equation*}
e^{\gamma \rad_M^2/2}=M^{\gamma\alpha/d}(\log M)^{\gamma d/4},
\end{equation*}
and $\gamma\alpha/d=2\alpha/(2\alpha+d)$. Since $\radext_M^2=\rad_M^2+2\sqrt d\rad_M+d$ and, for $M\ge3$, $\rad_M\le C_R(\log M)^{1/2}$ with $C_R=(2\alpha/d+d/2)^{1/2}$,
\begin{equation*}
    e^{\gamma \radext_M^2/2}
    =M^{2\alpha/(2\alpha+d)}e^{\mathcal{O}((\log M)^{1/2})}
    =M^{2\alpha/(2\alpha+d)+o(1)} .
\end{equation*}
We first check the roots. Every root lies in $\overline \EuclideanBall(\zero_d,\radext_M)$ and has unit volume, so its $p^\star$-mass is at least $\varrho(\radext_M)=(\gamma/2\pi)^{d/2}e^{-\gamma \radext_M^2/2}=M^{-2\alpha/(2\alpha+d)+o(1)}$, which exceeds $1/M\ge\tau_M$ for large $M$, because $2\alpha/(2\alpha+d)<1$.
We next check the oscillation condition. Put $\delta_M=2\sqrt d\,\eta_M(\radext_M)$. Then
\begin{equation*}
    \delta_M
    =2\sqrt d\,(2\pi/\gamma)^{1/2}M^{-1/d}e^{\gamma \radext_M^2/(2d)}
    =M^{-1/(2\alpha+d)+o(1)}
    \xrightarrow[M\to\infty]{}0 ,
\end{equation*}
since $1/d-2\alpha/\{d(2\alpha+d)\}=1/(2\alpha+d)$. For $0\le r,s\le \radext_M$ with $\lvert s-r\rvert\le2\sqrt d\,\eta_M(r)\le\delta_M$,
\begin{equation*}
    \lvert\log\varrho(s)-\log\varrho(r)\rvert
    =\frac\gamma2\lvert s^2-r^2\rvert
    \le\gamma\delta_M\left(\radext_M+\delta_M\right)
    \xrightarrow[M\to\infty]{}0,
\end{equation*}
because $\radext_M\le C_R(\log M)^{1/2}+\sqrt d$. Thus \eqref{eq:radial-local-osc} holds with $B=1$ for $M$ large.
Moreover, 
\begin{equation*}
P^\star(\Kregion_M)\ge P^\star(\EuclideanBall(\zero_d,\rad_M))\xrightarrow[M\to\infty]{}1,
\end{equation*}
and $\Kregion_M$ is non-decreasing with union $\mathbb{R}^d$.
Choosing $M_0$ such that these properties hold for $M\ge M_0$, Lemma~\ref{lem:radial-check} with $r_{\max}=\radext_M$ gives \hypref{hyp:admissible-design}\textbf{(U)} with $D=2^de^2$ and \eqref{eq:RH} with $C_{\mathrm{RH}}=e^2$, since $\pi^\gamma=Z_\gamma p^\star$.
Finally, for the tail, since $\Kregion_M^c\subseteq \EuclideanBall(\zero_d,\rad_M)^c$ and $\pi(x)^2=(2\pi)^{-d}e^{-\lVert x\rVert_2^2}$, the Gaussian tail estimate gives, for $\rad_M\ge1$,
\begin{equation*}
    \int_{\Kregion_M^c}\pi(x)^2\,\dd x
    \le C \rad_M^{d-2}e^{-\rad_M^2}
    =\mathcal{O}\!\left(M^{-2\alpha/d}(\log M)^{-1}\right),
\end{equation*}
which gives \hypref{hyp:tail}. Since $\pi>0$, Lemma~\ref{lem:J-limit} gives \hypref{hyp:cell-average}.
 
For $\Pi=\mathcal{N}(\mu,\bSigma)$, write $\mathcal T_{\bSigma}(y)=\mu+\bSigma^{1/2}y$ and $\Pi_0=\mathcal{N}(\zero_d,\id_d)$. If the nodes are $x_i=\mathcal T_{\bSigma}(y_i)$, then $\zeta_n\circ\mathcal T_{\bSigma}$ vanishes at the $y_i$, and 
\begin{equation*}
\int_{\mathbb{R}^d}\zeta_n(x)\,\Pi(\dd x)=\int_{\mathbb{R}^d}\zeta_n(\mathcal T_{\bSigma}(y))\,\Pi_0(\dd y).
\end{equation*}
Moreover,
\begin{equation*}
    \lVert \zeta_n\circ\mathcal T_{\bSigma}\rVert_{\mathcal{H}^\alpha(\mathbb{R}^d)}\le C_{\bSigma}\lVert \zeta_n\rVert_{\mathcal{H}^\alpha(\mathbb{R}^d)}.
\end{equation*} 
The proofs of Proposition~\ref{prop:cell-bound} and Theorem~\ref{thm:up-int} use $\zeta_n$ only through these properties, with $C_{\mathrm{eq}}$ replaced by $C_{\bSigma} C_{\mathrm{eq}}$, and therefore apply in the $y$-coordinates, where the target is   $\mathcal{N}(\zero_d,\id_d)$ and the optimal proposal is $\mathcal{N}(\zero_d,\gamma^{-1}\id_d)$. In the original coordinates, this proposal is $\mathcal{N}(\mu,\gamma^{-1}\bSigma)$.
Writing  $a=\det(\bSigma)^{1/2}$ and denoting the design densities in the two coordinate systems by $p_x$ and $p_y$, the change of variables gives $\Jcrit\left(p_x\right)=a^{2\alpha/d-1}\Jcrit\left(p_y\right)$, with the criterion in each coordinate system computed for its corresponding target. This factor is absorbed into the constant. The high-probability bound in Corollary~\ref{cor:up-int} then yields Corollary~\ref{cor:marginal-postvar} by the same expectation argument, since $V_n\leq C_k$ also holds in the original coordinates.
\end{proof}
 
The Student profile varies on the scale $1+r$, so a single cube suffices.
The new feature is that the normalising constant of $\pi^\gamma$ on $\Kregion_M$ may diverge with $M$, which produces the three regimes of Proposition~\ref{prop:student-target}.

\vspace{\baselineskip}
\begin{proof}[\textbf{Proposition~\ref{prop:student-target}}]
The target satisfies $\pi(x)=c_\nu(1+\lVert x\rVert_2^2/\nu)^{-(\nu+d)/2}$, so $\pi^\gamma\propto(1+\lVert x\rVert_2^2/\nu)^{-\beta}$ with $\beta=\gamma(\nu+d)/2$. A direct computation gives
\begin{equation}
\label{eq:student-beta}
    2\beta-d=\frac{2d(\nu+d/2-\alpha)}{2\alpha+d},
\end{equation}
so that $Z_\gamma<\infty$ if and only if $\alpha<\nu+d/2$. In that case, matching exponents with the $t_{\nu^\star}(\zero_d,s\id_d)$ density, $\nu^\star+d=2\beta$ and $\nu^\star s=\nu$, identifies $p^\star$.
 
Let $P_M=p_M\dd x$ denote the design measure. In (a), $p_M=p^\star$ on $\Xset$ and $P_M(\Kregion_M)\rightarrow 1$ as $\rad_M\rightarrow\infty$. In (b), $p_M=\pi^\gamma\mathds{1}_{\Kregion_M}/Z_M$ is supported on $\Kregion_M$, so that $P_M(\Kregion_M)=1$. Here $\RootFamily_M=\{\Kregion_M\}$ is a single cube, so the condition on the root is automatic once $P_M(\Kregion_M)\ge1/2$, which holds for $M$ large enough in (a) and for every $M$ in (b).
Since $\Kregion_M\subseteq\overline \EuclideanBall(\zero_d,\sqrt d\rad_M)$, polar coordinates give, for $\rad_M\ge2$ and a constant $C_Z$ depending only on $d$, $\nu$ and $\beta$,
\begin{equation}
\label{eq:student-normalization}
    Z_M\le
    \begin{cases}
        Z_\gamma,&\alpha<\nu+d/2,\\
        C_Z\log \rad_M,&\alpha=\nu+d/2,\\
        C_Z\rad_M^{d-2\beta},&\alpha>\nu+d/2.
    \end{cases}
\end{equation}
The radial profile of $p_M$ is $\varrho(r)=Z^{-1}c_\nu^\gamma(1+r^2/\nu)^{-\beta}$. Hence
\begin{equation*}
    \left\lvert(\log\varrho)'(r)\right\rvert
    =\frac{2\beta r}{\nu+r^2}
    \le\frac{\widetilde{C}_1}{1+r},
    \qquad
    \eta_M(r)\le C_\eta\left(\frac ZM\right)^{1/d}(1+r^2)^{\beta/d},
\end{equation*}
with $\widetilde{C}_1$ depending only on $\beta$ and $\nu$, and $C_\eta$ only on $d$, $\nu$ and $\beta$; the second bound uses $1+r^2/\nu\le\max\{1,1/\nu\}(1+r^2)$.
If $\lvert s-r\rvert\le(1+r)/2$, then $1+t\ge(1+r)/2$ between $r$ and $s$, and integrating the derivative bound gives $\lvert\log\varrho(s)-\log\varrho(r)\rvert\le2\widetilde{C}_1\lvert s-r\rvert/(1+r)\le \widetilde{C}_1$.
Condition \eqref{eq:radial-local-osc} therefore holds with $B=\widetilde{C}_1$ as soon as
\begin{equation}
\label{eq:student-relative-scale}
    \sup_{0\le r\le\sqrt d\rad_M}
    \frac{2\sqrt d\,\eta_M(r)}{1+r}
    \le\frac12 .
\end{equation}
The left-hand side is at most a constant multiple of $(Z/M)^{1/d}\max\{1,\rad_M^{2\beta/d-1}\}$, and this tends to zero in every regime.
If $\alpha<\nu+d/2$, then $Z=Z_\gamma$ and, by \eqref{eq:student-beta} and \eqref{eq:student-Rn}, $\rad_M^{2\beta/d-1}=M^{\kappa_R}$ with
\begin{equation*}
        \kappa_R=\frac{2\alpha(\nu+d/2-\alpha)}{d(\nu+d/2)(2\alpha+d)}
    <\frac1d,
\end{equation*}
because $2\alpha(\nu+d/2-\alpha)<2\alpha(\nu+d/2)<(\nu+d/2)(2\alpha+d)$.
If $\alpha\ge\nu+d/2$, then $2\beta\le d$, so the maximum equals one, and $Z_M/M\to0$ as $M\to\infty$ by \eqref{eq:student-normalization}, since $d-2\beta<d$.
Choosing $M_0$ such that \eqref{eq:student-relative-scale} and $P_M(\Kregion_M)\ge1/2$ hold for $M\ge M_0$, Lemma~\ref{lem:radial-check} with $r_{\max}=\sqrt d\rad_M$ gives \hypref{hyp:admissible-design}\textbf{(U)} with $D=2^de^{2\widetilde{C}_1}$, independent of $M$, and \eqref{eq:RH} with $C_{\mathrm{RH}}=e^{2\widetilde{C}_1}$. Lemma~\ref{lem:J-finite} gives $\Jcrit_M(p_M)\le(C_{\mathrm{RH}}Z)^{2/\gamma}$.
 
For the tail, $\Kregion_M^c\subseteq \EuclideanBall(\zero_d,\rad_M)^c$ gives, for $\rad_M\ge1$,
\begin{equation}
\label{eq:student-tail}
    \left(\int_{\Kregion_M^c}\pi(x)^2\,\dd x\right)^{1/2}
    \le C \rad_M^{-(\nu+d/2)} .
\end{equation}
In case (a), \eqref{eq:student-Rn} makes \eqref{eq:student-tail} equal to $\mathcal{O}(M^{-\alpha/d})$, which is \hypref{hyp:tail}, and Lemma~\ref{lem:J-limit} gives \hypref{hyp:cell-average}.
 
In case (b), fix $M\ge M_0$ and apply Proposition~\ref{prop:cell-bound} to the density $p_M$ at level $M$, together with \eqref{eq:error-identity} and the first two inequalities of \eqref{eq:tail-bound-proof}. For every design filling $\StopPartition_M(\RootFamily_M,p_M\,\dd x)$, this bounds the worst-case error by
\begin{equation}
\label{eq:student-finite-bound}
    c_{\mathrm{loc}}M^{-\alpha/d}\Jcrit_M(p_M)^{1/2}
    +C_{\mathrm{eq}}\left(\int_{\Kregion_M^c}\pi^2\,\dd x\right)^{1/2}
    \le
    C M^{-\alpha/d}Z_M^{1/\gamma}
    +C \rad_M^{-(\nu+d/2)} .
\end{equation}
With $\rad_M=M^{1/d}$, the second term is $CM^{-(\nu+d/2)/d}$.
If $\alpha=\nu+d/2$, then $Z_M\le C_Z\log \rad_M\le C_Z\log M$ and $\alpha/d=(\nu+d/2)/d$, so the first term is $\mathcal{O}(M^{-(\nu+d/2)/d}(\log M)^{1/\gamma})$.
If $\alpha>\nu+d/2$, then $Z_M^{1/\gamma}\le C_Z^{1/\gamma}M^{(d-2\beta)/(d\gamma)}$, and $(d-2\beta)/\gamma=(2\alpha+d)/2-(\nu+d)=\alpha-\nu-d/2$, so that the first term is at most a constant multiple of $M^{-(\nu+d/2)/d}$.
For the deterministic design, take $M=M_n$: as in Appendix~\ref{app:thm}, one node in each depth-$\ell_0$ descendant of each cell uses at most $m_0DM_n\le n$ nodes, and $M_n\ge n/\left(2Dm_0\right)$ for large $n$.
For independent sampling from $p_{\widetilde M_n}$, choose $A_\nu\ge(c+2)/a$ and take $M=\widetilde M_n$: Lemma~\ref{lem:min_occupancy}, applied as in \eqref{eq:fill-prob}, shows that the sample fills $\StopPartition_{\widetilde M_n}$ with probability at least $1-n^{-c}$, and $\widetilde M_n\ge n/(2A_{\nu}\log n)$ and $(\log\widetilde M_n)^{1/\gamma}\le(\log n)^{1/\gamma}$ for large $n$.
\end{proof}

\subsection{Proofs of the lower bounds}

The lower bounds use no randomness and no property of the kernel beyond \hypref{hyp:sobolev}.
For an arbitrary rule with $n$ nodes, we construct a function in the unit ball of $\Ha$ that vanishes at every node but has a large integral.
Such a function is a sum of smooth bumps placed in node-free cubes.
Everything follows from one estimate: if $\pi$ is bounded below by $m$ on a cube of side $L$, splitting that cube into about $2n$ subcubes leaves $n$ of them free of nodes, and the bumps placed there give an error of order $m\,L^{\alpha+d/2}\,n^{-\alpha/d}$.
The two results then differ only in the choice of the cube: a fixed cube for Corollary~\ref{cor:lower-local}, and a cube far in the tail for Theorem~\ref{thm:lower-student}.
 
The bumps are rescaled copies of a single reference function. For an open set $\OpenSet\subseteq\mathbb{R}^d$, let $\SmoothClass_c^\infty(\OpenSet)$ denote the set of infinitely differentiable functions $\phi:\mathbb{R}^d\to\mathbb{R}$ whose support is a compact subset of $\OpenSet$. 
Throughout, fix $\psi_0\in \SmoothClass_c^\infty((0,1)^d)$ with $\psi_0\ge0$ and $\int_{(0,1)^d}\psi_0(x)\,\dd x=1$. Each of the three properties is used below. Smoothness ensures that $\kappa_{\psi_0}=\lVert\psi_0\rVert_{\mathcal{H}^\alpha(\mathbb{R}^d)}$ is finite, for every $\alpha$. Compact support in the open cube ensures that a rescaled copy of $\psi_0$ placed in a subcube vanishes on the boundary of that subcube, so that bumps in different subcubes have disjoint supports and vanish at nodes. Finally,  $\psi_0\ge0$ and $\int_{(0,1)^d}\psi_0(x)\,\dd x=1$ give an explicit lower bound on the integral of each bump.
\begin{lemma}
\label{lem:lower-bump}
Let $\alpha>d/2$ be an integer, let $n\in\mathbb{N}$, and let $\BumpCube\subseteq\Xset$ be a cube of side $L\le(2n)^{1/d}$ such that $\pi\ge m>0$ almost everywhere on $\BumpCube$. Then, for every $\QuadRule_n\in\QuadClass_n$,
\begin{equation}
\label{eq:lower-bump}
    \sup_{\lVert f\rVert_{\Ha}\le1}
    \left\lvert\Pi(f)-\QuadRule_n(f)\right\rvert
    \ge
    c_{\psi_0}\,m\,L^{\alpha+d/2}\,n^{-\alpha/d},
    \qquad
    c_{\psi_0}=\frac{(2\cdot2^{1/d})^{-(\alpha+d/2)}}{\kappa_{\psi_0}}.
\end{equation} 
\end{lemma}
\begin{proof}
Let $m_{\mathrm{sub}}=\lceil(2n)^{1/d}\rceil$, so that $2n\le m_{\mathrm{sub}}^d$ and $m_{\mathrm{sub}}\le2(2n)^{1/d}$, and split $\BumpCube$ into $m_{\mathrm{sub}}^d$ subcubes of side $h=L/m_{\mathrm{sub}}$. Then
\begin{equation*}
    \frac{L}{2(2n)^{1/d}}\le h\le\frac{L}{(2n)^{1/d}}\le1 .
\end{equation*}
For a subcube $\Subcube=a+[0,h]^d$, let $\psi_\Subcube(x)=\psi_0((x-a)/h)$, which is supported in the interior of $\Subcube$.

For every multi-index $\beta$ with $\lVert\beta\rVert_1\le\alpha$, the change of variables $x=a+ht$ gives
\begin{equation*}
    \lVert D^\beta\psi_\Subcube\rVert_{L^2(\mathbb{R}^d)}^2=h^{d-2\lVert\beta\rVert_1}\lVert D^\beta\psi_0\rVert_{L^2(\mathbb{R}^d)}^2\le h^{d-2\alpha}\lVert D^\beta\psi_0\rVert_{L^2(\mathbb{R}^d)}^2,
\end{equation*} 
since $h\le1$. Summing over $\beta$ , and using the same change of variables for the integral,
\begin{equation*}
    \lVert\psi_\Subcube\rVert_{\Ha}\le\kappa_{\psi_0}\,h^{d/2-\alpha},
    \qquad
    \int_{\Subcube}\psi_{\Subcube}(x)\,\dd x=h^d.
\end{equation*} 

The subcubes have pairwise disjoint interiors, so each node lies in the interior of at most one of them. Since the rule has at most $n$ nodes, at most $n$ subcubes contain a node in their interior, and, since there are $m_{\mathrm{sub}}^d\ge2n$ subcubes, at least   $m_{\mathrm{sub}}^d-n\ge n$ of them contain no node in their interior. Choose $n$ such subcubes $\Subcube_1,\ldots,\Subcube_n$, and let $g=\sum_{j=1}^n\psi_{\Subcube_j}$. Then $g$ vanishes at every node, so $\QuadRule_n(g)=0$. The bumps have disjoint supports, hence
\begin{equation*}
    \lVert g\rVert_{\Ha}^2=\sum_{j=1}^n\lVert\psi_{\Subcube_j}\rVert_{\Ha}^2\le\kappa_{\psi_0}^2\,n\,h^{d-2\alpha},
    \qquad
    \Pi(g)\ge m\sum_{j=1}^n\int_{\Subcube_j}\psi_{\Subcube_j}(x)\,\dd x=m\,n\,h^d .
\end{equation*} 

The function $f=g/\lVert g\rVert_{\Ha}$ lies in the unit ball, and
\begin{equation*}
\left\lvert\Pi(f)-\QuadRule_n(f)\right\rvert
    =\frac{\Pi(g)}{\lVert g\rVert_{\Ha}}
    \ge\frac{m}{\kappa_{\psi_0}}\,n^{1/2}h^{\alpha+d/2}
    \ge\frac{m}{\kappa_{\psi_0}}\,n^{1/2}\left(\frac{L}{2(2n)^{1/d}}\right)^{\alpha+d/2}
    =c_{\psi_0}\,m\,L^{\alpha+d/2}\,n^{-\alpha/d}.
\end{equation*} 
\end{proof}
Lemma~\ref{lem:lower-bump} with $\BumpCube=[0,1]^d$ and $m=1$ also recovers the lower bound in Theorem~\ref{thm:novak-lower}. For a general density, it suffices to choose the cube.

\vspace{\baselineskip}
\begin{proof}[\textbf{Corollary~\ref{cor:lower-local}}]
Let $\BumpCube\subseteq \LocalCube$ be a cube of side $L=\min\{\ell_\LocalCube,1\}$, where $\ell_\LocalCube$ is the side length of $\LocalCube$. Then $L\le1\le(2n)^{1/d}$ for every $n\in\mathbb{N}$, and $\pi\ge m$ almost everywhere on $\BumpCube$. Lemma~\ref{lem:lower-bump} gives, for every rule $\QuadRule_n\in\QuadClass_n$, a worst-case error of at least   $c_{\psi_0}\,m\,L^{\alpha+d/2}n^{-\alpha/d}$. The constant does not depend on the rule, and taking the infimum over $\QuadClass_n$ proves \eqref{eq:lower-local} with   $c=c_{\psi_0}\,m\,L^{\alpha+d/2}$.
\end{proof}

For the Student target, the density is small far from the origin, but large cubes fit there. Placing the cube at distance $s$ from the origin trades the density $s^{-(\nu+d)}$ against the factor $L^{\alpha+d/2}$ with $L$ proportional to $s$, and gives an error of order $s^{\alpha-\nu-d/2}n^{-\alpha/d}$. When $\alpha\le\nu+d/2$, the best choice is $s=1$. When $\alpha>\nu+d/2$, the cube should be pushed as far as the constraint $L\le(2n)^{1/d}$ allows, which yields the tail-limited rate.

\vspace{\baselineskip}
\begin{proof}[\textbf{Theorem~\ref{thm:lower-student}}]
For $s\ge1$, let $\mathbf{e}_1$ be the first coordinate vector and
\begin{equation*}
    \BumpCube_s=\tfrac32s\,\mathbf{e}_1+\left[-\tfrac{s}{2\sqrt d},\tfrac{s}{2\sqrt d}\right]^d,
\end{equation*} 
a cube of side $L_s=s/\sqrt d$. Every $x\in \BumpCube_s$ satisfies $\lVert x-\tfrac32s\,\mathbf{e}_1\rVert_2\le s/2$, so $s\le\lVert x\rVert_2\le2s$ and $1+\lVert x\rVert_2^2\le5s^2$. By \eqref{eq:student-density-comparison}, for every $x\in\BumpCube_s$,
\begin{equation*}
    \pi_\nu(x)\ge m_s=c_1\,5^{-(\nu+d)/2}s^{-(\nu+d)}
    .
\end{equation*} 
If $L_s\le(2n)^{1/d}$, Lemma~\ref{lem:lower-bump} gives, for every $\QuadRule_n\in\QuadClass_n$,
\begin{equation*}
    \sup_{\lVert f\rVert_{\Ha}\le1}\left\lvert\Pi(f)-\QuadRule_n(f)\right\rvert
    \ge c_{\psi_0}\,m_s\,L_s^{\alpha+d/2}\,n^{-\alpha/d}
    =c_2\,s^{\alpha-\nu-d/2}\,n^{-\alpha/d},
\end{equation*} 
where $c_2=c_{\psi_0}\,c_1\,5^{-(\nu+d)/2}d^{-(\alpha+d/2)/2}$. If $\alpha\le\nu+d/2$, choose $s=1$; then $L_s\le1\le(2n)^{1/d}$, and the bound is $c_2\,n^{-\alpha/d}$.
If $\alpha>\nu+d/2$, choose $s=n^{1/d}\ge1$; then $L_s\le n^{1/d}\le(2n)^{1/d}$, and the bound is $c_2\,n^{(\alpha-\nu-d/2)/d}n^{-\alpha/d}=c_2\,n^{-(\nu+d/2)/d}$.
In both cases the bound is $c_2\,n^{-\min(\alpha,\nu+d/2)/d}$, uniformly over the rule, and taking the infimum over $\QuadClass_n$ proves \eqref{eq:lower-student} with $c=c_2$.
\end{proof}

\section{Kernel Hyperparameter Estimation}
\label{app:hyper}

Let $\boldsymbol{\theta}=(\sigma_f^2,\ell_f)\in(0,\infty)^2$ denote the kernel hyperparameters. Given observations $\f_{1:n}$ at design points $\bx_{1:n}$, we assume the hierarchical model
\begin{equation*}
f^\star \sim \mathcal{GP}(0,k_{\boldsymbol{\theta}}),
\qquad
\sigma_f^2 \sim \mathrm{IG}(2,2),
\qquad
\ell_f\sim\operatorname{LogNormal}(0,100^2),
\end{equation*} 
where the log-normal parameters are the mean and variance on the log scale, and the inverse-gamma distribution is parameterised by its shape and scale.

Conditionally on ${\boldsymbol{\theta}}$, the observations satisfy
\begin{equation*}
\f_{1:n}\mid{\boldsymbol{\theta}}
\sim
\mathcal{N}(\zero_n,\bK_{\boldsymbol{\theta}}),
\qquad
\bK_{\boldsymbol{\theta}}=[k_{\boldsymbol{\theta}}(x_i,x_j)]_{i,j=1}^n.
\end{equation*}
Since the Matérn kernels are linear in $\sigma_f^2$, by conjugacy, the conditional posterior of $\sigma_f^2$ is inverse-gamma:
\begin{equation}
\label{eq:ig_full}
\sigma_f^2\mid \ell_f,\f_{1:n}
\sim
\mathrm{IG}
\left(
2+\frac n2,\;
2+\frac12\f_{1:n}^\top \bK_{\ell_f}^{-1}\f_{1:n}
\right),
\qquad \bK_{\ell_f}=\left[k_{\boldsymbol{\theta}}(x_i,x_j)/\sigma_f^2\right]_{i,j=1}^n.
\end{equation}

\begin{algorithm2e}[ht]
\caption{Metropolis-within-Gibbs sampler for kernel hyperparameters}
\label{alg:mwg}

\KwIn{Pilot design $\bx_{1:n}$ and evaluations $\f_{1:n}$, initial hyperparameter values $(\sigma_{f,0}^{2},\ell_{f,0})$, number of iterations $T$}

\KwOut{$\{(\sigma_{f,t}^{2},\ell_{f,t})\}_{t=1}^T$}

\For{$t = 1$ \KwTo $T$}{
\textbf{set} $u_{t-1}=\log \ell_{f, t-1}$\;
\textbf{propose} $u^\star=u_{t-1}+\varepsilon$ with $\varepsilon\sim\mathcal{N}\left(0,s_{\ell_f}^2\right)$,
\textbf{set} $\ell_f^\star=\exp(u^\star)$\;
\textbf{compute}
\begin{equation*}
\rho_{\rm MH}=
\frac{
p\left(\ell_f^\star\mid\sigma_{f,t-1}^{2},\f_{1:n}\right)
\,\ell_f^\star
}{
p\left(\ell_{f,t-1}\mid\sigma_{f, t-1}^{2},\f_{1:n}\right)
\,\ell_{f,t-1}
};\;
\end{equation*}
\textbf{set} $\ell_{f,t} = \ell_f^\star$, with probability $\min(1, \rho_{\rm MH})$, and $\ell_{f,t} = \ell_{f,t-1}$, otherwise\;
\textbf{sample}
\begin{equation*}
\sigma_{f,t}^{2}
\sim
\mathrm{IG}
\left(
2+\frac n2,
2+\frac12
\f_{1:n}^\top
\bK_{\ell_{f,t}}^{-1}
\f_{1:n}
\right);\;
\end{equation*}
}
\end{algorithm2e}

The conditional posterior of $\ell_f$ is available up to normalization as
\begin{equation}
\label{eq:length_full}
p(\ell_f\mid\sigma_f^2,\f_{1:n})
\propto
p(\ell_f)\,
\lvert \bK_{\ell_f}\rvert^{-1/2}
\exp\left(
-\frac{1}{2\sigma_f^2}
\f_{1:n}^\top
\bK_{\ell_f}^{-1}
\f_{1:n}
\right),
\end{equation}
and is sampled using a random-walk Metropolis--Hastings update on the log scale to ensure positivity of $\ell_f$ and improve numerical stability.

The resulting Metropolis-within-Gibbs sampler is summarized in Algorithm~\ref{alg:mwg}. Note that, since the proposal is made on the transformed scale, the Metropolis--Hastings ratio must include a Jacobian adjustment for the change of variables. In practice, we set the proposal standard deviation to $s_{\ell_f}=0.2$. Moreover, a Cholesky decomposition of $\bK_{\ell_f}$ is used to efficiently compute both $\bK_{\ell_f}^{-1} \textbf{f}_{1:n}$ and the log-determinant required in the posterior computation. A small nugget $\eta \id_n$, $\eta=10^{-8}$ is added to $\bK_{\ell_f}$ for numerical stability.

Under standard irreducibility conditions and since $p(\ell_f)$ is proper, the chain is Harris recurrent and ergodic with invariant distribution $p({\boldsymbol{\theta}}\mid\f_{1:n})$ \citep{RobertsRosenthal2006Harris}.

\section{Further Numerical Experiment}
\label{app:num}

We consider the heavy-tailed target $\Pi=t_8\left(0,1\right)$ and the integrand in \eqref{eq:f-gauss}, under the experimental protocol of Section~\ref{sec:num}, with the Matérn $3/2$ and $7/2$ kernels, using $N=100$ independent designs instead of $1000$.
In dimension one, $\pi^\gamma$ is integrable whenever $\nu>1/\gamma-1$, hence whenever $\nu>7/2$ for both kernels considered here.
The choice $\nu=8$ therefore permits the untruncated proposal of Theorem~\ref{thm:up-int}, and the corresponding location--scale Student laws are
\begin{equation*}
    P^\star =
    \begin{cases}
        t_{13/5}(0,40/13),&
        \text{Matérn }3/2,\quad \gamma=2/5,\\[2pt]
        t_{1}(0,8),&
        \text{Matérn }7/2,\quad \gamma=2/9.
    \end{cases}
\end{equation*}
As for the Beta target, the kernel means and the prior variance are computed numerically with relative tolerance $10^{-11}$.
With the fitted hyperparameters, these prior integral variances are $1.04\times10^{-1}$ for the Matérn $3/2$ kernel and $1.00\times10^{-1}$ for the Matérn $7/2$ kernel.

Figure~\ref{fig:student} shows that sampling from $p^\star\propto\pi^\gamma$ gives the smallest error in the single runs and the smallest average posterior variance across replications for both kernels.
For the Matérn $3/2$ single run, at $n=250$ the absolute errors under $\pi$, $\widehat{q}_\lambda^{\,\,\star}$ and $\pi^\gamma$ are $2.93\times10^{-4}$, $1.90\times10^{-4}$ and $2.18\times10^{-5}$, respectively, while the corresponding posterior variances are $1.47\times10^{-5}$, $1.18\times10^{-6}$ and $3.70\times10^{-7}$.
For the Matérn $7/2$ single run, the absolute errors are $9.38\times10^{-4}$, $2.02\times10^{-5}$ and $1.81\times10^{-6}$, and the posterior variances are $1.21\times10^{-5}$, $1.46\times10^{-7}$ and $1.23\times10^{-8}$, in the same order.
The averaged posterior means in the centre panel confirm that this behaviour is stable across the $100$ independent designs.

\begin{figure}[!t]
    \centering
    \begin{subfigure}{0.32\textwidth}
        \centering
        \includegraphics[trim={0.2cm 0.2cm 0 0}, clip, width=\textwidth]{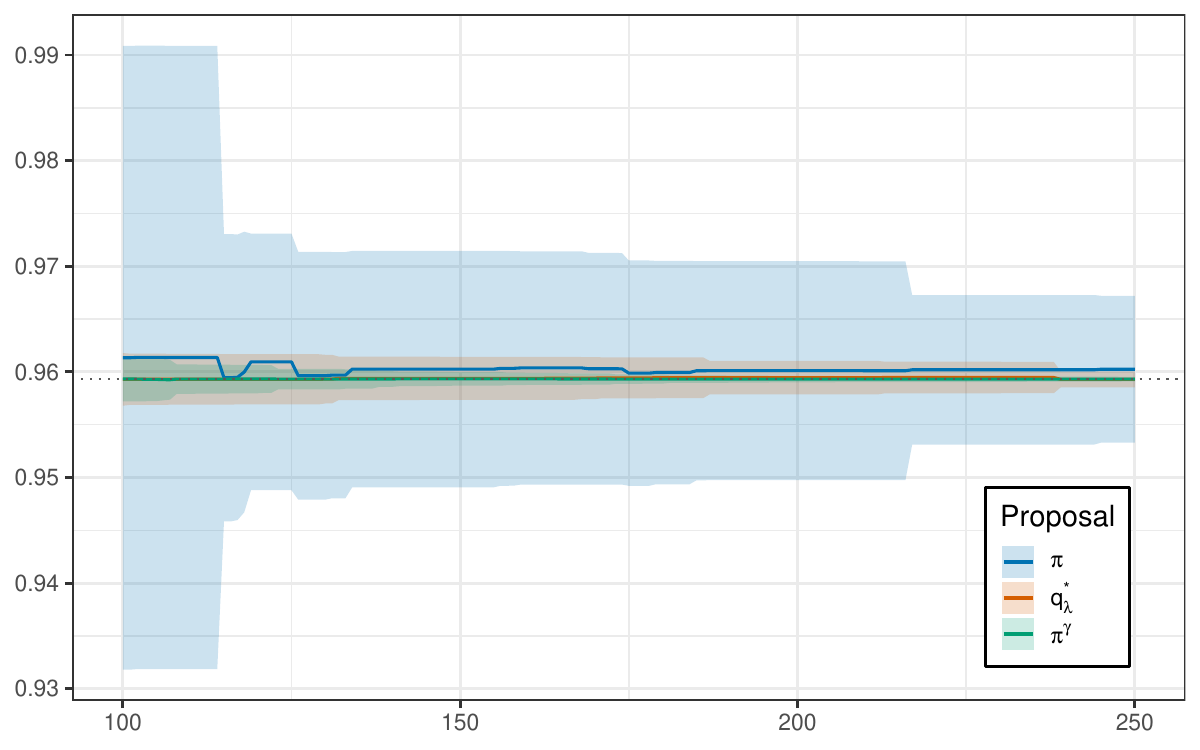}
    \end{subfigure}
    \hfill
    \begin{subfigure}{0.32\textwidth}
        \centering
        \includegraphics[trim={0.2cm 0.2cm 0 0}, clip, width=\textwidth]{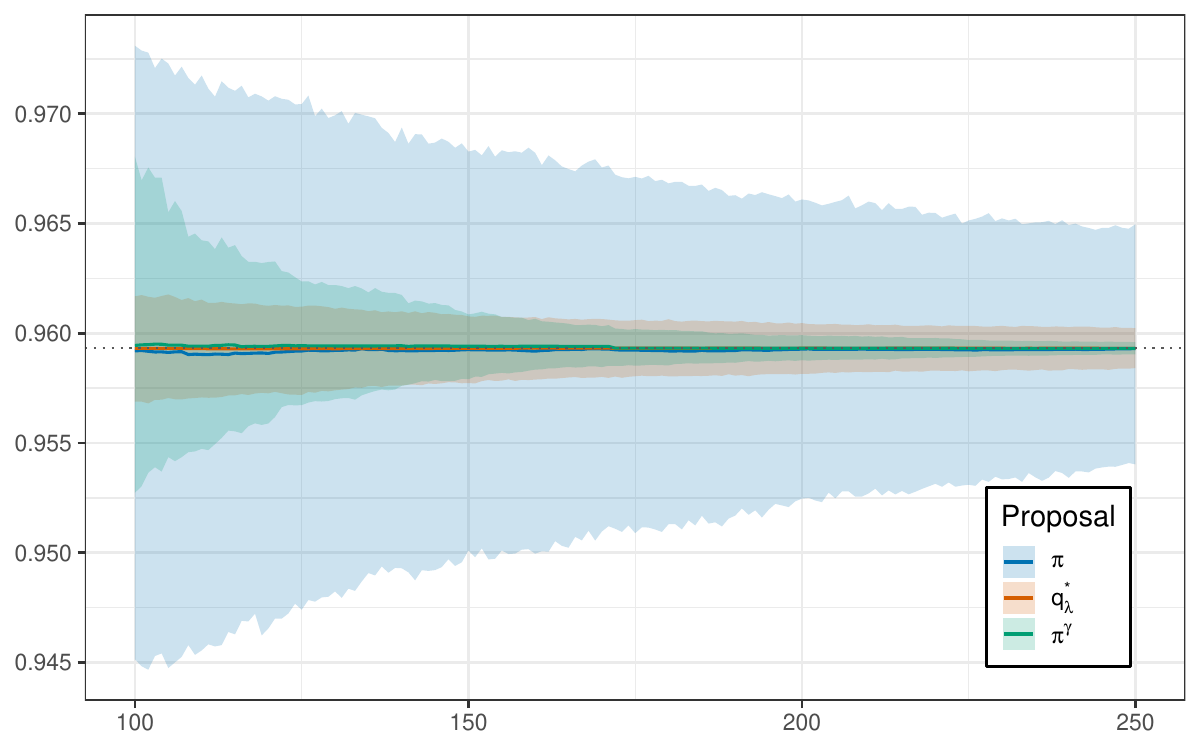}
    \end{subfigure}
    \hfill
    \begin{subfigure}{0.32\textwidth}
        \centering
        \includegraphics[trim={0.2cm 0.2cm 0 0}, clip, width=\textwidth]{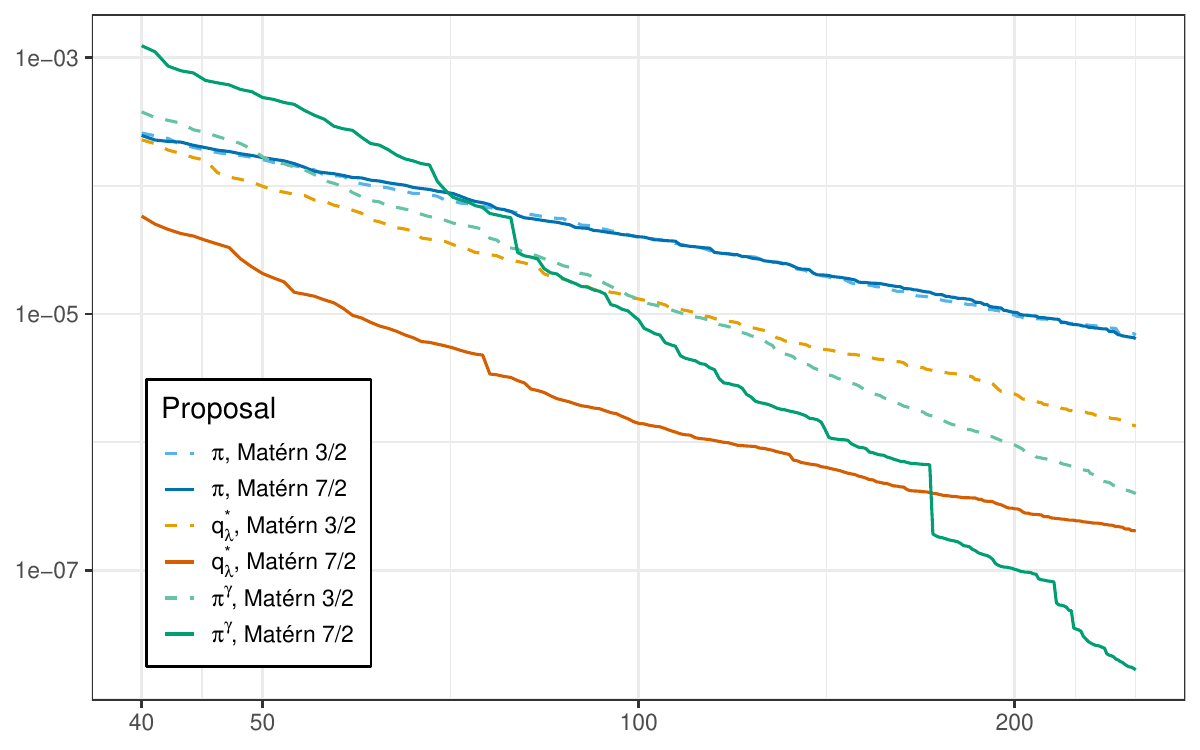}
    \end{subfigure}
    \caption{Posterior summaries for $f$ as in \eqref{eq:f-gauss} under the Student $t_8$ target and proposals $\pi$, $\widehat{q}_\lambda^{\,\,\star}$ and $p^\star\propto\pi^\gamma$. 
    Left: single-run posterior mean (solid line) $\pm 2$ standard deviation bands for the Matérn $7/2$ kernel against the number $n$ of evaluations; the horizontal dashed line is the true integral value. 
    Centre: marginal posterior mean and $95\%$ bands for the Matérn $7/2$ kernel, estimated from $100$ independent runs.
    Right: average posterior variance across the same $100$ runs, for the Matérn $3/2$ (dashed) and Matérn $7/2$ (solid) kernels (log--log scale, from $n=40$).
    }
    \label{fig:student}
\end{figure}

The right panel also shows that the benefit becomes more pronounced for smoother kernels.
Under $\pi^\gamma$, the final average conditional posterior variance decreases from $3.97\times10^{-7}$ for Matérn $3/2$ to $1.67\times10^{-8}$ for Matérn $7/2$.
The corresponding values for $\widehat{q}_\lambda^{\,\,\star}$ are $1.33\times10^{-6}$ and $2.02\times10^{-7}$, while those for $\pi$ are $6.85\times10^{-6}$ and $6.42\times10^{-6}$.
The proposal $\pi^\gamma$ does not give the smallest posterior variance in the early iterations, but its variance then decays much faster.
This effect is more pronounced for the Matérn $7/2$ kernel. The proposal $p^\star$ has heavier tails and places nodes far into the tails, so the posterior variance contracts slowly at first. The corresponding variance curve is also more irregular, because the sampled locations vary more widely across the domain.

\vskip 0.2in
\bibliography{bibliography}

\end{document}